\documentclass[11pt]{article}

\usepackage{fullpage,amsmath,amssymb,url,hyperref,rotating}
\usepackage{xcolor,graphicx}
 
\def\sphere{\mathrm{sphere}}
\def\sign{\mathrm{sign}}
\def\inner#1#2{{\langle #1,#2 \rangle}}

\def\dt{\mathrm{d}t}
\def\diag{\mathrm{diag}}
\def\Mult{\mathrm{Mult}}
\def\JS{\mathrm{JS}}
\def\JB{\mathrm{JB}}
\def\IS{\mathrm{IS}}
\def\KL{\mathrm{KL}}
\def\Bhat{\mathrm{Bhat}}
\def\dmu{\mathrm{d}\mu}
\def\Pr{\mathrm{Pr}}

\def\Euc{\mathrm{Euc}}
\def\Burg{\mathrm{Burg}}
\def\Shannon{\mathrm{Shannon}}

\def\st{\ :\ }

\def\vectwo#1#2{\left[\begin{array}{l}#1 \cr #2 \end{array}\right]}
\def\dettwotwo#1#2#3#4{\left|\begin{array}{ll}#1 & #2\cr #3 & #4\end{array}\right|}
\def\sqr{\mathrm{sqr}}

\def\eqdef{:=}
\def\bbR{\mathbb{R}}

\def\vectwo#1#2{\left[\begin{array}{l}#1 \cr #2 \end{array}\right]}
\def\dettwotwo#1#2#3#4{\left|\begin{array}{ll}#1 & #2\cr #3 & #4\end{array}\right|}
\def\QED{\ensuremath{{\square}}}
\def\markatright#1{\leavevmode\unskip\nobreak\quad\hspace*{\fill}{#1}}
\newenvironment{proof}
 {\begin{trivlist}\item[\hskip\labelsep{\bf Proof.}]}
 {\markatright{\QED}\end{trivlist}}
\def\primal{\textcolor{red}{$\bullet$}}
\def\dual{\textcolor{blue}{$\bullet$}}

\newtheorem{proposition}{Proposition}
\newtheorem{theorem}{Theorem}
\newtheorem{property}{Property}
 
\title{On geodesic triangles with right angles in a dually flat space\footnote{This work was published in~\cite{BregmanManifold-PIGTA-2021}. This report further contains a section on Bregman balls and Bregman spheres in \S\ref{sec:BBall}.}}

\author{Frank Nielsen\footnote{E-mail: {\tt Frank.Nielsen@acm.org}. Web: \url{https://FrankNielsen.github.io/}}\\ Sony Computer Science Laboratories, Inc.\\ Tokyo, Japan}

\date{}

\begin{document}
\maketitle

\begin{abstract}
The dualistic structure of statistical manifolds in information geometry yields  eight types of geodesic triangles passing through three given points, the triangle vertices.
The interior angles of geodesic triangles can sum up to $\pi$ like in Euclidean/Mahalanobis flat geometry, or exhibit otherwise angle excesses or angle defects.
In this paper, we initiate the study of geodesic triangles in dually flat spaces, termed Bregman manifolds, where a generalized Pythagorean theorem holds. We consider non-self dual Bregman manifolds  since Mahalanobis self-dual manifolds amount to Euclidean geometry.
First, we show how to construct geodesic triangles with either one, two, or three interior right angles, whenever it is possible.
Second, we report a construction of triples of points for which the dual Pythagorean theorems hold simultaneously at a point, yielding two dual pairs of dual-type geodesics with right angles at that point.
\end{abstract}


\section{Introduction and motivation}

In Euclidean geometry, it is well-known that the sum of the three interior angles of any triangle {\em always} sum up to $\pi$ 
(Figure~\ref{fig:triangleEuc}), and that  Euclidean triangles can have {\em at most} one right angle ($\frac{\pi}{2}$ radians or $90^{\circ}$).
These facts are not true anymore in hyperbolic geometry nor in spherical geometry~\cite{NonEucl-2012}, where the total sum of the interior angles of a triangle may vary~\cite{Wolf-1972}:

\begin{itemize}
\item In hyperbolic geometry, hyperbolic triangles have always   {\em angle defects}, 
meaning that the total sum of the three interior angles of any hyperbolic triangle is always strictly less than $\pi$.
In the extreme case of hyperbolic ideal triangles, the total sums of their interior angles vanish since the interior angle at an ideal vertex is always $0$. Figure~\ref{fig:curvatureTriangle} displays a right angle hyperbolic triangle and a hyperbolic ideal triangle.

\item In spherical geometry, spherical triangles have always {\em angle excesses}, 
meaning that the sum of interior angles of any spherical triangle is always strictly greater than $\pi$, and is provably upper bounded by $3\pi$ radians or $540$ degrees. 
Moreover, there exist spherical triangles with one, two, or three right angles, as depicted in
Figure~\ref{fig:curvatureTriangle}.
\end{itemize}

More generally, in Riemannian geometry~\cite{Rie-1997}, the angle excess or defect of a geodesic triangle reflects the total curvature enclosed by the geodesic triangle via the Gauss-Bonnet formula;
The general Gauss-Bonnet formula states that the integral of the scalar
curvature of a closed surface is $2\pi$ times the Euler characteristic $\chi$ of that surface (a topological characteristic).
The hyperbolic geometry, Euclidean geometry and  spherical geometry can be studied under the framework of Riemannian geometry,
as a manifold of constant negative curvature, a flat manifold, and a manifold of constant positive curvature, respectively.

\begin{figure}
\centering
\includegraphics[width=0.8\textwidth]{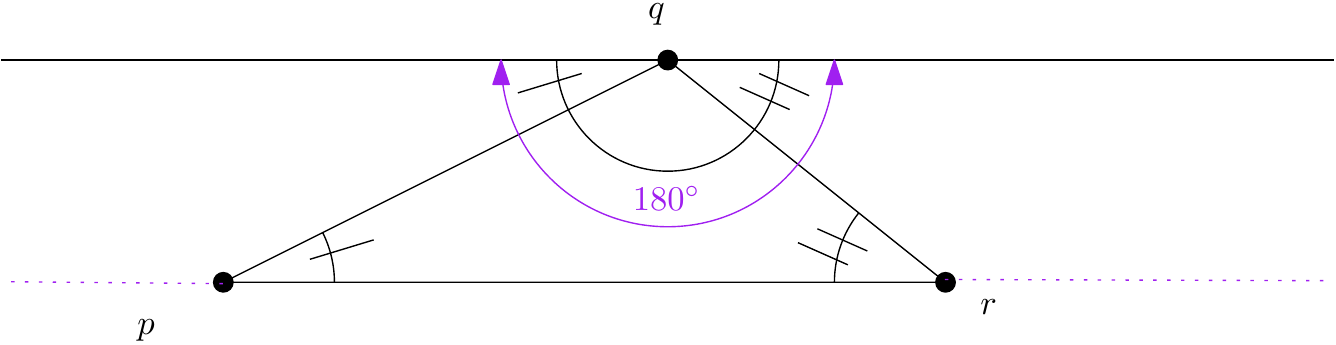}

\caption{Proof without words: The sum of the three interior angles of any triangle in Euclidean geometry always sum up to $\pi$ ($180$ degrees). Euclidean triangles dot not show angle defect nor angle excess: The Euclidean space is flat.
\label{fig:triangleEuc}}

\end{figure}

In this work, we consider the {\em dualistic structure} $(M,g,\nabla,\nabla^*)$ of information geometry~\cite{Eguchi-1992,IG-2016} which can be derived from the statistical manifold structure of Lauritzen~\cite{StatMfd-1987} $(M,g,C)$:
Namely, a simply connected smooth manifold $M$ is equipped with two dual torsion-free affine connections $\nabla$ and $\nabla^*$ such that these connections are coupled with the metric tensor $g$, meaning that their induced dual parallel transport preserves the metric~\cite{EIG-2018}.  We shall describe in details that dualistic structure in~\S\ref{sec:BM}.

Any distinct pair of points $p$ and $q$ on the manifold $M$ can either be joined by a primal $\nabla$-geodesic arc $\gamma_{pq}$ or by a dual $\nabla^*$-geodesic arc $\gamma_{pq}^*$. 
Thus any triple of points $(p,q,r)$ of $M$ can be connected pairwise using one of $2\times \binom{3}{2}=6$ primal/dual geodesic arcs. 
Let us parameterize\footnote{More precisely, a geodesic $\gamma_{pq}^\nabla(t)$  with respect to an affine connection $\nabla$ 
satisfies $\nabla_{\dot\gamma_{pq}} \dot\gamma_{pq}=0$. A $\nabla$-geodesic is an autoparallel curve at it is invariant by affine reparameterization of $t$ (i.e., $t'=at+b$).} these primal/dual geodesic arcs by $\gamma_{pq}(t)$ and  $\gamma_{pq}^*(t)$ 
so that $\gamma_{pq}(0)=\gamma_{pq}^*(0)=p$ 
and $\gamma_{pq}(1)=\gamma_{pq}^*(1)=q$.
Denote by $v_{pq}\eqdef \left.\frac{d}{\dt}\gamma_{pq}(t)\right|_{t=0} = \dot\gamma_{pq}(0)$ and
 $v_{pq}^*\eqdef \eqdef \left.\frac{d}{\dt}\gamma_{pq}^*(t)\right|_{t=0} = \dot\gamma_{pq}(0)$ the tangent vectors of the tangent plane $T_p$ to the primal and dual geodesics, respectively.
Two vectors $u,v\in T_p$ are orthogonal in $T_p$ (denoted notationally by $u\perp_p v$) iff. $g_p(u,v)=0$: 
\begin{equation}
u\perp_p v \Leftrightarrow g_p(u,v)=0.
\end{equation}
More generally, two smooth curves $c_1(t)$ and $c_2(t)$ on the manifold are said orthogonal at a point $p=c_1(t_1)=c_2(t_2)$ 
iff. $g_p(\dot c_1(t_1),\dot c_2(t_2))=0$.

\begin{figure}
\centering

\begin{tabular}{cc}
\includegraphics[width=0.25\textwidth]{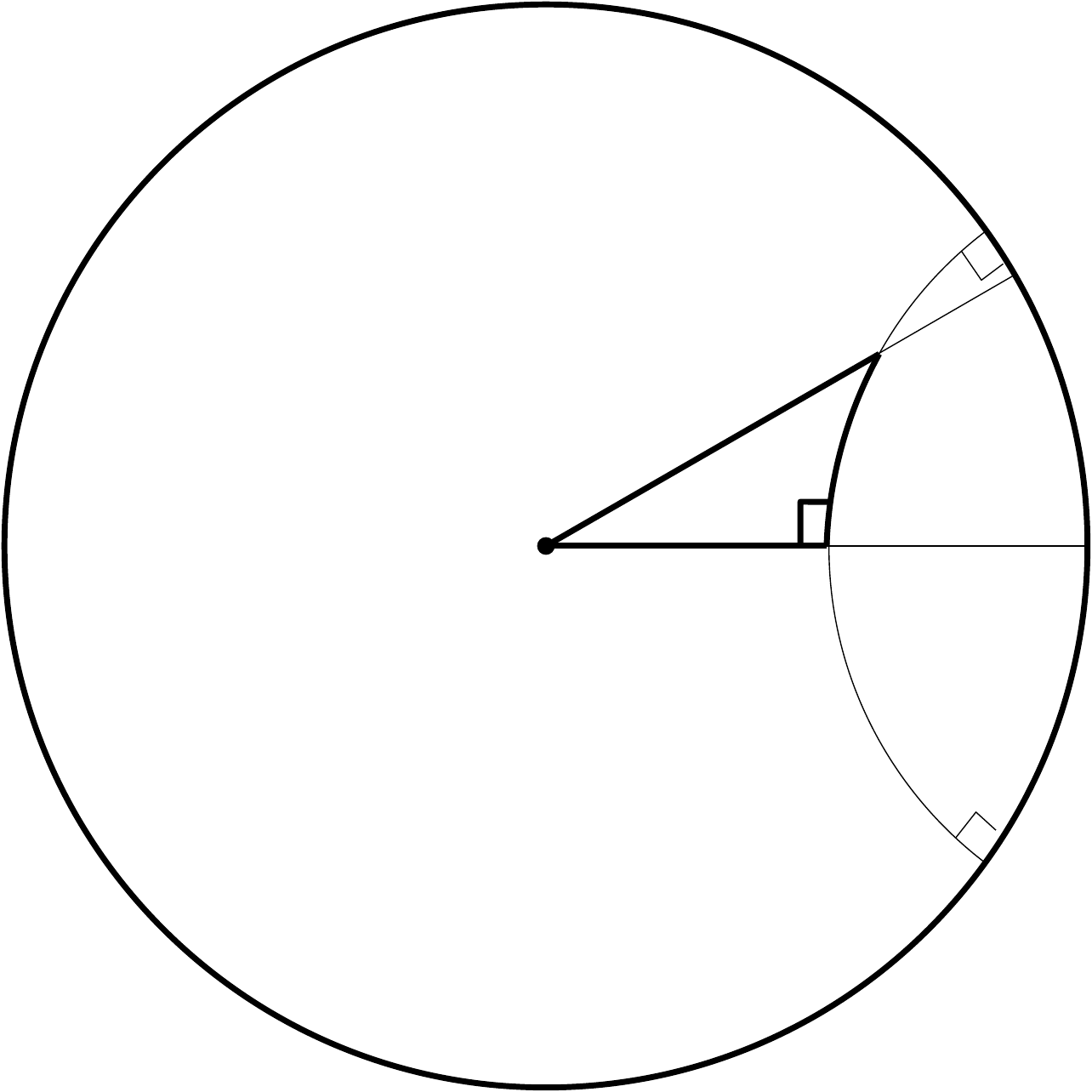}&
\includegraphics[width=0.25\textwidth]{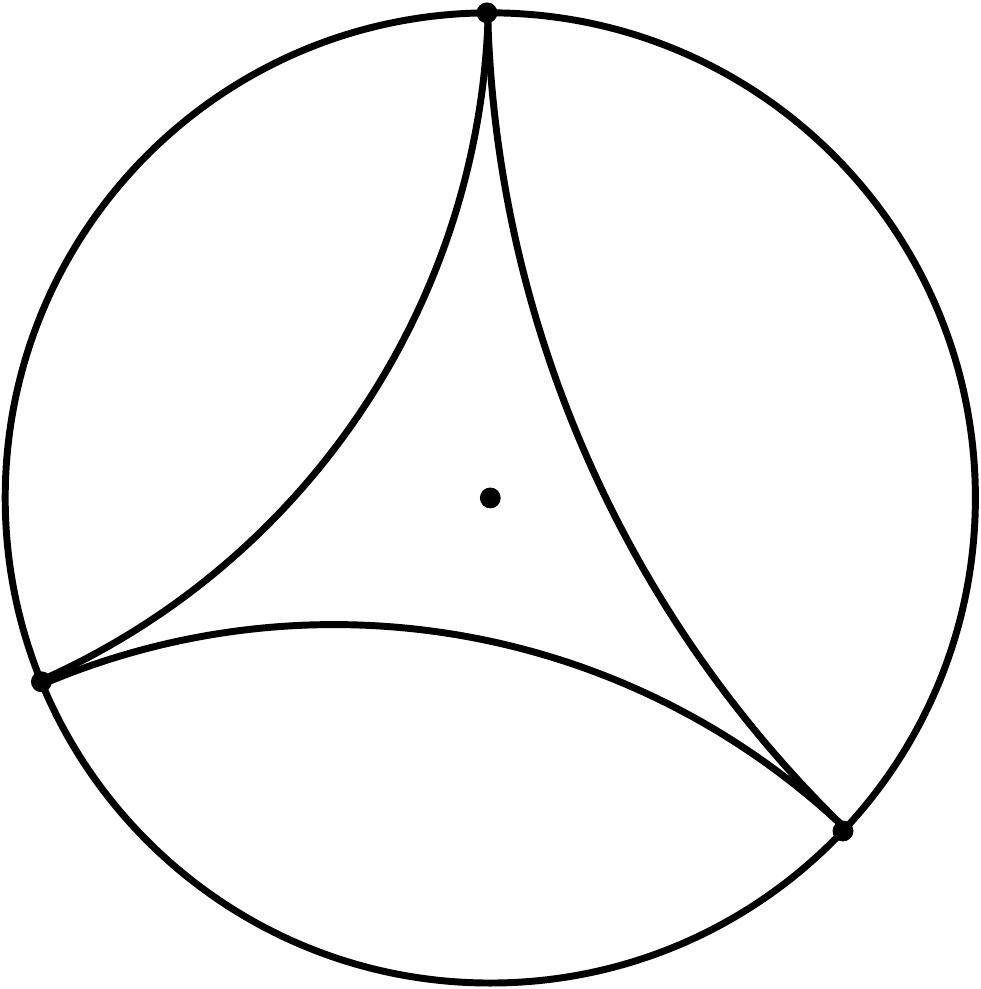}\\
(a) & (b)
\end{tabular}

\begin{tabular}{ccc}
\includegraphics[width=0.25\textwidth]{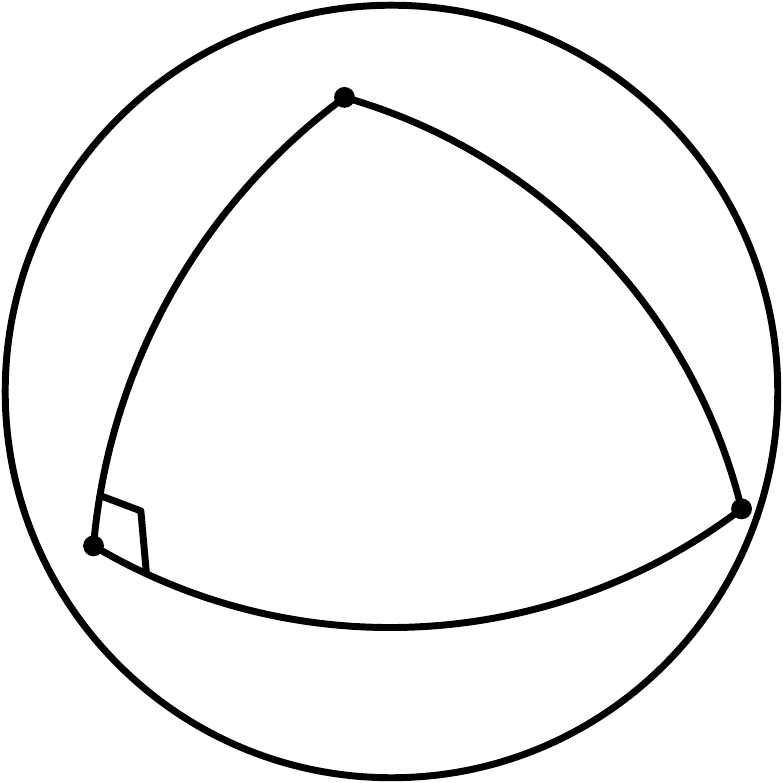}&
\includegraphics[width=0.25\textwidth]{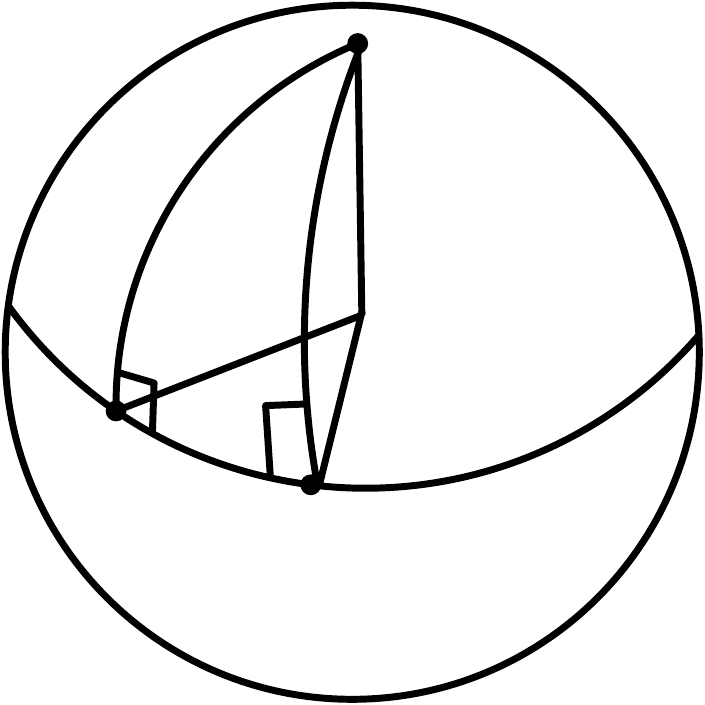}&
\includegraphics[width=0.25\textwidth]{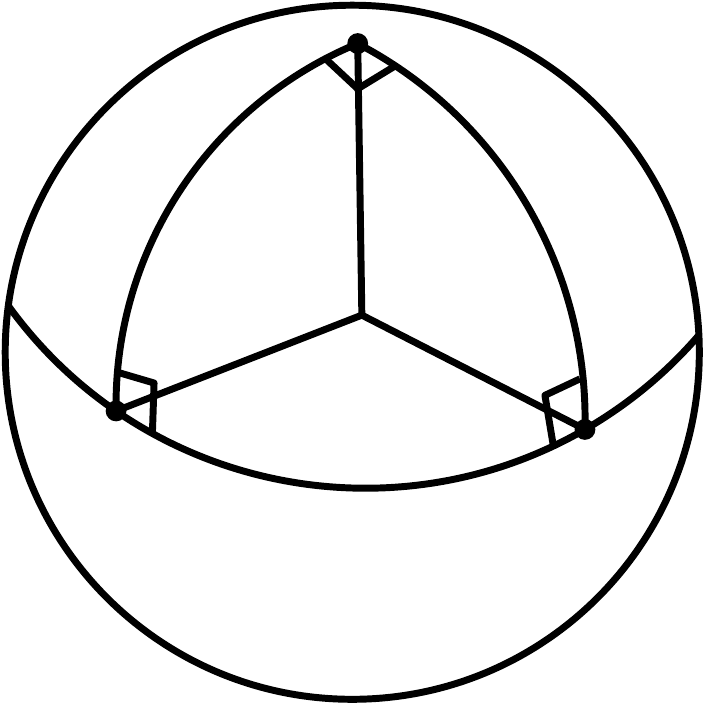}\\
(c) & (d) & (e)
\end{tabular}

 \caption{In hyperbolic geometry, triangles have angle defects: Visualization in the Poincar\'e conformal disk model.
(a) hyperbolic triangle with one right angle, (b) hyperbolic ideal triangle with all interior angles equal to zero (and area always $\pi$).
In spherical geometry, triangles have angle excesses: Visualization of spherical triangles on the unit 3D sphere.
(c) spherical triangle with one right angle, (d) spherical triangle with two right angles, and (e) spherical triangle with three right angles.
\label{fig:curvatureTriangle}
}
\end{figure}

A geodesic triangle $T$ passing through three points $p$, $q$, and $r$ (i.e., the triangle vertices) is a triangle with edges linking two vertices defined either by  primal or dual geodesic arcs.
Thus the dualistic structure of information geometry yields $2^3=8$ {\em types} of geodesic triangles passing through any three distinct given points $p$, $q$ and $r$.
Let ``p'' or \primal{} stands for {\em primal} and ``d'' or \dual{} stands for {\em dual}. 
A primal geodesic arc $\gamma_{pq}$ and a dual geodesic arc $\gamma_{pq}^*$  
may also be written $\gamma_{pq}^{\mbox{\primal{}}}$  
and $\gamma_{pq}^{\mbox{\dual{}}}$, respectively.
The $8$ types of geodesic triangles are:  
\primal\primal\primal{} ppp, 
\primal\primal\dual{} ppd, 
\primal\dual\dual{} pdd, 
\primal\dual\primal{} pdp, 
\dual\dual\dual{} ddd,  
\dual\dual\primal{} ddp, 
\dual\primal\primal{} dpp, and 
\dual\primal\dual{} dpd. 
We can group these eight types of geodesic triangles into two groups of four geodesic triangles each as shown in Table~\ref{tab:type}.

\begin{table}
\begin{center}
\begin{tabular}{|l|l|}\hline
Geodesic triangle $T$ & Dual geodesic triangle $T^*$\\ \hline
$T=\gamma_{pq}\gamma_{qr}\gamma_{rp}$ (type \primal\primal\primal{} ppp  : a $\nabla$-triangle) & $T^*=\gamma_{pq}^*\gamma_{qr}^*\gamma_{rp}^*$ (type \dual\dual\dual{} ddd: a $\nabla^*$-triangle)\\
$T=\gamma_{pq}\gamma_{qr}\gamma_{rp}^*$ (type \primal\primal\dual{} ppd) & $T^*=\gamma_{pq}^*\gamma_{qr}^*\gamma_{rp}$ (type \dual\dual\primal{} ddp)\\
$T=\gamma_{pq}\gamma_{qr}^*\gamma_{rp}^*$ (type \primal\dual\dual{} pdd) & $T^*=\gamma_{pq}^*\gamma_{qr}\gamma_{rp}$ (type \dual\primal\primal{}  dpp) \\
$T=\gamma_{pq}\gamma_{qr}^*\gamma_{rp}$ (type \primal\dual\primal{} pdp) & $T^*=\gamma_{pq}^*\gamma_{qr}\gamma_{rp}^*$ (type \dual\primal\dual{} dpd) \\ \hline
\end{tabular}
\end{center}
\caption{Eight types of geodesic triangles paired by duality.\label{tab:type}}
\end{table}

A geodesic triangle $T^*$ is dual to another geodesic triangle $T$ iff. each geodesic arc of $T^*$ is corresponding to a dual geodesic arc of $T$.
For example, the triangle  $T^*=\gamma_{pq}^*\gamma_{qr}\gamma_{rp}$ (type \dual\primal\primal{}  dpp)  is dual to the triangle $T=\gamma_{pq}\gamma_{qr}^*\gamma_{rp}^*$ (type \primal\dual\dual{} pdd).
At each triangle vertex, we have  $4$ geodesic arcs defining $\binom{4}{2}=6$ angles.
A triple of points defines $8$ geodesic triangles with $4\times 3=12$ interior angles.\footnote{We do consider at a triangle vertex only pairs of geodesics with interior angles linking the two other triangle vertices.}
The interior angle $\alpha_p(c_{pq},c_{pr})$ of a geodesic triangle $T=c_{pq}c_{qr}c_{rp}$ (where $c_{ab}$ stands  either  for $\gamma_{ab}$ or for $\gamma_{ab}^*$) at $p$ is measured according to the metric tensor $g$ as follows:
\begin{equation}\label{eq:angle}
\alpha_p(c_{pq},c_{pr}) = \arccos \left( \frac{g_p(\dot c_{pq}(0), \dot c_{pr}(0))}{\|\dot c_{pq}(0)\|_p \  \|\dot c_{pr}(0)\|_p} \right) 
=\alpha_p(c_{pr},c_{pq}),
\end{equation}
where $\|v\|_p=\sqrt{g_p(v,v)}$ measures the length of vector $v\in T_p$.
The total sum $\alpha(T)$ of the three interior angles of a  geodesic triangle $T=c_{pq}c_{qr}c_{rp}$ is defined by
\begin{equation}\label{eq:totalangle}
\alpha(T) \eqdef  \alpha_p(c_{pq},c_{pr}) + \alpha_q(c_{qp},c_{qr}) + \alpha_r(c_{rp},c_{rq}).
\end{equation}

We may specify a geodesic triangle as follows:
 $T_E(p,q,r)$ where $p,q,r\in M$  are the three triangle vertices (points on $M$), and $E\in(t_1\in\{p,d\},t_2\in\{p,d\},t_3\in\{p,d\})$ is the type of primal/dual geodesic edges of the triangle so that $c_{ab}^{t_i}=\gamma_{ab}$ if $t_i=p$ and $c_{ab}^{t_i}=\gamma_{ab}^*$ if $t_i=d$.
A geodesic triangle $T=T_{E}(p,q,r)$ with all primal geodesics arcs ($E=(p,p,p)$) is called a {\em $\nabla$-triangle}, and its dual geodesic triangle $T^*=T_{E^*}(p,q,r)$ (with all dual geodesic arcs, i.e., $E^*=(d,d,d)$) is called a {\em $\nabla^*$-triangle}. 

By considering duality of geodesic triangles and permutations of the triple of points $p$, $q$ and $r$, we may reduce the study of interior angles between a pair of geodesics at a vertex to three types of primal geodesic triangles (type ppp, pdp, and pdd) with their corresponding dual geodesic triangles (type ddd, dpd and dpp, respectively).

A particular case of information-geometric manifolds are {\em dually flat spaces}~\cite{IG-2016} induced by a strictly convex and $C^3$ function $F$ called the {\em potential function}. 
In a dually flat space $M$, there are two global dual affine coordinate systems $\theta$ and $\eta$  (with $M=\{ p \st \theta(p)\in\mathrm{dom}(F)\}$),
and a generalization of the Pythagoras' theorem holds~\cite{IG-2016}. 
The divergence $D_F(p:q)$ (potentially oriented distance, $D_F(p:q)\not =D_F(q:p)$) between two points $p$ and $q$ can be expressed using the Bregman divergence $B_F$ on their primal coordinates:
\begin{equation}
D_F(p:q) = B_F(\theta(p):\theta(q))= F(\theta(p))-F(\theta(q))-(\theta(p)-\theta(q))^\top\nabla F(\theta(q)).
\end{equation}
Whenever it is clear from context, we shall write $D$ instead of $D_F$ for conciseness.

Since the canonical divergence of a dually flat space $(M,g,\nabla,\nabla^*)$  amounts to a Bregman divergence~\cite{IG-2016}, we shall also call these dually flat spaces {\em Bregman manifolds} $(M,F)$ in the remainder.
When the Bregman generator is $F_\Euc(\theta)=\frac{1}{2}\theta^\top\theta$, we recover the Euclidean geometry (a self-dual Bregman manifold)
 where the {\em Euclidean divergence} is half of the squared of the Euclidean distance.
Beware that the Euclidean distance is a metric distance but the Euclidean divergence is not. 

Figure~\ref{fig:geotriangles} displays the $8$ types of geodesic triangles in a Bregman manifold obtained for 
the 2D generator $F_\Burg(\theta)=-\log\theta^1-\log\theta^2$ (called Burg negentropy~\cite{BDcentroid-2009}) with corresponding Bregman divergence called the Itakura-Saito divergence~\cite{BDcentroid-2009,BVD-2010}.

\begin{figure}
\centering
\begin{tabular}{|c|c|}\hline
\includegraphics[width=0.4\textwidth]{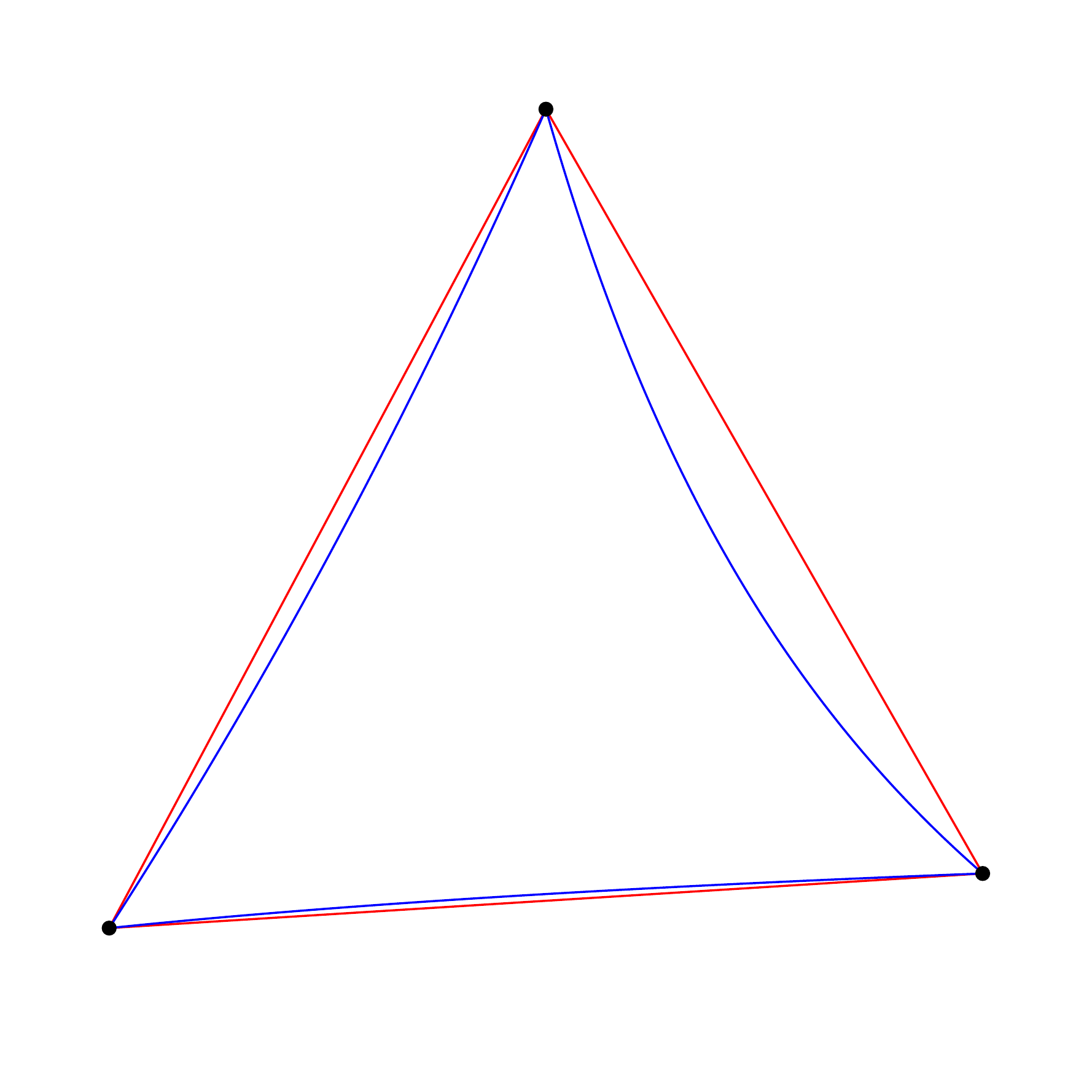}&
\includegraphics[width=0.4\textwidth]{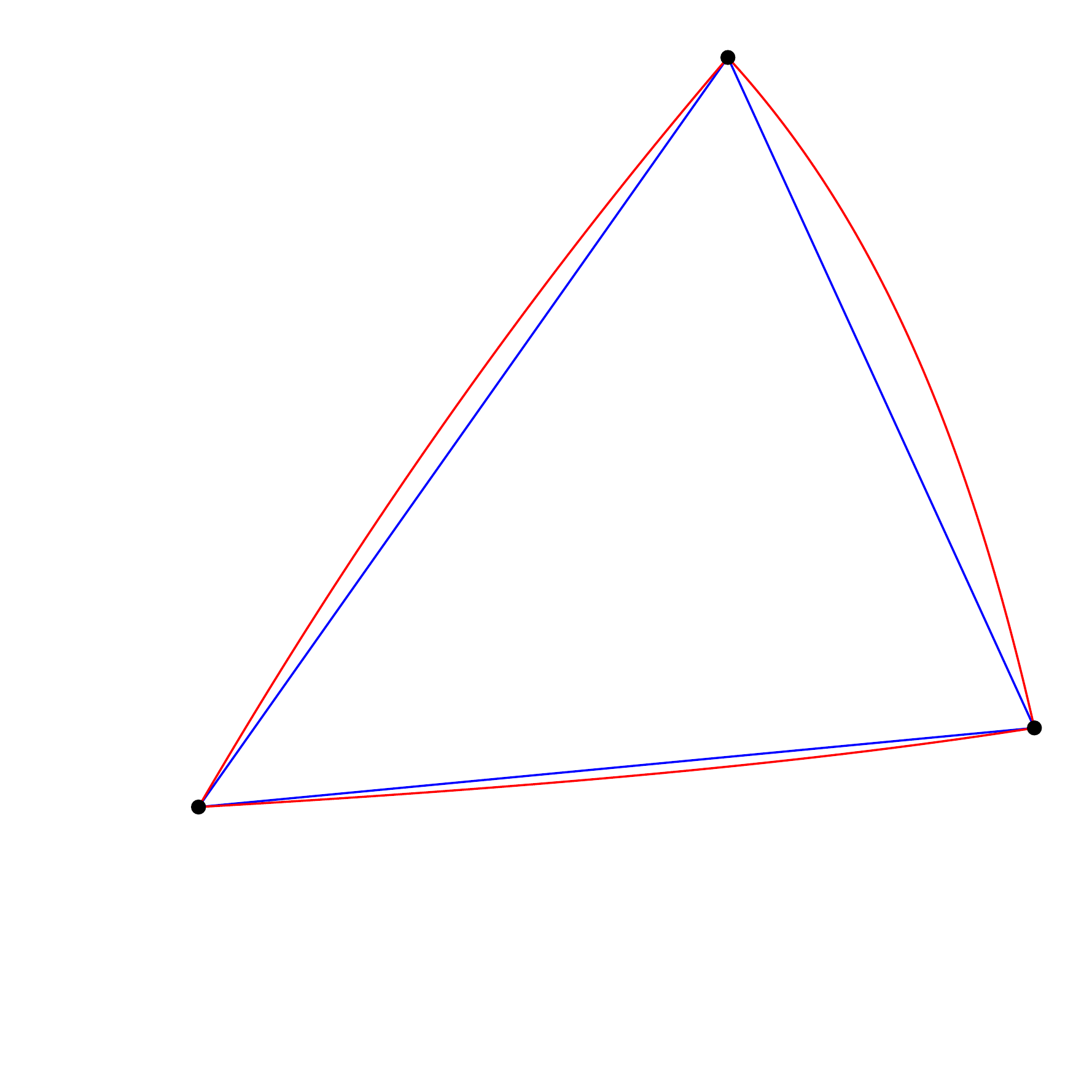} \\ \hline
$\theta$-coordinates & $\eta$-coordinates \\ \hline
\end{tabular}

\caption{
 The $2\times 3=6$ potential geodesic arcs (edges) of a geodesic triangle $T$ visualized both in the primal $\theta$-coordinate system (left) and in the dual $\eta$-coordinate system (right).
The vertices are $\theta(p)=(0.55,0.575)$, $\theta(q)=(0.75,0.95)$ and $\theta(r)=(0.95,0.6)$ for the Bregman manifold defined by the 2D Burg negentropy: $F_\Burg(\theta^1,\theta^2)=-\log(\theta^1)-\log(\theta^2)$.
Primal geodesic arcs are shown in red and are visualized as straight line segments in the $\theta$-coordinate system.
Dual geodesic arcs are shown in blue and are visualized as straight line segments in the $\eta$-coordinate system.
\label{fig:allgeodesics}
}
\end{figure}

\def\ttt{0.20}
\begin{sidewaysfigure}
\centering
\begin{tabular}{l|cccc}
$\theta$ & \includegraphics[width=\ttt\textwidth]{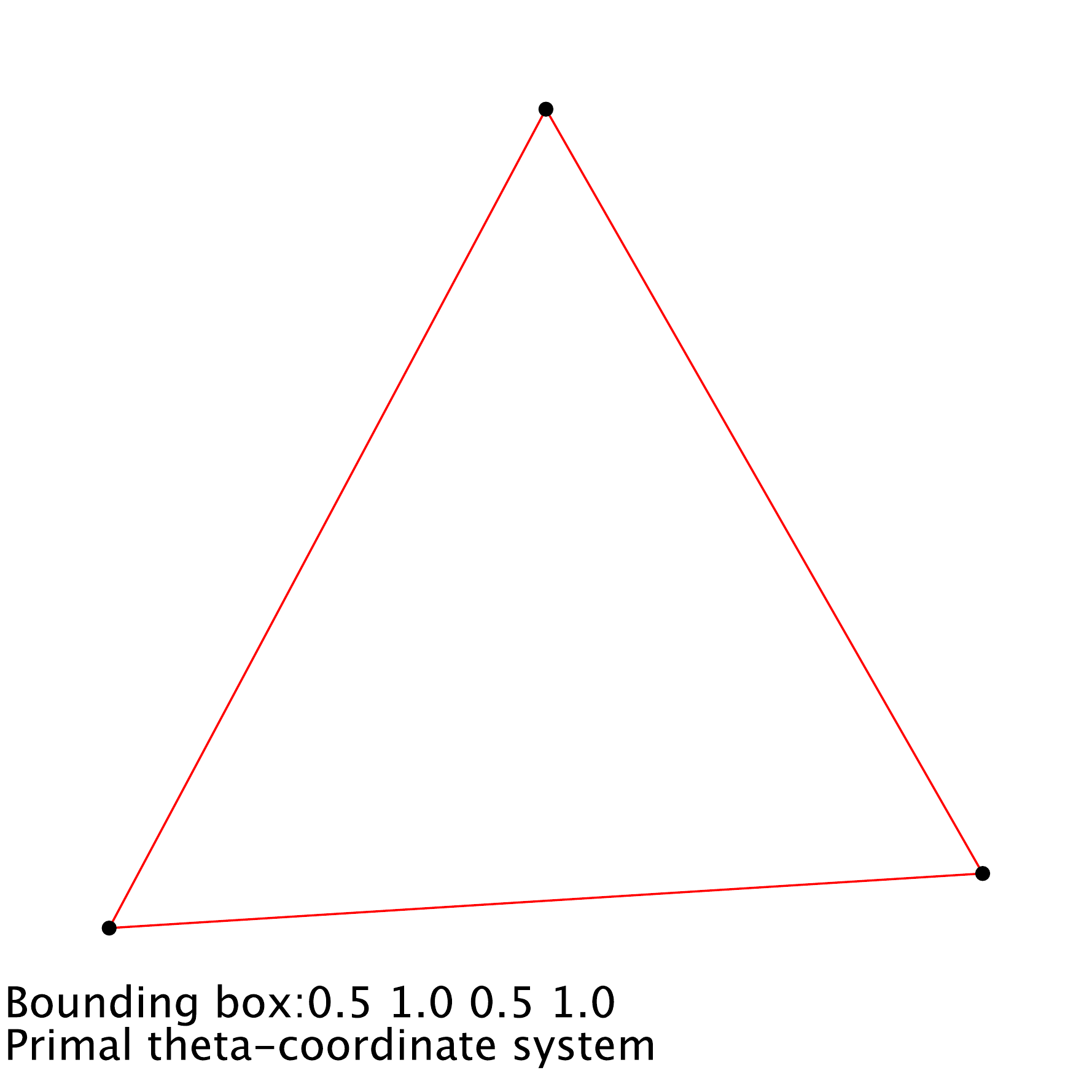}&
\includegraphics[width=\ttt\textwidth]{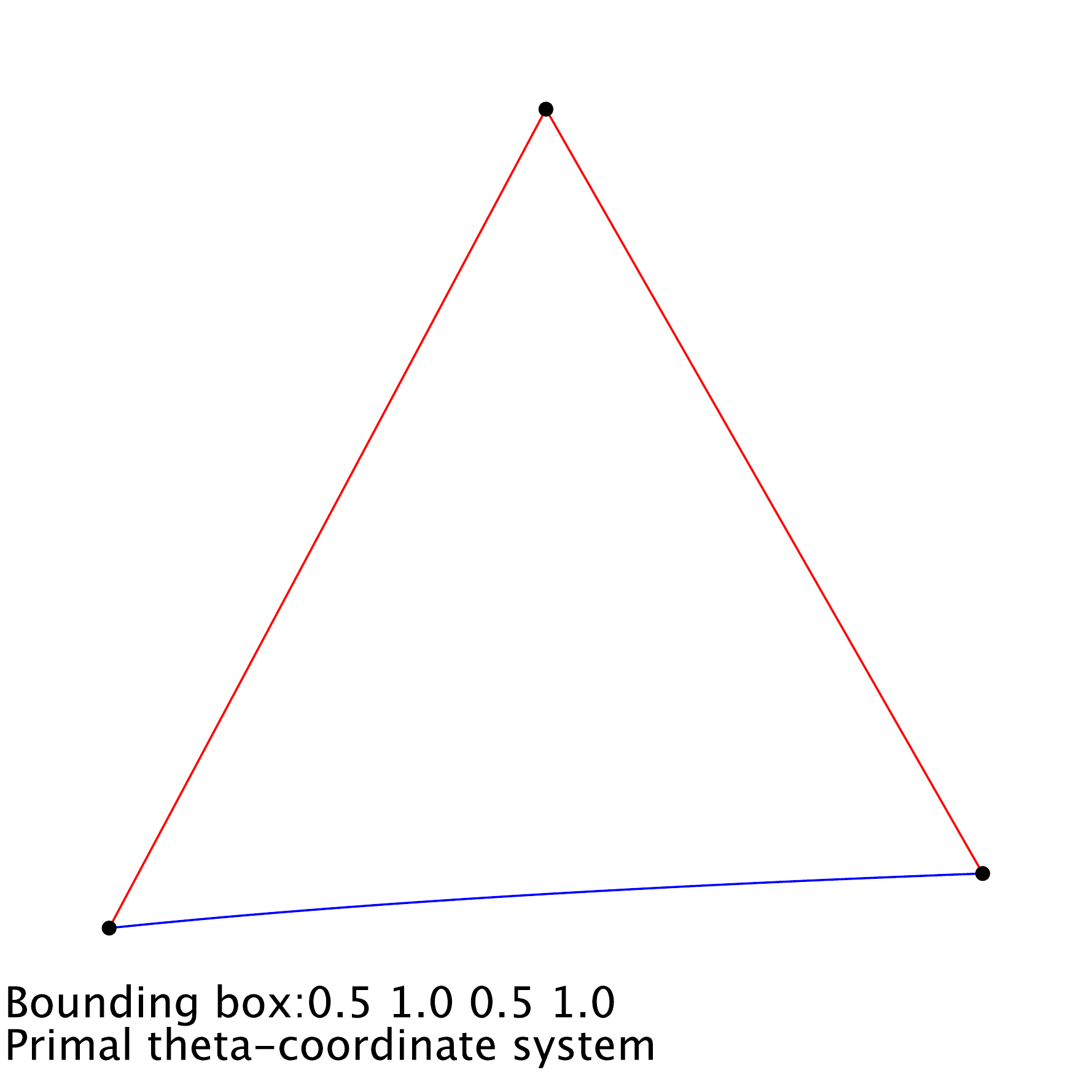}&
\includegraphics[width=\ttt\textwidth]{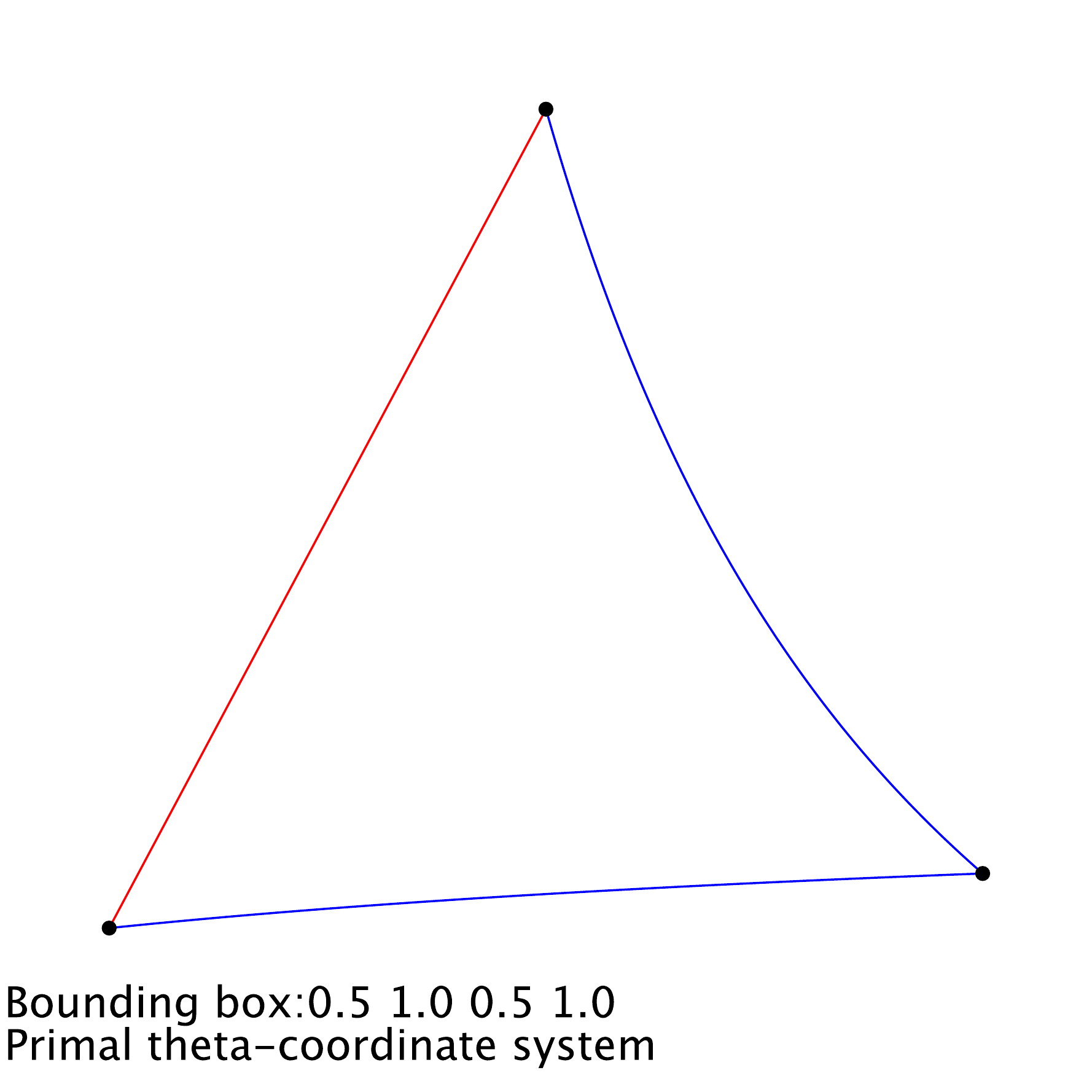}&
\includegraphics[width=\ttt\textwidth]{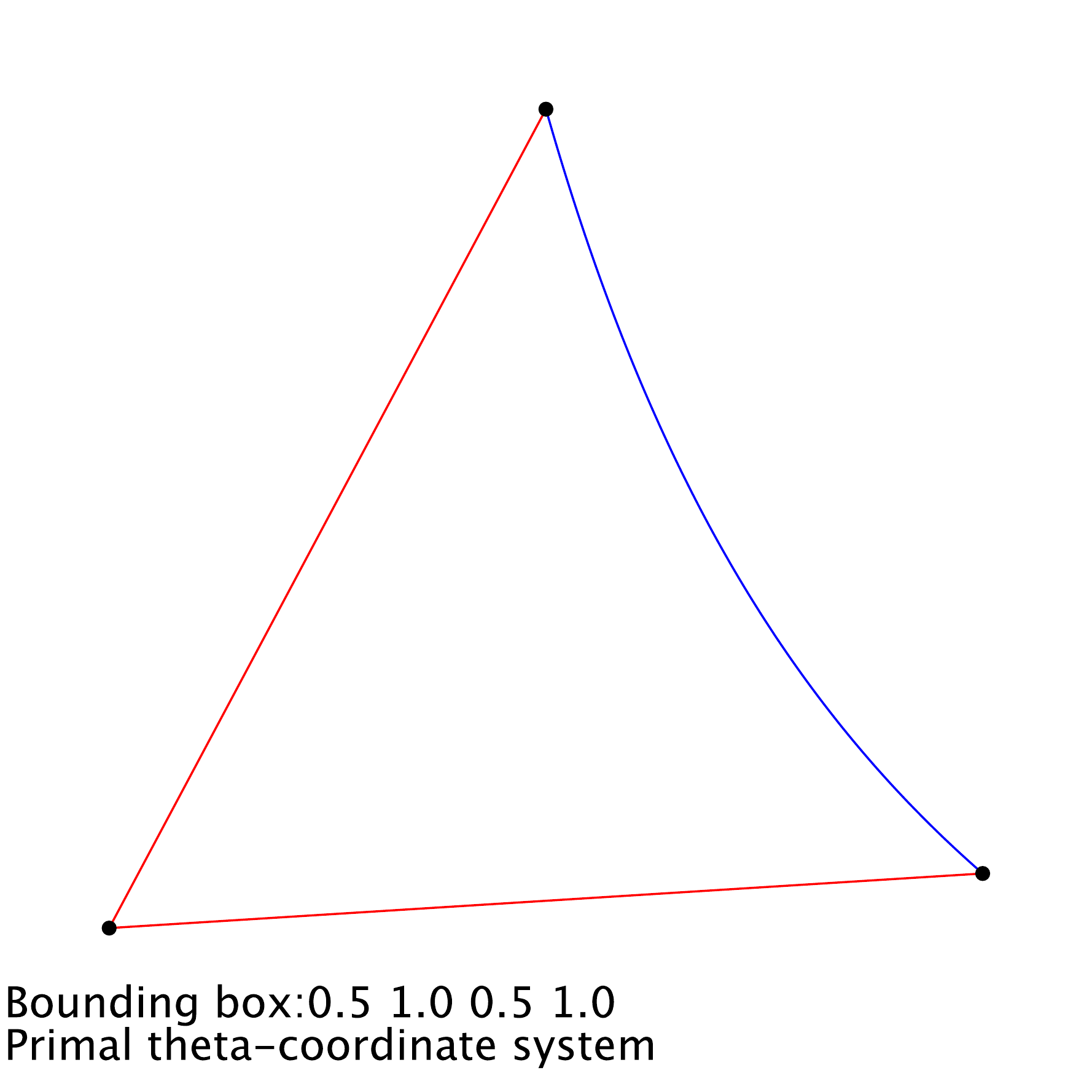}\\
$\eta$ &\includegraphics[width=\ttt\textwidth]{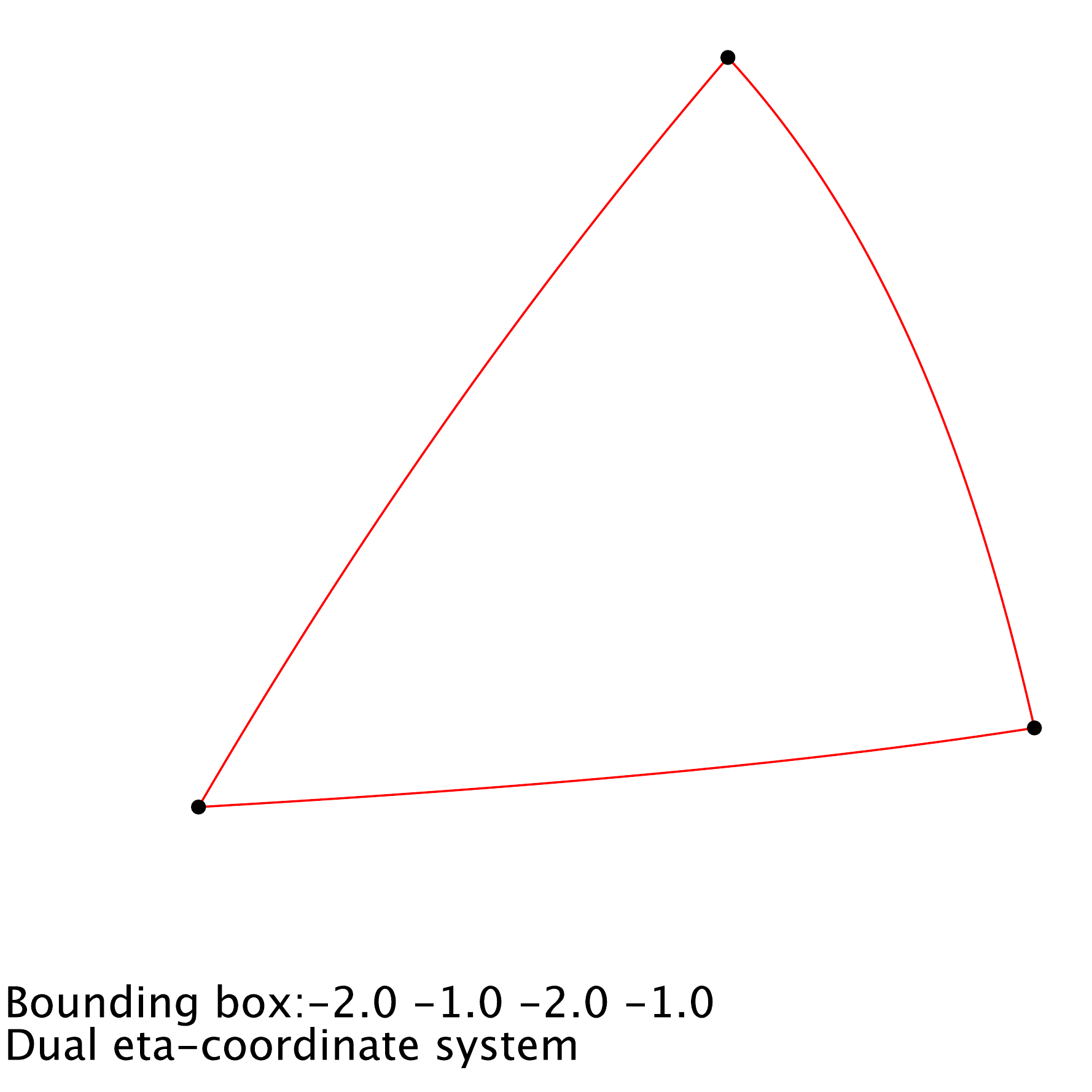}&
\includegraphics[width=\ttt\textwidth]{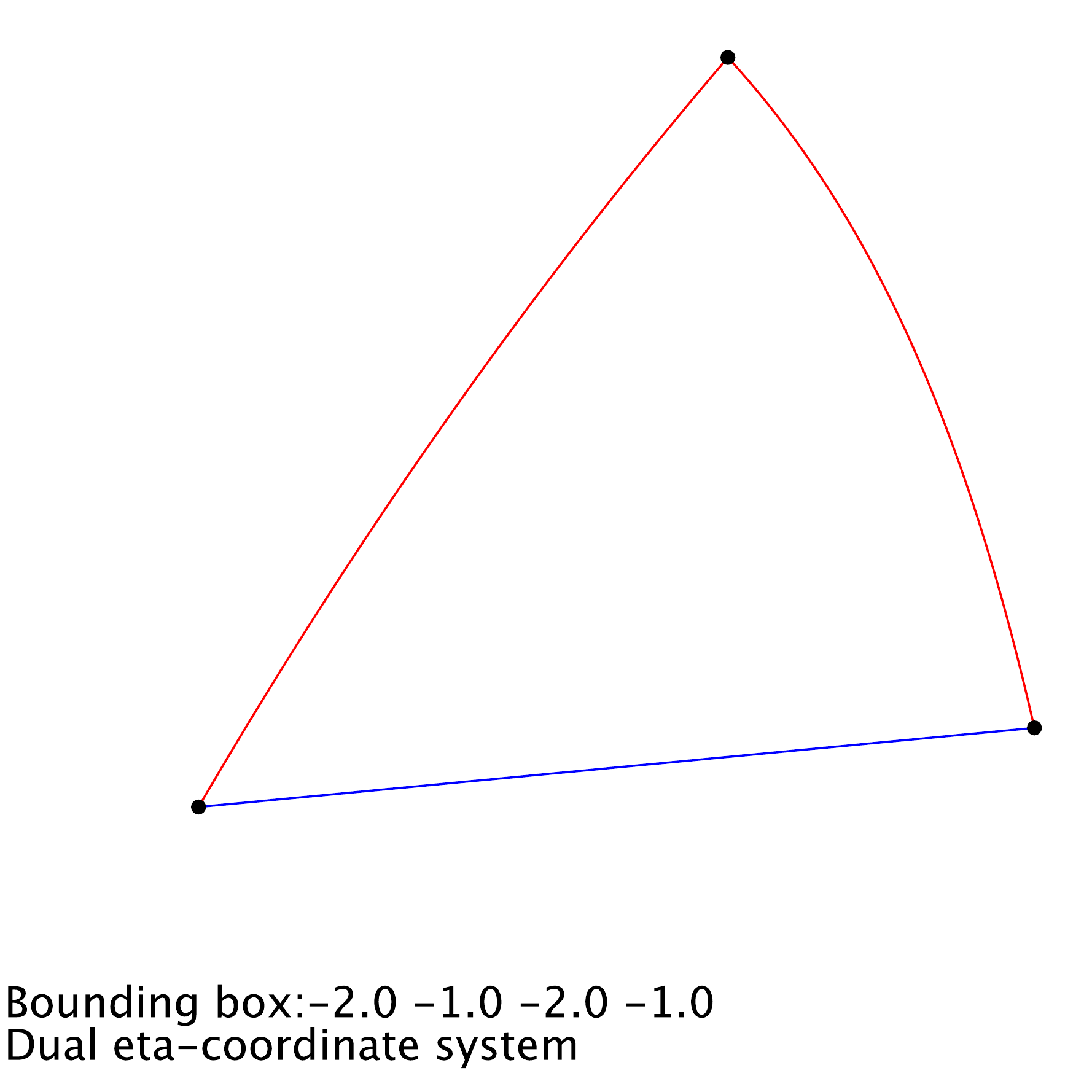}&
\includegraphics[width=\ttt\textwidth]{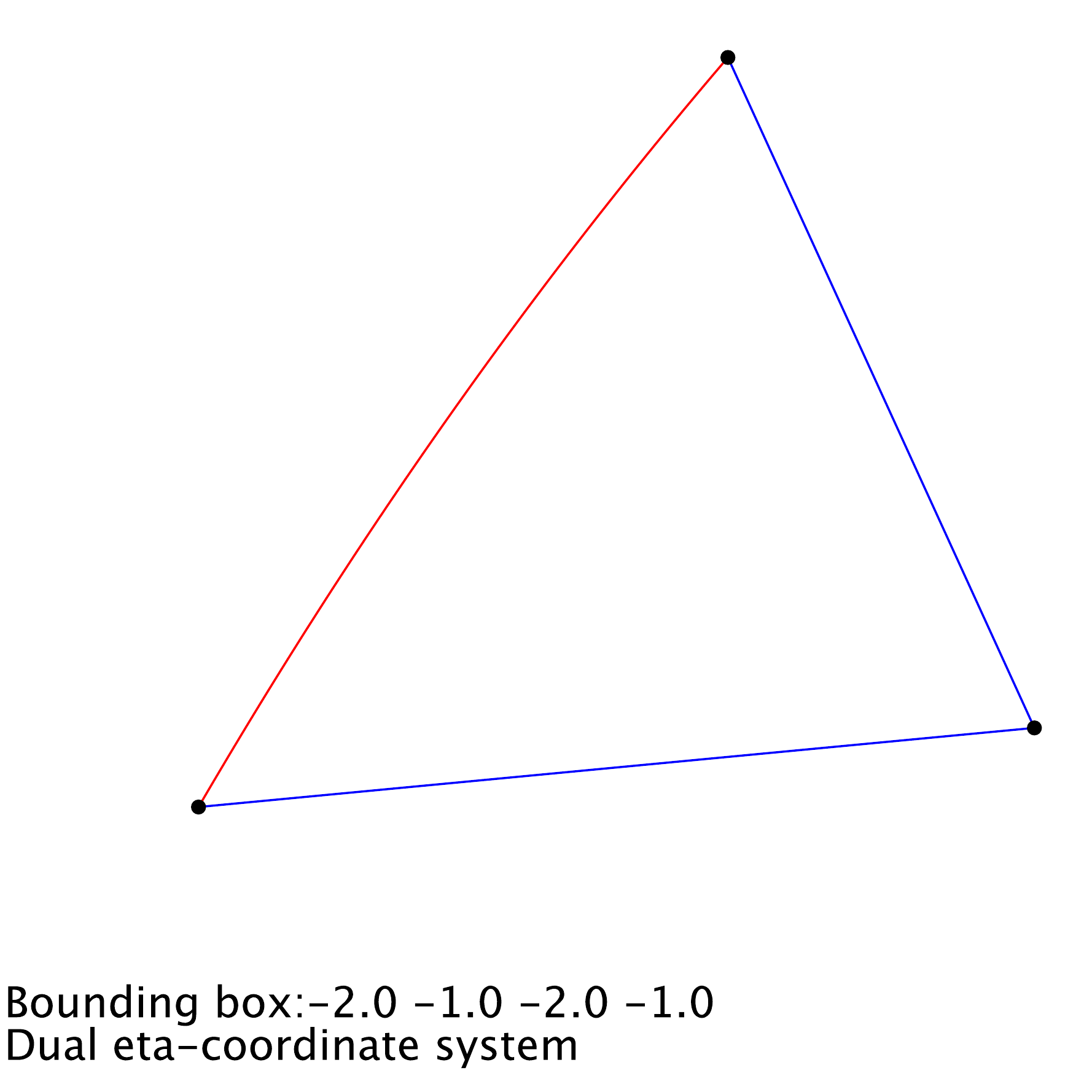}&
\includegraphics[width=\ttt\textwidth]{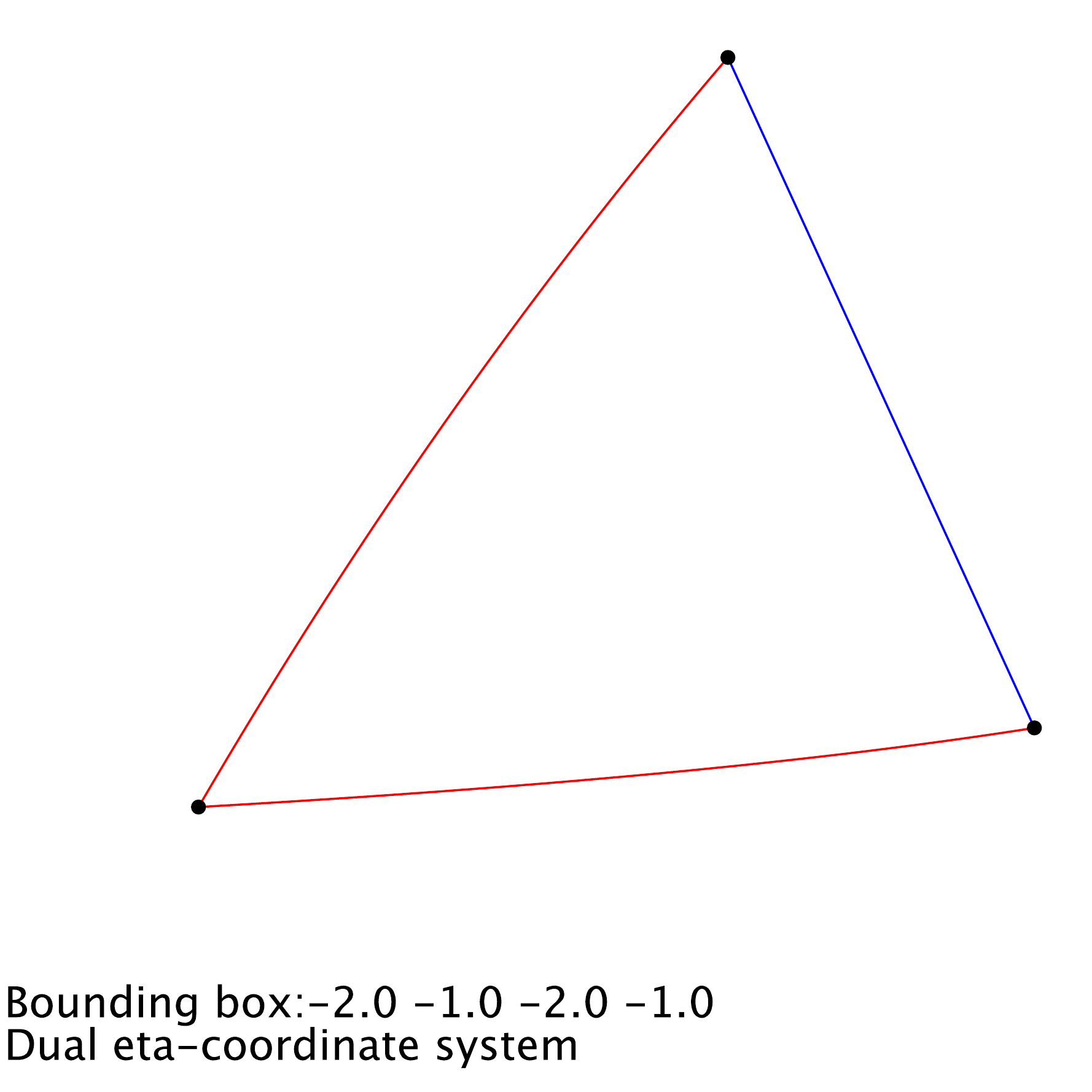}
\end{tabular}

\noindent\rule{\textwidth}{1pt}

\begin{tabular}{l|cccc}
$\theta$ &  \includegraphics[width=\ttt\textwidth]{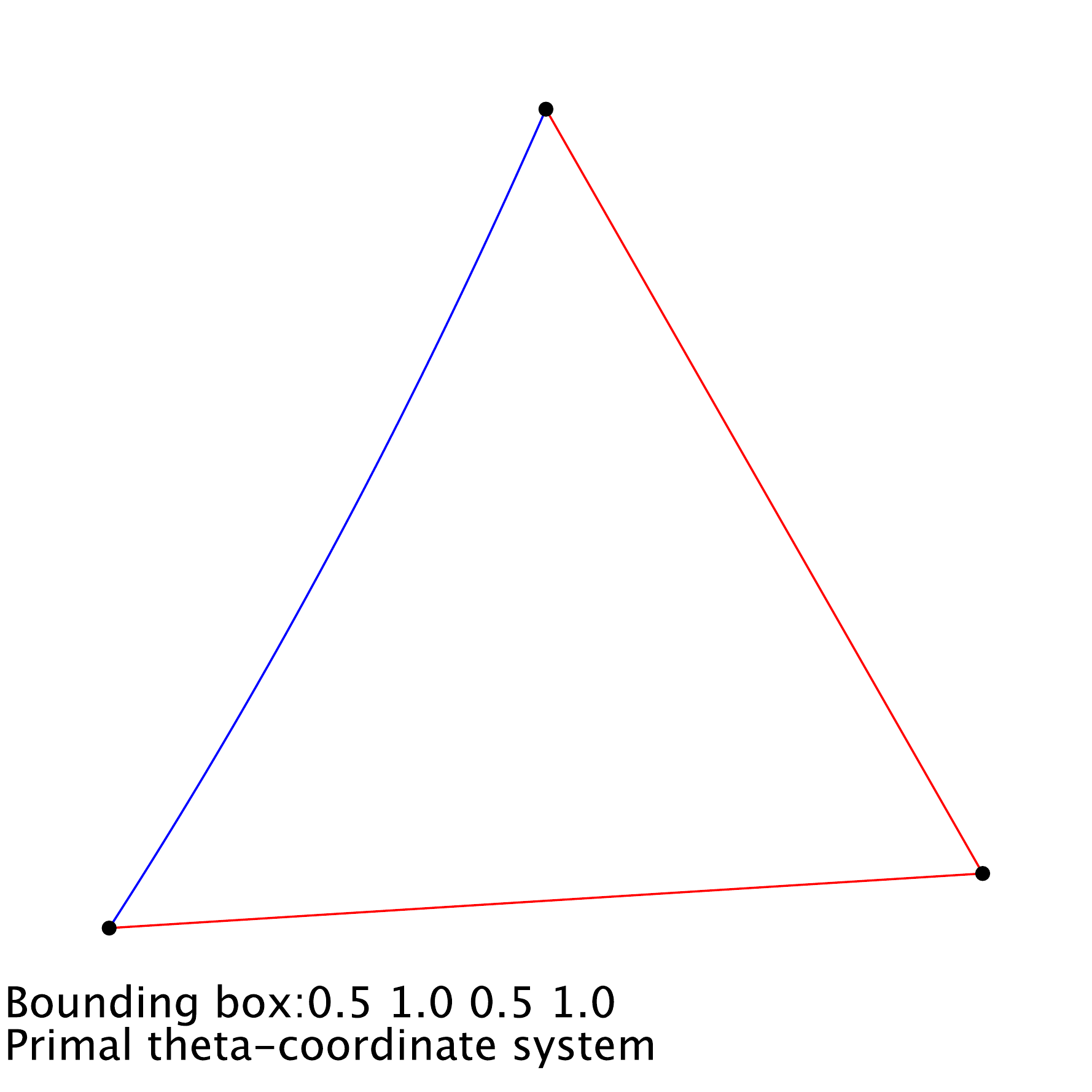}&
\includegraphics[width=\ttt\textwidth]{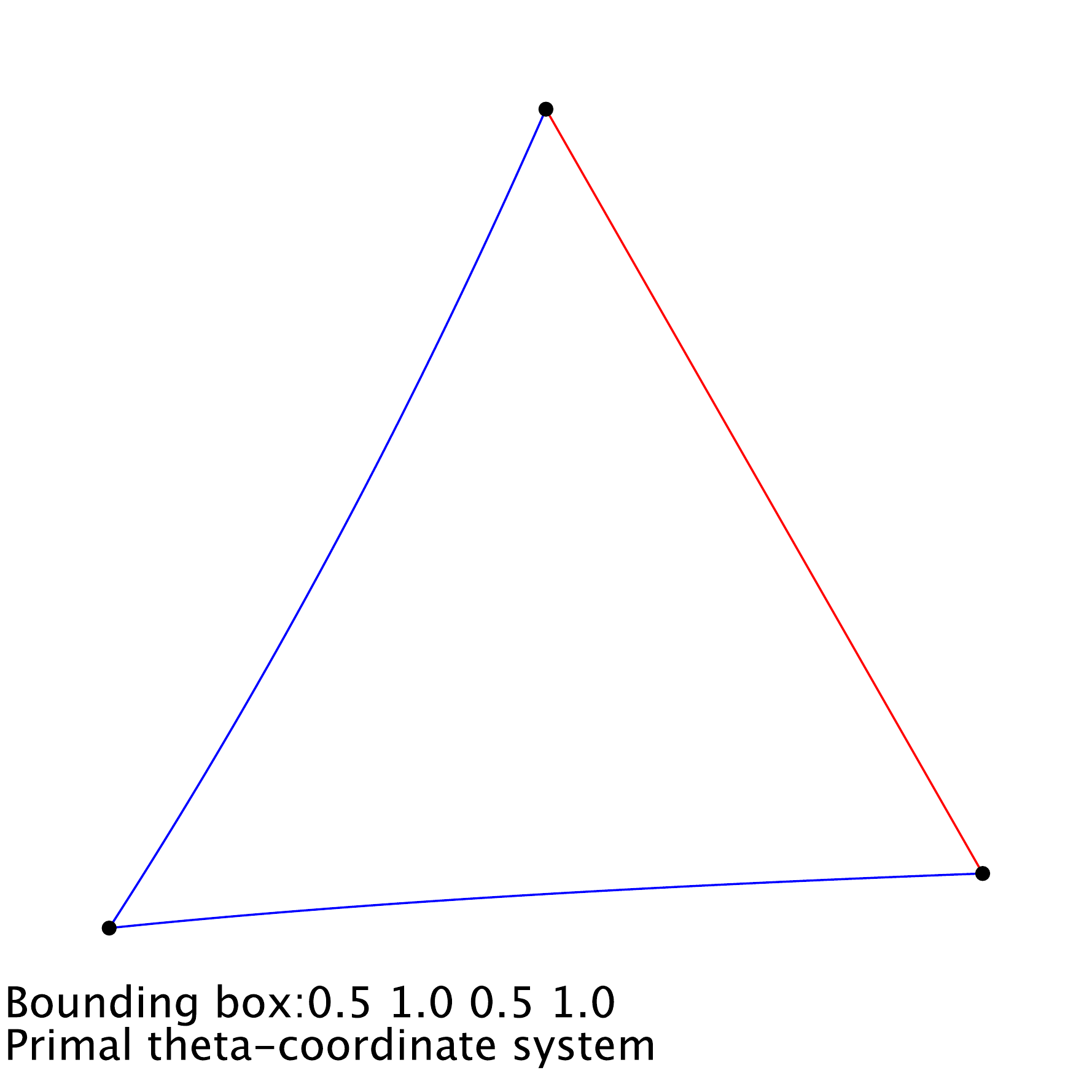}&
\includegraphics[width=\ttt\textwidth]{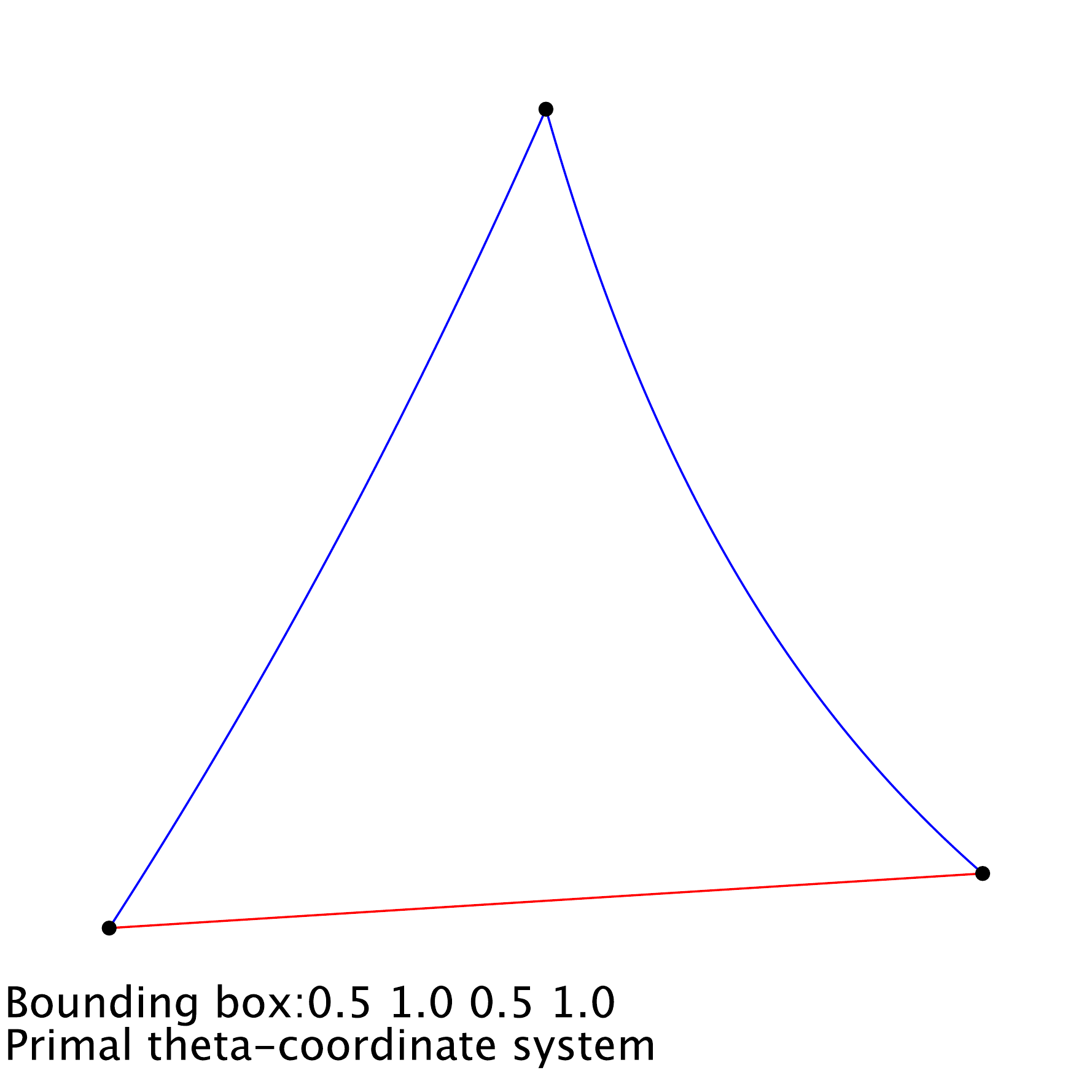}&
\includegraphics[width=\ttt\textwidth]{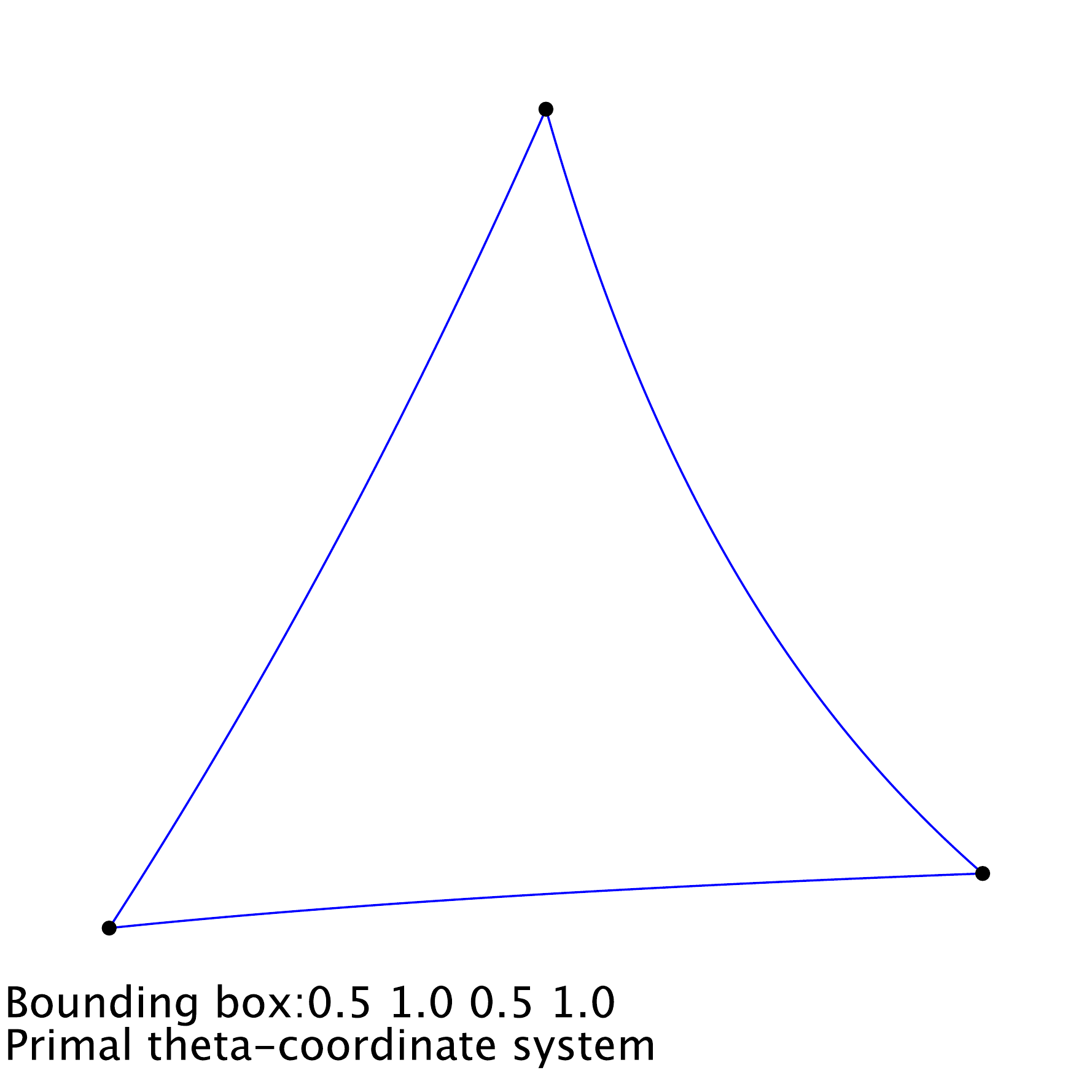}\\
$\eta$ &\includegraphics[width=\ttt\textwidth]{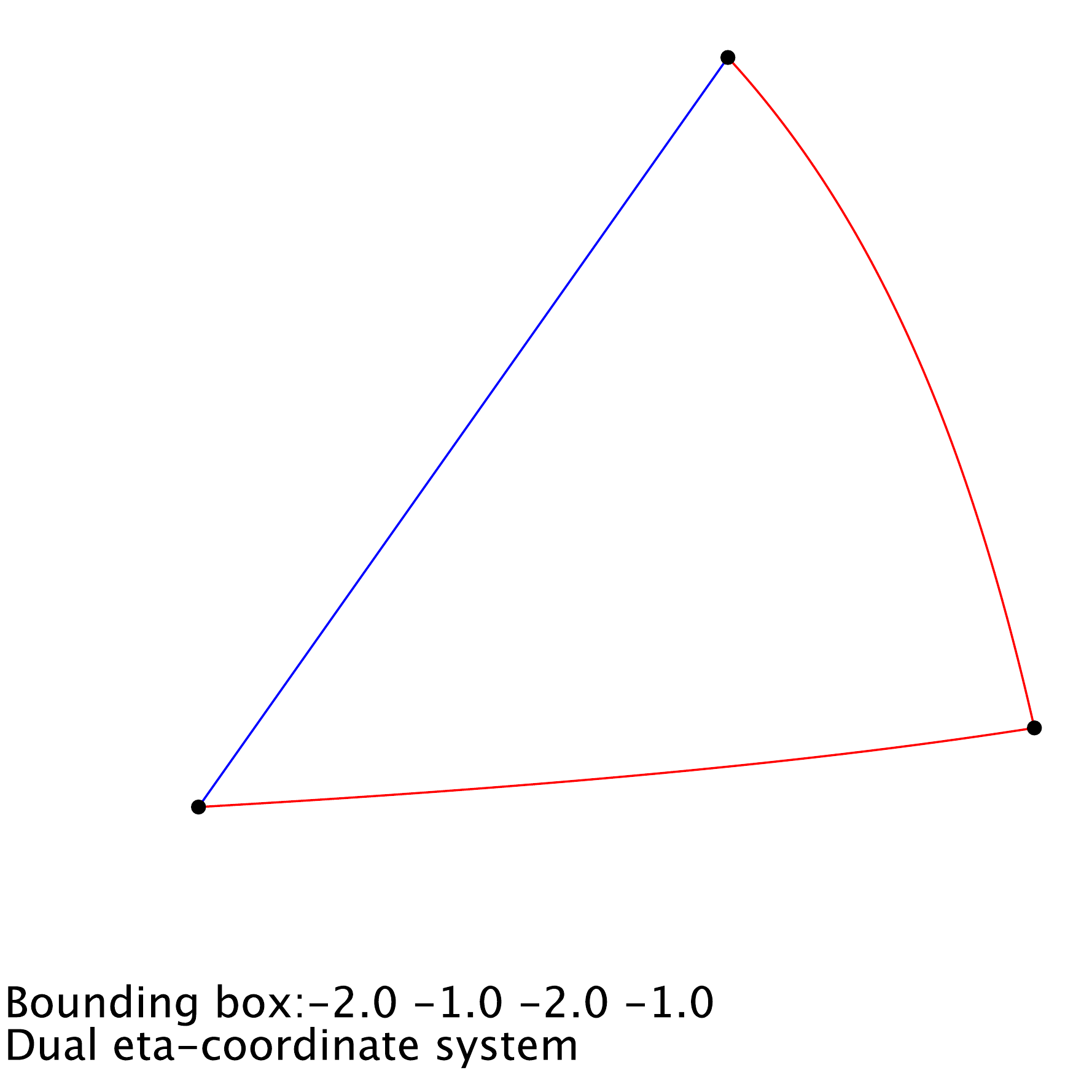}&
\includegraphics[width=\ttt\textwidth]{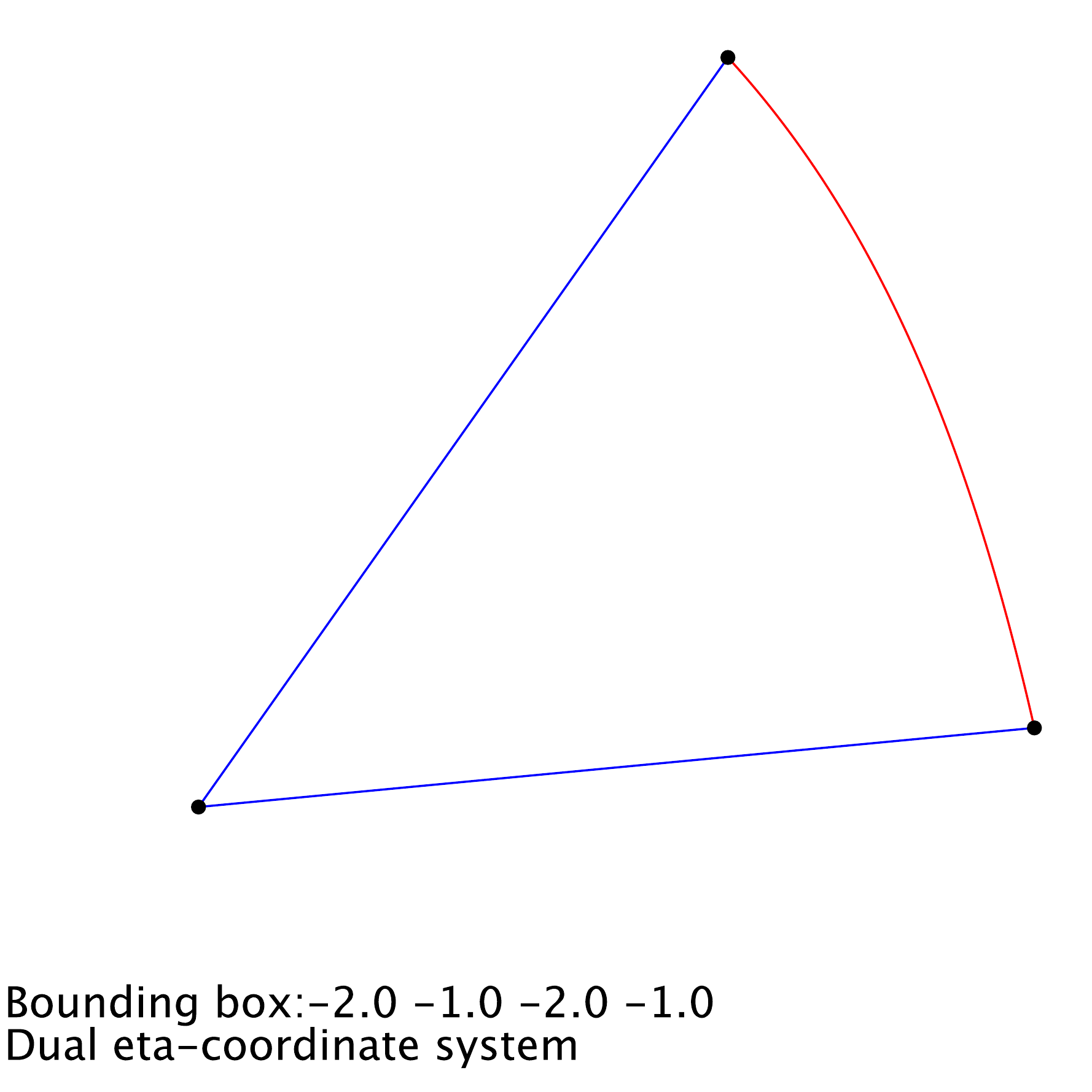}&
\includegraphics[width=\ttt\textwidth]{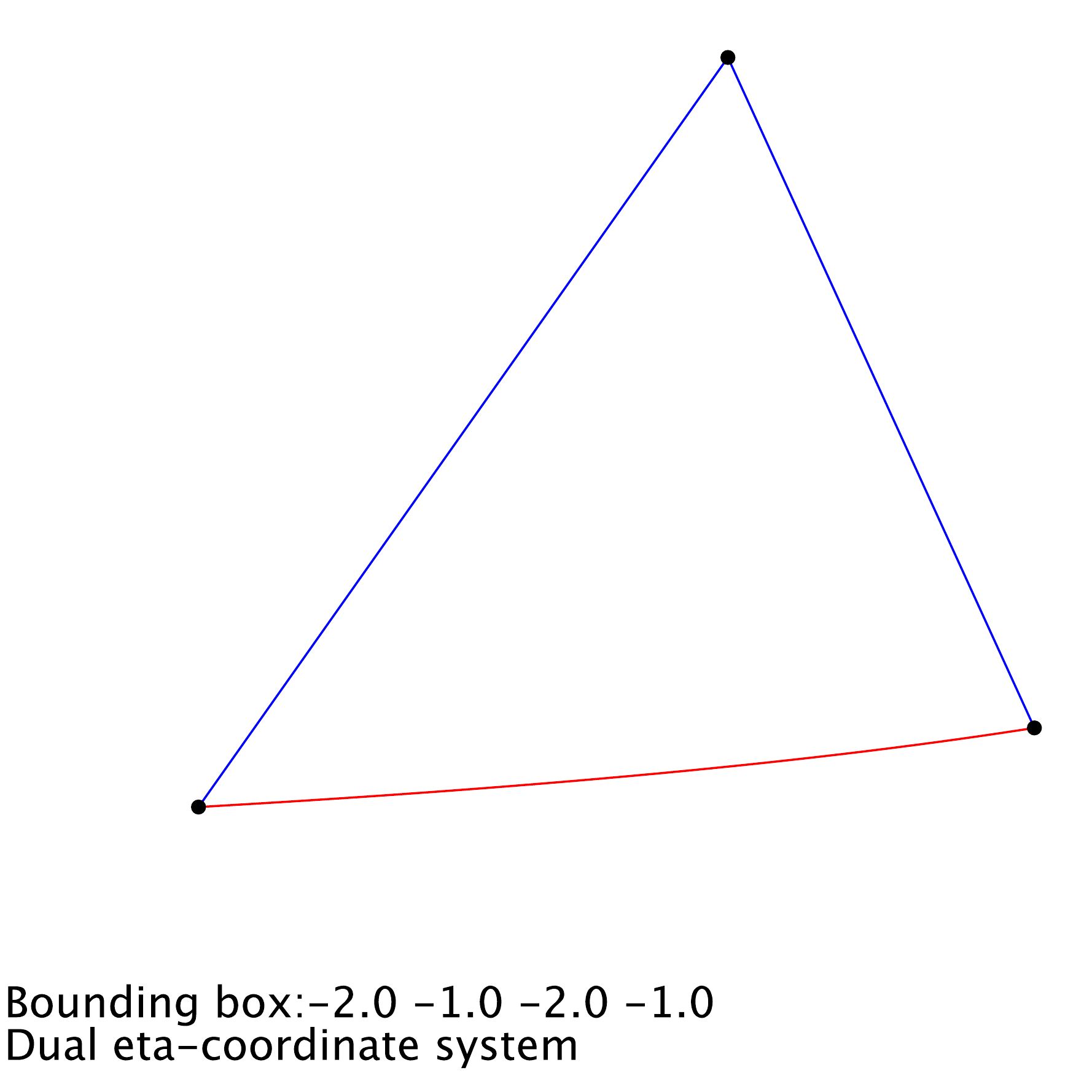}&
\includegraphics[width=\ttt\textwidth]{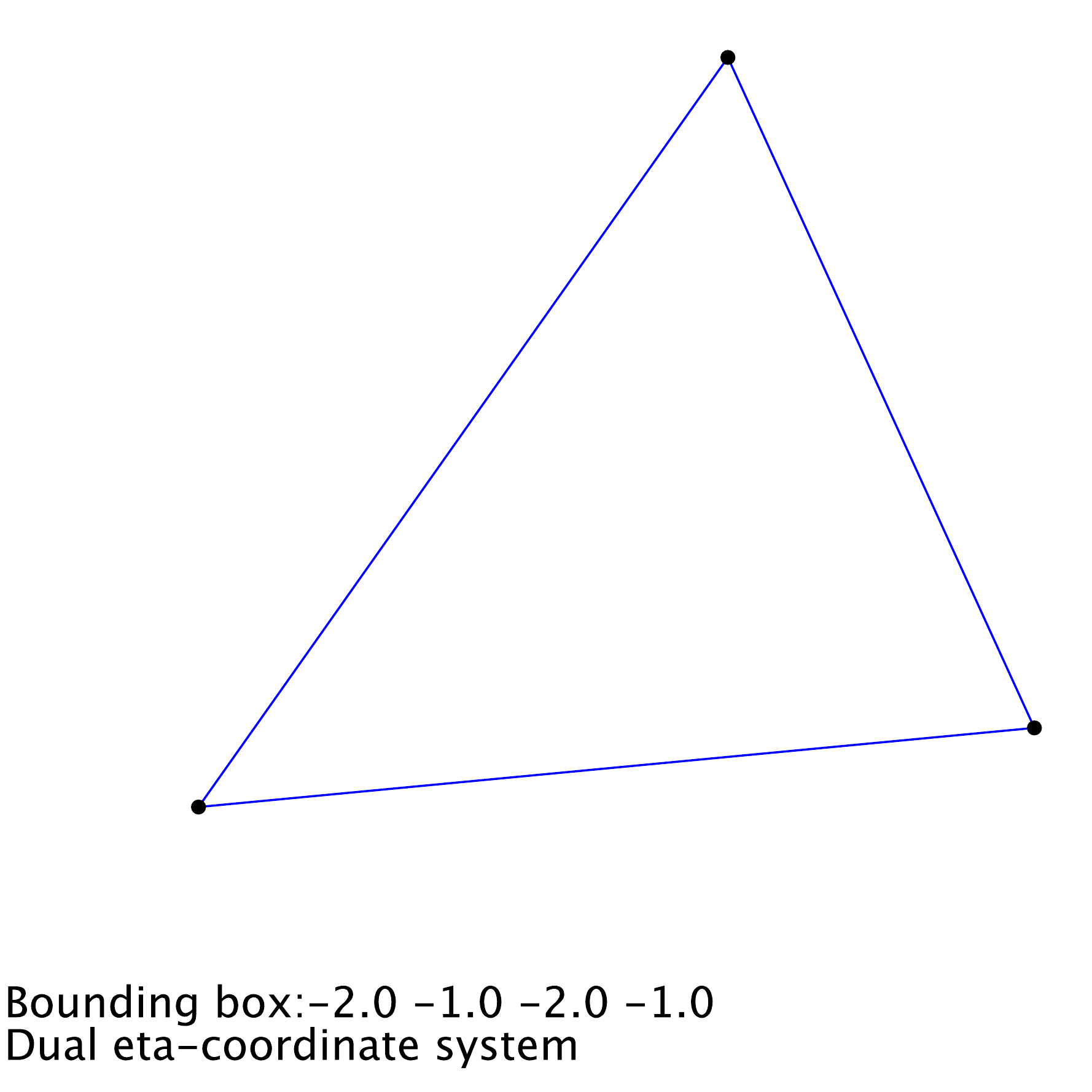}
\end{tabular}

\caption{
The $2^3=8$ types of geodesic triangles visualized both in the primal $\theta$-coordinate system  and the dual $\eta$-coordinate system.
Primal geodesic edges are shown in red and are straight in the $\theta$-coordinate system.
Dual geodesic edges are shown in blue and are straight in the $\eta$-coordinate system.
Refer to Figure~\ref{fig:allgeodesics} for point coordinates.
\label{fig:geotriangles}}

\end{sidewaysfigure}

In this paper, we raise and (partially) answer the following questions on a statistical manifold:
\begin{description}
\item[Q1.] When and how can one build geodesic triangles with one, two and three right angles?


\item[Q2.] When and how can one build dual geodesic triangles $T$ and $T^*$ such that at a triangle vertex, we have two pairs of dual geodesics emanating from that vertex that are simultaneously orthogonal?

\item[Q3.] When is there a relationship or inequality between $\alpha(T)$ and $\alpha(T^*)$?
\end{description}

In this work, we are interested in studying these questions and unraveling some properties of geodesic triangles in the particular case of dually flat spaces,  Bregman manifolds.
In Bregman manifolds, the  Pythagorean theorem allows one to check that a pair of dual geodesics is orthogonal at a given point by checking some divergence identity.
In general, the  Pythagorean theorem is useful to prove uniquess of projections in particular settings, interpret geometrically projections based on the divergence, and design derivative-free projection algorithms~\cite{Akaho-2019}. 
In information geometry, besides the dually flat spaces, Pythagorean theorems with corresponding divergence identities have also been reported for $\alpha$-divergences on the probability simplex~\cite{IG-2016} and the  logarithmic $L^\alpha_H$-divergences~\cite{Wong-2018} where $H$ is an exponentially concave generator.

The paper is organized as follows: 
First, we quickly review the construction of Bregman manifolds in~\S\ref{sec:BM}, thereby introducing familiar concepts and notations of information geometry~\cite{IG-2016}, and give as examples of Bregman manifolds details for the  Mahalanobis manifolds (\S\ref{sec:Mah}), the extended Kullback-Leibler manifold (\S\ref{sec:KL}), the Itakura-Saito manifold (\S\ref{sec:IS}), and the multinoulli manifolds (\S\ref{sec:Multinoulli}).
Then Section~\ref{sec:doubleright} shows, whenever it is possible, how to build geodesic $\nabla$-triangles which have one, two or three right angles (thus necessarily exhibiting angle excesses for two and three right angle triangles).
In \S\ref{sec:simultaneous},  we prove that given two distinct points $p$ and $q$, the locii of points $r$ for which we have simultaneously the dual Pythagorean theorems holding at $r$ are the intersection of an autoparallel $\nabla$-submanifold with an autoparallel $\nabla^*$ submanifold (i.e., the intersection of a $\theta$-flat with a $\eta$-flat~\cite{IG-2016}). 
We report explicitly the construction method for the 2D Itakura-Saito manifold and visualize several such triangles.
Finally, we summarize and hint at further perspectives in~\S\ref{sec:concl}.

\section{Dually flat spaces: Bregman manifolds}\label{sec:BM}

We first explain the mutually orthogonal primal basis and reciprocal basis in an inner product space in~\S\ref{sec:innerp}.
Then we describe dual geodesics and their tangent vectors and the dual parallel transport in~\S\ref{sec:dualpt}.
In \S\ref{sec:pt}, we explain the Pythagorean theorem, and  \S\ref{sec:BMex} provides some common examples of Bregman manifold: Mahalanobis self-dual manifolds, the extended Kullback-Leibler manifold and the Itakura-Saito manifold.

\subsection{Preliminary: inner product space and reciprocal basis}\label{sec:innerp}
An {\em inner product space} is a vector space $V$ equipped with a symmetric positive definite bilinear form $\inner{\cdot}{\cdot}: V\times V\rightarrow \bbR$
called the inner product.
The length of a vector $v\in V$ is given by its induced norm $\|v\|=\sqrt{\inner{v}{v}}$, and the angle between any two vectors $u$ and $v$ is measured  as
$\alpha(u,v)=\arccos  \left(\frac{\inner{u}{v}}{\|u\| \|v\|} \right)$ (in radians).
We consider finite $D$-dimensional inner product spaces where a vector $v$ can be expressed in {\em any} basis $B=\{e_1,\ldots, e_D\}$ (a maximum set of linearly independent vectors) by its components $v_B=(v^1,\ldots, v^D)$:
$v=\sum_{i=1}^D v^i e_i$. 
Vector $v$ can also be expressed equivalently in an other basis $\hat{B}=\{\hat{e}_1,\ldots, \hat{e}_D\}$: $v=\sum_{i=1}^D {\hat{v}^i} \hat{e}_i$. 
Notice that in general the components $v_B$ and $v_{\hat{B}}$ are different although they express the same geometric vector $v$ according to their respective basis $B$ and $\hat{B}$.
We can express the basis vectors $\hat{e}_i$ using basis $B$ and the basis vectors $e_i$ using basis $\hat{B}$ as
$\hat{e}_i=\hat{A}_i^j e_j$ and $e_j=A_i^j\hat{e}_j$, respectively (using Einstein summation convention). 
We have $\hat{A}_i^jA_j^k=\delta_i^k$ where $\delta_{i}^{j}$ is the Kr\"onecker symbol ($\delta_{i}^{j}=1$ iff. $i=j$ and $0$ otherwise), and the changes of basis reflects on components as
 $v^i=\hat{A}_j^i \hat{v}^j$ and $\hat{v}^i=A_j^i v_i$.
When the basis $B$ is orthonormal (i.e., $\inner{e_i}{e_j}=\delta_{ij}$ where $\delta_{ij}$ is the Kr\"onecker symbol: $\delta_{ij}=1$ iff. $i=j$ and $0$ otherwise) the vector components can be retrieved from the inner product as
$v^i=\inner{v}{e_i}=\inner{\sum_{j=1}^D v^j e_j}{e_i}=\sum_{j=1}^D v^j \inner{e_j}{e_i}$. 
This is no longer true for non-orthormal basis (e.g., an orthogonal but non-orthonormal basis or an oblique basis).
Let us introduce the unique {\em reciprocal basis}  $B^*=\{{e^*}^1,\ldots, {e^*}^D\}$ such that by construction, we have $\inner{e_i}{{e^*}^j}=\delta_i^j$.
The vector $v$ can be expressed in the reciprocal basis as $v=\sum_{i=1}^D v_i {e^*}^i$.
The vector components $v^i$ (superscript notation) wrt. basis $B$ are called the {\em contravariant components}, and $v^i=\inner{v}{{e^*}^i}$.
The vector components $v_i$ (subscript notation) wrt. basis $B^*$ are called the {\em covariant components}, and $v_i=\inner{v}{{e}_i}$.
(We shall explain this contravariant/covariant component terminology at the end of this section.)

Let $G=[g_{ij}=\inner{e_i}{e_j}]_{ij}$ and $G^*=[{g^*}^{ij}=\inner{{e^*}^i}{{e^*}^j}]_{ij}$ denote the $D\times D$ positive definite matrices, called dual metrics. These dual metric matrices are inverse of each other: $G^*=G^{-1}$.
In textbooks, one often  drops the superscript star '*' in the notation of the reciprocal basis and the dual riemannian metric, see~\cite{IG-2016}. Here, we keep them explicitly for easing the understanding, even if they load the notations.

We can convert the contravariant components $v^i$ of a vector $v$ to its covariant components $v_i$, and vice versa, using these metric matrices:
$v_i=  \sum_{i=1}^D g_{ij} v^j$ and $v^i=\sum_{i=1}^D {g^*}^{ij} v_j$.
Let $[u]_B$ denote the vector components of $u$ in basis $B$ arranged in a column vector.
Then we rewrite the contravariant/covariant conversions as matrix-vector multiplications of linear algebra:
 $[v]_{B^*} = G \times [v]_B$ and $[v]_B= G^* \times [v]_{B^*}$.
The inner product between two vectors can be written equivalently using algebra as
\begin{eqnarray}
\inner{u}{v} &=& \sum_{i=1}^D u_i v^i = [u]_{B^*}^\top \times [v]_B,\\
 &=& \sum_{i=1}^D u^i v_i =   [u]_{B}^\top\times  [v]_{B^*},\\
 &=& \sum_{i=1}^D [u]_{B}^\top \times  G^* \times [v]_{B},\\
 &=& \sum_{i=1}^D [u]_{B^*}^\top\times  G \times [v]_{B^*}.
\end{eqnarray}

In differential geometry~\cite{IG-2016,EIG-2018}, a smooth manifold $M$ is equipped with a {\em metric tensor field} $g$ that defines on each tangent plane $T_p$ of $p\in M$ an inner product. The dual of a tangent plane $T_p$ is the cotangent plane $T_p^*$, a vector space of linear functionals.
In general, tensor fields define at each point of the manifold component-free geometric entities that can be interpreted as multilinear functionals over Cartesian products of dual covector and vector spaces.
A vector $v$ of $T_p$ is a rank-$1$ tensor that can be expressed either using  the covariant or contravariant components of a reciprocal basis.
Thus one should not confuse the notion of ``geometric vectors'' that are tensors (coordinate-free objects independent of the choice of the basis) with the ``column vectors'' of vector components in a basis which are used to perform linear algebra calculations on (geometric) vectors.
When the components of a tensor vector varies with the inverse transformation of the change of basis, we say that we have a contravariant (tensor) vector (a $(0,1)$-tensor),  and its components are called contravariant components.
When the components of a tensor vector varies according to the transformation of the change of basis, we say that we have a covariant (tensor) vector (a covector or $(1,0)$-tensor, i.e., an element of the dual vector space $V^*$ of linear functionals or linear forms),  and its components are called covariant components. The metric tensor $g$ is a $(2,0)$ covariant tensor~\cite{EIG-2018}.

\subsection{Dual geodesics and their tangent vectors, and dually coupled parallel transport}\label{sec:dualpt}

Let $F(\theta)$ be a $D$-dimensional $C^3$ real-valued function defined on an {\em open} convex domain $\Theta$, 
and denote by $F^*(\eta)$ its Legendre-Fenchel convex conjugate~\cite{IG-2016,EIG-2018}: $F^*(\eta)=\sup_{\theta\in\Theta} \theta^\top\eta-F(\theta)$.
The dual potential functions $F$ and $F^*$ induce two torsion-free flat affine connections~\cite{EIG-2018} $\nabla$ and $\nabla^*$, respectively.
A Bregman manifold $M$ is equipped with two global affine coordinate systems $\theta(\cdot)$ (the $\nabla$-affine coordinate system) and $\eta(\cdot)$ (the $\nabla^*$-affine coordinate system)
such that it comes from Legendre-Fenchel transformation that $\eta(\theta)=\nabla F(\theta)$ and $\theta(\eta)=\nabla F^*(\eta)$.
Let $\theta^i(p)$ and $\eta_i(p)$ denote the primal $i$-th $\theta$-coordinate functions and the dual $i$-th  $\eta$-coordinate functions of a point $p$, for $i\in\{1,\ldots, D\}$ so that $\theta(p)=(\theta^1(p),\ldots, \theta^D(p))$ and $\eta(p)=(\eta_1(p),\ldots, \eta_D(p))$.
Notations are summarized in Appendix~\ref{sec:notations}.
Any point $p\in M$ can be expressed equivalently either in the primal global $\theta$-chart or the dual global $\eta$-chart.
The dual geodesics\footnote{Given an affine connection $\nabla$, the $\nabla$-geodesic is an autoparallel curve~\cite{EIG-2018}.} $\gamma_{pq}$ and $\gamma_{pq}^*$ passing through two given points $p, q \in M$ write simply using the dual coordinate systems as follows:
\begin{eqnarray}
\gamma_{pq}  &=& \{ x_\lambda\in M \st \theta(x_\lambda)=(1-\lambda)\theta(p)+\lambda\theta(q), \quad \lambda\in [0,1] \},\\
\gamma_{pq}^* &=& \{ x_\lambda\in M \st \eta(x_\lambda)=(1-\lambda)\eta(p)+\lambda\eta(q), \quad \lambda\in [0,1] \}.
\end{eqnarray} 

In general, a vector field $v(t)$ is parallel along a smooth curve $c(t)$ iff.
\begin{equation}
\forall i\in\{1,\ldots, D\},\ \dot{v}^i+\sum_{j,k=1}^D \Gamma_{jk}^i \dot{x}^j v^k=0.
\end{equation}
Since the Christoffel symbols $\Gamma_{jk}^i(\theta)=0$ (and ${\Gamma^*}_{jk}^i(\eta)=0$) in a dually flat space, we have $\dot{v}^i=0$ and recover the equation of the primal geodesic (and $\dot{v}_i=0$ for the equation of the dual geodesic, respectively).

The dual Riemannian metrics $[g_{ij}]=[g_{ji}]$ and $[{g^*}^{ij}]=[{g^*}^{ji}]$ are induced by the Hessians of the dual potential functions $F$ and $F^*$ (both  symmetric positive-definite matrices), respectively.
At any given point $p\in M$, we consider the natural basis $\{e_i=\partial_i=\frac{\partial}{\partial\theta^i}\}$ and the reciprocal basis   $\{{e^*}^i=\partial^i=\frac{\partial}{\partial\eta_i}\}_i$ of the tangent plane $T_p$  so that $g(e_i,{e^*}^j)=\delta_{i}^j$.
That is, the basis vectors of the primal and reciprocal basis are {\em mutually orthogonal}. 
We have $[g_{ij}]=[g(e_i,e_j)]_{ij}=\nabla^2F(\theta)$ and $[{g^*}^{ij}]=[g^*({e^*}^i,{e^*}^j)]_{ij}=\nabla^2F^*(\eta)$.
We can check that the natural primal basis and the dual reciprocal basis are mutually orthogonal  from Crouzeix identity:  
\begin{equation}\label{eq:crouzeix}
\nabla^2F(\theta)\nabla^2F^*(\eta)=I,
\end{equation}
where $I$ denotes the $D\times D$ identity matrix.
Note that Eq.~\ref{eq:crouzeix} can be read using a single coordinate system as 
$\nabla^2F(\theta)\nabla^2F^*(\eta(\theta))=\nabla^2F(\theta)\nabla^2F^*(\nabla F(\theta))=I$ 
or $\nabla^2F(\theta(\eta))\nabla^2F^*(\eta)=\nabla^2F(\nabla F^*(\eta))\nabla^2F^*(\eta)=I$.

A tangent vector $v_{pq}$ of the primal geodesic $\gamma_{pq}$ at $p$ is a vector of the tangent plane $T_p$ with 
contravariant components $\theta(q)-\theta(p)$ (say, expressed in the primal   basis 
$B_p=\{e_1,\ldots, e_D\}$):  $v=\sum_{i=1}^D (\theta^i(q)-\theta^i(p))e_i$ with $v^i=\theta^i(q)-\theta^i(p)$.
Indeed, we have 
\begin{equation}
\frac{d\gamma_{pq}(t)}{\dt} = \frac{d}{\dt} \left((1-t)\theta(p)+t\theta(q)\right)= \sum_i  (\theta^i(q)-\theta^i(p))\partial_i,
\end{equation}
where $\partial_i=\frac{\partial}{\partial\theta^i}$.
Similarly, a tangent vector $v_{pq}^*$ at the dual geodesic $\gamma_{pq}^*$ is a vector of the tangent plane $T_p$ with covariant components $\eta(q)-\eta(p)$ (expressed in the reciprocal basis $B_p^*=\{{e^*}1,\ldots, {e^*}^D\}$):    
$v=\sum_{i=1}^D (\eta_ii(q)-\eta_i(p)){e^*}^i$.
Indeed, we check that 
\begin{equation}
\frac{d\gamma_{pq}^*(t)}{\dt} = \frac{d}{\dt} \left((1-t)\eta(p)+t\eta(q)\right)= \sum_i  (\eta_i(q)-\eta_i(p))\partial^i,
\end{equation}
where $\partial^i=\frac{\partial}{\partial\eta_i}$.
In general, for separable Bregman generators, i.e., $F(\theta)=\sum_{i=1}^D F_i(\theta^i)$ where the $F_i$'s are scalar strictly convex and $C^3$ functions, we can choose both the primal and reciprocal basis to be orthogonal but they are not necessarily orthonormal 
(except in the special case of the Euclidean geometry obtained by $F(\theta)=\sum_{i=1}^D \sqr(\theta^i)$ where $\sqr(x):=x^2$ denotes the square function).

The metric tensor field $g$ defines a smooth scalar product $g(\cdot,\cdot)$ on the tangent bundle $TM$ (informally, the union of all tangent planes) such that  for any two vectors $u,v\in T_p$, we have $g_p(u,v)=u_iv^i=u^iv_i$. 
In each tangent plane, we thus have an inner product space.
The contravariant and covariant components of a vector $v$ can be retrieved using the inner product with the reciprocal basis and the primal basis, respectively:
$v^i=g(v,{e^*}^i)$, and $v_i=g(v,e_i)$.
We can also use vector-matrix multiplications of linear algebra to calculate 
 the inner product:
$g_p(u,v)=[u^i]^\top\times \nabla^2 F(\theta(p))\times [v^i] = [u_i]^\top\times \nabla^2 F^*(\eta(p))\times [v_i]$, where $[u^i]$ and $[u_i]$ denote the column-vector of the contravariant components of $u$ and covariant components of $u$, respectively.
Since $v^i=\nabla^2F^*(\eta(p))\times [v_i]$ and $v_i=\nabla^2 F(\theta(p))\times [v^i]$, we therefore have the following equivalent rewritings of the inner product:
\begin{eqnarray}
g_p(u,v) &=& [u_i]^\top \times [v^i],\\
&=& [u^i]^\top\times [v_j],\\
&=& [u^i]^\top\times \nabla^2 F(\theta(p))\times [v^i],\\
&=& [u_i]^\top\times \nabla^2 F^*(\eta(p))\times [v_i].
\end{eqnarray}

When the generator $F(\theta)$ is separable, i.e., $F(\theta)=\sum_{i=1}^D F_i(\theta^i)$ for univariate generators $F_i$'s, the inner product at a tangent plane $T_p$ writes equivalently as (using Einstein summation convention):
\begin{eqnarray}
g_p(u,v) &=& u^i F_i''(\theta^i(p)) v^i,\\
&=& u_i {F_i^*}''(\eta_i(p)) v_i,\\
&=& u_iv^i,\\
&=& u^iv_i.
\end{eqnarray}

Given two vectors $u,v\in T_p$, we measure their lengths $\|u\|_p$ and $\|v\|_p$ and the interior angle $\alpha_p(u,v)$ between them as:
\begin{eqnarray}
\|u\|_p &=& \sqrt{g_p(u,u)} = \sqrt{u^i u_i},\\
 \|v\|_p&=& \sqrt{g_p(v,v)} = \sqrt{v^i v_i},\\
 \alpha_p(u,v)&=&\arccos\left(\frac{u^i v_i}{\|u\|_p \|v\|_p} \right)=\alpha_p(v,u).
\end{eqnarray}

Notice that the Hessian metric $g$ and $g^*$ are not conformal (i.e., not a positive scalar function of the Euclidean metric), and that we cannot ``read'' the angles directly in the $\theta$- or $\eta$-coordinate systems. 
In other words, the Euclidean angles displayed in the $\theta$- or $\eta$-coordinate systems do not correspond to the intrinsic angles of the underlying Bregman geometry.
Figure~\ref{fig:metric} shows a visualization of the metric tensor field for the Itakura-Saito manifold (with $[g_{ij}](\theta)=\nabla^2F(\theta)=\diag(\frac{1}{\sqr(\theta^1}, \frac{1}{\sqr(\theta^2)})$ for 2D $\theta=(\theta^1,\theta^2)$).
Dually flat spaces are a particular case of more general (local) Hessian structures studied in~\cite{Hessian-2007}.

\begin{figure}
\centering
\fbox{\includegraphics[width=0.45\textwidth]{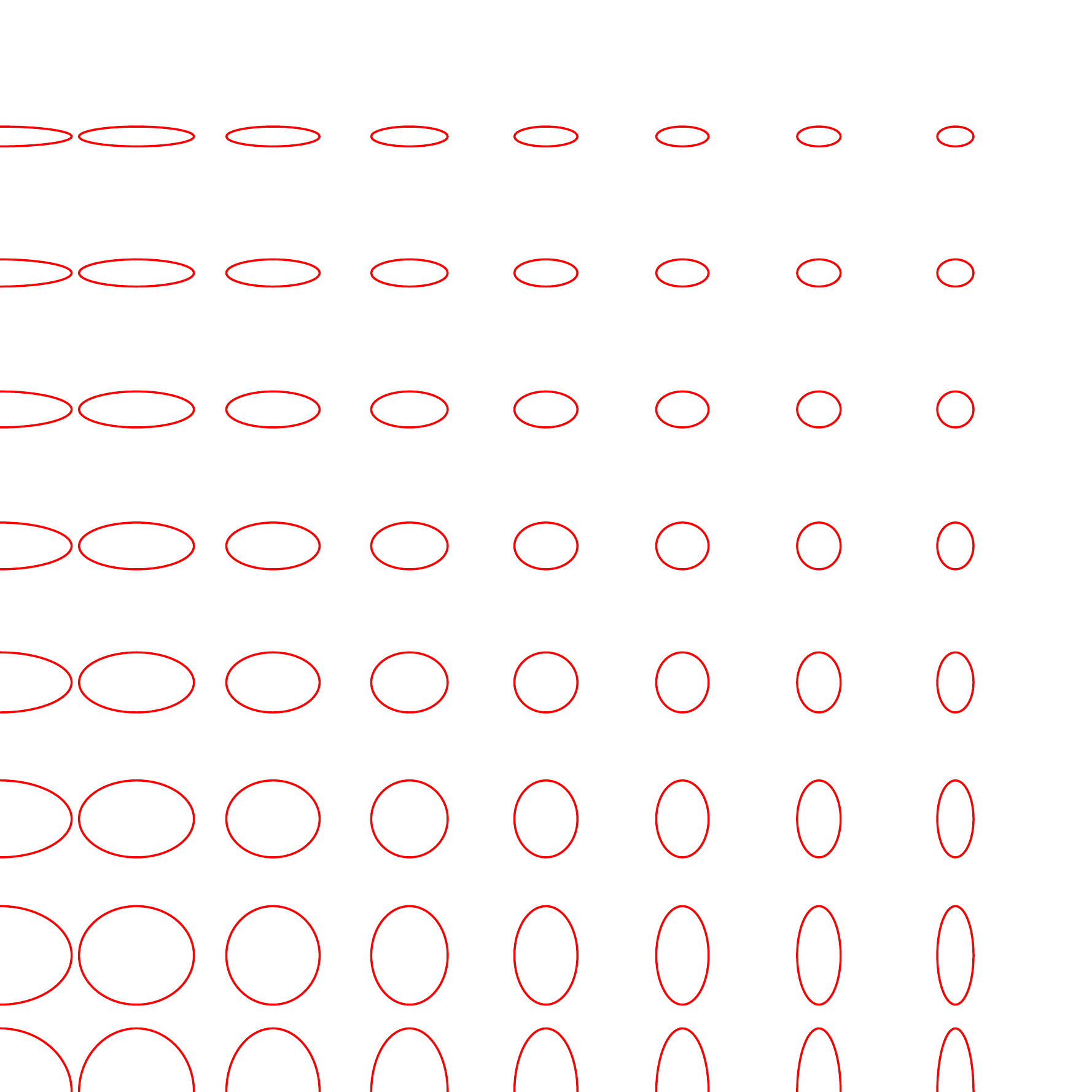}}

\caption{Visualization of the metric tensor field $g(\theta)=\diag\left(\frac{1}{\sqr(\theta^1)},\frac{1}{\sqr(\theta^2)}\right)$ for the 2D Itakura-Saito manifold.
The metric tensor field $g$ is not conformal.
This can be seen from ellipses depicting the metric tensors (tiny Bregman circles) at  regularly sampled grid positions which are not scaled (Euclidean) circles.
\label{fig:metric}}
\end{figure}

In information geometry~\cite{IG-2016}, the dually flat torsion-free affine connections\footnote{The notion of dual connections of information geometry is more general than the notion of conjugate connections of affine differential geometry~\cite{Kurose-1990} which stems from dual affine immersions.} $\nabla$ and $\nabla^*$ are coupled the metric tensor $g$.
An affine connection $\nabla$ is flat iff. there exists a coordinate system such that its Christoffel symbols $\Gamma$ expressed in a coordinate system vanishes.
For Bregman manifolds, we have both $\Gamma_{ijk}(\theta)=0$ and ${\Gamma^*}^{ijk}(\eta)=0$, see~\cite{IG-2016}, so the connections $\nabla$ and $\nabla^*$ are both flat.
Notice that a cylinder is $\nabla$-flat for the natural Euclidean connection but geodesic triangles on the cylinder have interior angles not summing up to $\pi$.
That is, let $\mathcal{X}({M})$ denotes the space of smooth vector fields (the cross sections of the tangent bundles $TM$). 
We say that the two torsion-free affine connections $\nabla$ and $\nabla^*$ are  dual when it holds that
\begin{equation}
\forall X, Y, Z \in \mathcal{X}({M}), \quad X g(Y, Z)=g\left(\nabla_{X} Y, Z\right)+g\left(Y, \nabla_{X}^{*} Z\right).
\end{equation}
An important consequence of the coupling of the dual connections to the metric is that the dual parallel transport preserves the metric.

Let us  choose by convention to fix the primal basis of the tangent vectors at $T_p$ for any $p\in M$ to be the $e_i$'s, the one-hot vectors in the $\theta$-coordinate system, i.e., $e_i=(\underbrace{0,\ldots,0}_{i-1},1,0,\ldots,0)$ (the standard basis of $\bbR^D$). This is the canonical basis.\footnote{In differential geometry, the tangent plane at a point $p$ is the space of all linear derivations that satisfies the Leibniz's rule. A basis $\{t_i\}_i$ of $T_p$ is such that $t_i(f)=\frac{\partial f}{\partial\theta^i}$.}
Since the $\nabla$-connection is flat, the primal parallel transport $\prod_{p,q}(v)$ of a vector $v$ of $T_p$ to a corresponding vector $T_q$ is {\em independent} of the chosen smooth curve, and we have for $v=\sum_i v^i \left.e_i\right|_{p}\in T_p$ 
\begin{equation}
\prod_{p,q}(v)= v^i \left.e_i\right|_{q}\in T_q.
\end{equation}
That is, the contravariant components of $v$ do  not change with primal parallel transport assuming that the primal basis is fixed for all tangent planes, i.e., $B_q=B_p, \forall p,q\in M$ (and $ \left.e_i\right|_{p}= \left.e_i\right|_{q}$).

Similarly, since the $\nabla^*$-connection is flat, the dual parallel transport $\prod_{p,q}^*(v)$ of a vector $v$ of $T_p$ to a corresponding vector of $T_q$ is independent of the chosen smooth curve, and we have
\begin{equation}
\prod_{p,q}^*(v)= v_i \left.{{e^*}^i}\right|_{q}.
\end{equation}
However, because $B_q^*=\{ \left. {{e^*}^i} \right|_{q} \}$ is the reciprocal basis of  the fixed primal basis
$B_q=\{ \left.{e_i}\right|_{q} \}=e_i$, 
the basis $B_q^*$ varies accordingly  to the metric tensor $g$ (and $B_q^*\neq B_p^*$).
Indeed, in general, we cannot  fix both the primal and reciprocal basis in all tangent planes because they relate to each other by construction  by the metric tensor. 
This is only possible when $G(\theta(p))=G^*(\eta(p))=I$, the $D\times D$ identity matrix, corresponding to the special case of Euclidean  or Mahalanobis geometry.

Let us now check the metric compatibility of the dual parallel transport of two vectors $u,v\in T_p$ to $T_q$:
\begin{equation}
g_q\left( \prod_{p,q}(u) , \prod_{p,q}^*(v)  \right) = g_p(u,v).
\end{equation}

\begin{proof}
We have
\begin{eqnarray}
g_q\left( \prod_{p,q}(u) = u^i \left.e_i\right|_{q}, \prod_{p,q}^*(v)=v_i \left.{{e^*}^i}\right|_{q} \right) &=& u^iv_i g_q\left(\left.e_i\right|_{q},\left.{{e^*}^i}\right|_{q} \right),\\
&=& u^iv_i,\\
&=&  u^iv_i g_p\left(\left.e_i\right|_{p},\left.{{e^*}^i}\right|_{p} \right),\\
&=& g_p(u,v).
\end{eqnarray}
\end{proof}

Figure~\ref{fig:PTex} displays the parallel transport of the tangent vector  $\dot\gamma_{pq}^*(0)$ from $T_p$ to $T_{\gamma_{pq}^*(t)}$ for several time steps $t$.
Although the dual parallel transport preserves the metric, the length of a vector transported by either the primal or the dual parallel transport varies.

\begin{figure}
\centering
\includegraphics[width=0.235\textwidth]{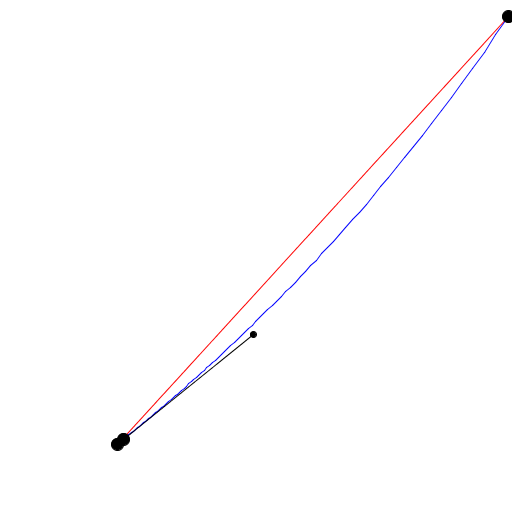}
\includegraphics[width=0.235\textwidth]{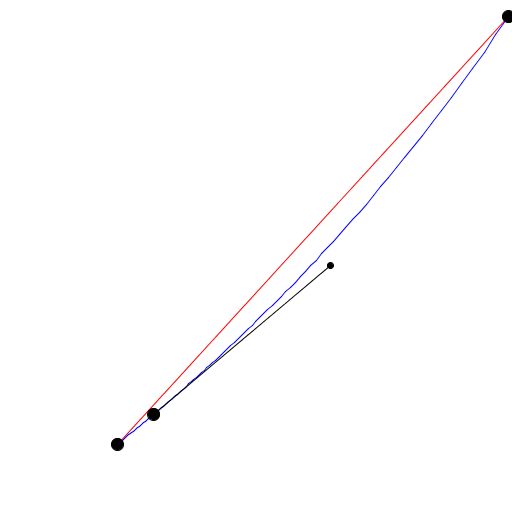}
\includegraphics[width=0.235\textwidth]{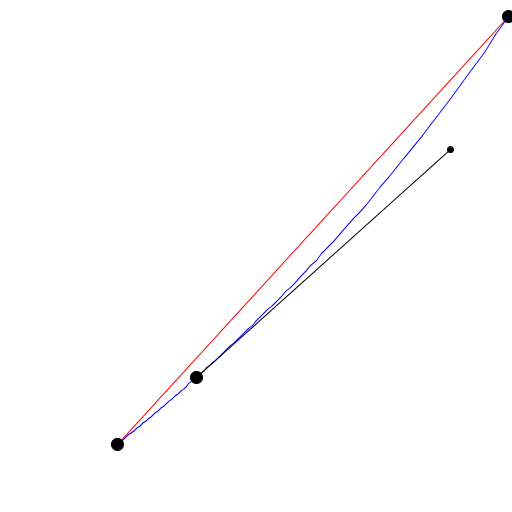}
\includegraphics[width=0.235\textwidth]{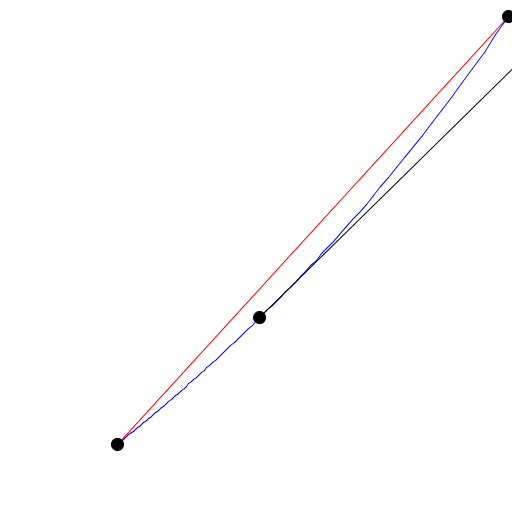}

\caption{Dual parallel transport of tangent vector $\dot\gamma_{pq}^*(0)$ at $T_p$ along the $\nabla^*$-geodesic: Since the dual geodesic
$\gamma_{pq}^*(t)$ is a $\nabla^*$-autoparallel curve, the tangent vector at $\dot\gamma_{pq}^*(0)$ is transported into tangent vectors along that dual geodesic.
 } 
\label{fig:PTex}
\end{figure}

Historically, the dual parallel transport coupled to the metric tensor was studied independently by Norden~\cite{Norden-1945} and Sen~\cite{Sen-1948} in the 1930's and 1940's, respectively.

\subsection{Dual Pythagorean theorems}\label{sec:pt}
The divergence $D_F(p:q)$ from a point $p$ to a point $q$ can be expressed using dual Bregman divergences as follows:
\begin{equation}
D_F(p:q)=B_F(\theta(p):\theta(q))=D^*(q:p)=B_{F^*}(\eta(q):\eta(p)),
\end{equation}
where $D^*_F(p:q)\eqdef D_F(q:p)$ denotes the reverse divergence, and 
\begin{equation}
B_F(\theta_1:\theta_2)=F(\theta_1)-F(\theta_2)-(\theta_1-\theta_2)^\top \nabla F(\theta_2),
\end{equation}
is the Bregman divergence associated to generator $F$.
Bregman divergences generalize the squared Euclidean distance with the relative entropy, and thus offer a nice framework to unify or develop generic algorithms~\cite{BD-2005}.
Bregman divergences proved useful in machine learning~\cite{BVD-2010} and topological data analysis~\cite{BD-TDA-2017} among others.

Now, consider two smooth curves $c_1(t)$ and $c_2(t)$ intersecting at a point $p=c_1(0)=c_2(0)$.
Denote by $\dot c_1(0)=\left.\frac{d}{\mathrm{d}t} c_1(t)\right|_{t=0}$ and $\dot c_2(0)=\left.\frac{d}{\mathrm{d}t} c_2(t)\right|_{t=0}$ the tangent vectors to the curves at point $p$, belonging to the tangent plane $T_p$.
Curve $c_1$ is said {\em orthogonal} to curve $c_2$ at $p$, i.e., $c_1(t)\perp_{p} c_2(t)$, iff. $g_p(\dot c_1(0),\dot c_2(0))=0$.

In a Bregman manifold, the remarkable following generalized Pythagorean theorem holds:

\begin{theorem}[Generalized Pythagorean theorem]\label{thm:pytha}
When a primal geodesic  $\gamma_{pq}$ is orthogonal to a dual geodesic $\gamma_{qr}^*$ at point $q$ (i.e., $\gamma_{pq} \perp_q \gamma_{qr}^*$),
we have $(\theta(p)-\theta(q))^\top (\eta(r)-\eta(q))=0$, and the following Pythagorean divergence identity holds:
\begin{equation}
D_F(p:q)+D_F(q:r)=D_F(p:r).
\end{equation}
\end{theorem}

\begin{proof}
The proof proceeds in two steps:
\begin{itemize}
\item First, Let us  show that the test $(\theta(p)-\theta(q))^\top (\eta(r)-\eta(q))=0$ can be rewritten as an inner product showing that
\begin{equation}
(\theta(p)-\theta(q))^\top (\eta(r)-\eta(q)) = g_q(v_{qp},v_{qr}^*)=g_q(\dot\gamma_{qp}(0),\dot\gamma_{qr}^*(0)),
\end{equation}
where $v_{qp}$ denotes the tangent vector at $q$ of the primal geodesic $\gamma_{pq}$, and $v_{qr}^*$ is the  tangent vector at $q$ of the dual geodesic $\gamma_{qr}^*$.
Indeed, the inner product is $g(v_{qp},v_{qr}^*)=(v_{qp})^i({v_{qr}^*})_i$ and
we have 
\begin{equation}
(\theta(p)-\theta(q))^\top (\eta(r)-\eta(q))=
 [\theta_{pq}]^\top \times \nabla^2F(\theta(q)) (\nabla^2F(\theta(q)))^{-1} \times [\eta_{rq}^*],
\end{equation}
where
$\theta_{pq} \eqdef \theta(p)-\theta(q)$ (contravariant components of $v_{qp}$) and $\eta_{rq}^*\eqdef\eta(r)-\eta(q)$ (covariant components of $v_{qr}^*$). That is, $[(v_{qp})^i]=\theta_{pq}$ and $[v_{qp}^*]_i=\eta_{pq}$.
Using the Crouzeix identity~\cite{Crouzeix-1977} of Eq.~\ref{eq:crouzeix}, we get
\begin{equation}
 (\theta(p)-\theta(q))^\top (\eta(r)-\eta(q))= [\theta_{pq}]^\top \nabla^2F(\theta(q)) (\nabla^2F^*(\eta(q))) [\eta_{rq}^*].
\end{equation}
The term $(\nabla^2F^*(\eta(q)))\times [\eta_{rq}^*]$ gives the contravariant components $\theta_{rq}^*$ of the vector $v_{qr}^*$.
Thus we have checked that 
\begin{equation}
(\theta(p)-\theta(q))^\top (\eta(r)-\eta(q))= [\theta_{pq}]^\top \times \nabla^2F(\theta(q))\times [\theta_{rq}^*] = g_q(v_{qp},v_{qr}^*).
\end{equation}
That is,  
\begin{equation}
(\theta(p)-\theta(q))^\top (\eta(r)-\eta(q))=0 \Leftrightarrow \dot\gamma_{pq}(0) \perp_q \dot\gamma_{qr}^*(0).
\end{equation}

\item Now, let us prove the Pythagorean identity when $\gamma_{pq} \perp_q \gamma_{qr}^*$.
We use the Bregman $3$-parameter identity~\cite{BD-2005} which generalizes the Euclidean law of cosines:
\begin{property}[Bregman $3$-parameter identity]
\begin{equation}\label{eq:bdcos}
B_F(\theta_1:\theta_2) = B_F(\theta_1:\theta_3) + B_F(\theta_3:\theta_2) - (\theta_1-\theta_3)^\top(\nabla F(\theta_2)-\nabla F(\theta_3)) \geq 0
\end{equation}
\end{property}

Instantiating this identity with $\theta_1=\theta(p)$, $\theta_2=\theta(r)$ and $\theta_3=\theta(q)$, and plugging the fact that
 $(\theta(p)-\theta(q))^\top (\nabla F(\theta(r))-\nabla F(\theta(q)))=0$, we get
\begin{eqnarray}
B_F(\theta(p):\theta(r)) &=& B_F(\theta(p):\theta(q))+B_F(\theta(q):\theta(r)),\\
D_F(p:r) &=& D_F(p:q) + D_F(q:r).
\end{eqnarray}
\end{itemize}
\end{proof}

Notice that the Bregman $3$-parameter identity can be proved by checking that the left-hand side equals the right-hand side of Eq.~\ref{eq:bdcos}.
Another direct proof consists in writing:
\begin{eqnarray}
B_F(\theta_1:\theta_3) &=& F(\theta_1)-F(\theta_3)-(\theta_1-\theta_3)^\top \nabla F(\theta_3),\\   \label{eq:p13}
B_F(\theta_3:\theta_2) &=& F(\theta_3)-F(\theta_2)-(\theta_3-\theta_2)^\top \nabla F(\theta_2),\\ \label{eq:p32} 
B_F(\theta_1:\theta_2) &=& F(\theta_1)-F(\theta_2)-(\theta_1-\theta_2)^\top \nabla F(\theta_2).  \label{eq:p12}
\end{eqnarray}
Adding Eq.~\ref{eq:p13} with Eq.~\ref{eq:p32} and subtracting Eq.~\ref{eq:p12} from them, we get:

\begin{equation}
B_F(\theta_1:\theta_3)+B_F(\theta_3:\theta_2)-B_F(\theta_1:\theta_2)=   (\theta_1-\theta_3)^\top(\nabla F(\theta_2)-\nabla F(\theta_3)), 
\end{equation}
from which it follows that
\begin{equation}
B_F(\theta_1:\theta_2)=B_F(\theta_1:\theta_3)+B_F(\theta_3:\theta_2)-(\theta_1-\theta_3)^\top(\nabla F(\theta_2)-\nabla F(\theta_3)). 
\end{equation}

More generally, we have
\begin{equation}
D_F(\gamma_{pq}(t):q)+D_F(q:\gamma_{qr}(t'))=D_F(\gamma_{pq}(t):\gamma_{qr}(t')),\quad \forall t,t'\in (0,1).
\end{equation}
That is, once we have a triple of points $(p,q,r)$ for which the generalized Pythagorean theorem holds, we can build an infinite number of such triple of points: $(\gamma_{pq}(t),q,\gamma_{qr}(t'))$ for  $t,t'\in (0,1)$.
Figure~\ref{fig:pytha} illustrates this view of the generalized Pythagorean theorem.

\begin{figure}
\centering
\includegraphics[width=0.95\textwidth]{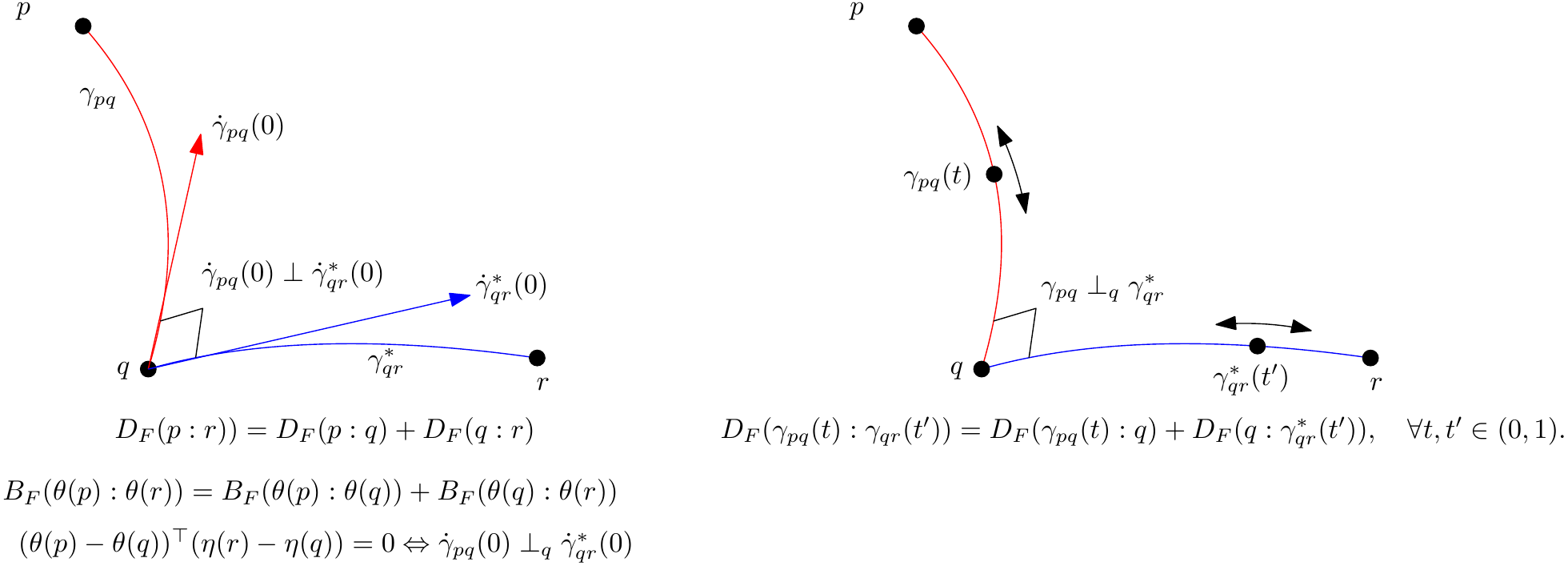}

\caption{The generalized Pythagorean theorem in a dually flat space.
Left: ordinary statement for dual-type geodesics $\gamma_{qp}$ and $\gamma_{qr}^*$ orthogonal at a point $q$. 
Right: Extended statement where the points can move along the dual geodesics $\gamma_{pq}(t)$ and $\gamma{qr}^*(t')$.
\label{fig:pytha}}
\end{figure}

Euclidean geometry with its ordinary Pythagoras' theorem is recovered as the special case of a self-dual Bregman manifold (when the dual potential functions coincide) induced by the Bregman generator $F_\Euc(\theta)=\frac{1}{2}\theta^\top\theta$.

Notice that the tangent vector $v_{pq}^*=\dot\gamma_{pq}^*(0)$ of the dual geodesic $\gamma_{pq}^*(t)$ is written using the covariant components in the reciprocal basis of $T_p$ as $\eta(q)-\eta(p)$. This tangent vector can be written equivalently using the contravariant coordinates
as $\nabla^2 F^*(\eta(p))\times [\eta(q)-\eta(p)]$.
That is, $[{v_{pq}^*}_i]=\eta(q)-\eta(p)$ and $[{v_{pq}^*}^i]=\nabla^2 F^*(\eta(p))\times [\eta(q)-\eta(p)]$.
Similarly, the tangent vector $v_{pq}=\dot\gamma_{pq}(0)$ of the primal geodesic $\gamma_{pq}(t)$ is written using the contravariant components in the basis of $T_p$ as 
 $\theta(q)-\theta(p)$. The covariant coordinates of this vector is $\nabla^2F(\theta(p))\times [\theta(q)-\theta(p)]$:
 $[{v_{pq}}^i]=\theta(q)-\theta(p)$ and $[{v_{pq}}_i]=\nabla^2 F(\theta(p))\times [\theta(q)-\theta(p)]$.
Figure~\ref{fig:tangentvectors} displays an example of two primal and dual geodesic arcs linking $p$ to $q$
 with the tangent vectors $\gamma_{pq}^*(0)$ at $p$ of the dual geodesic $\gamma_{pq}$.

\begin{figure}
\centering
\fbox{\includegraphics[width=0.45\textwidth]{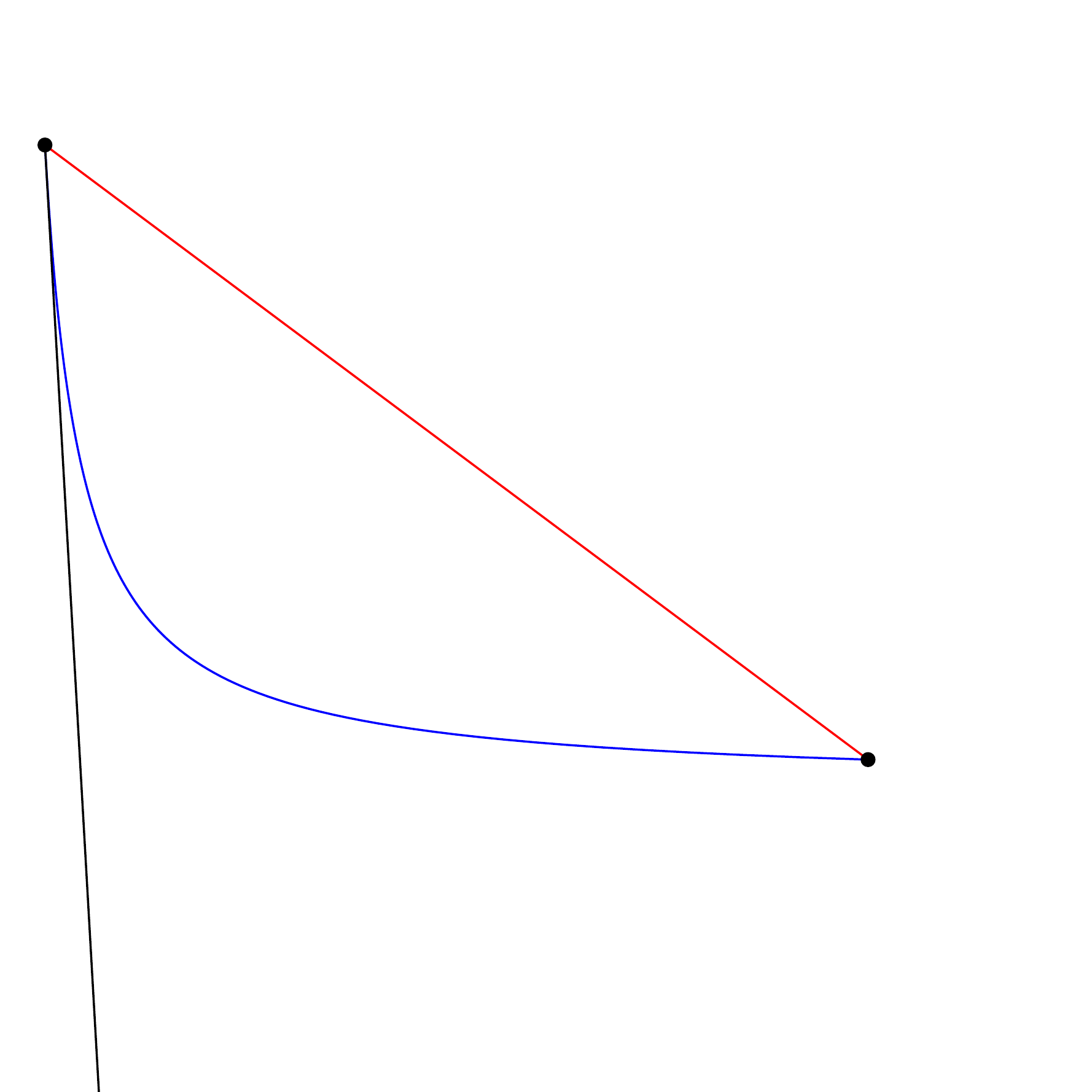}}\ \ 
\fbox{\includegraphics[width=0.45\textwidth]{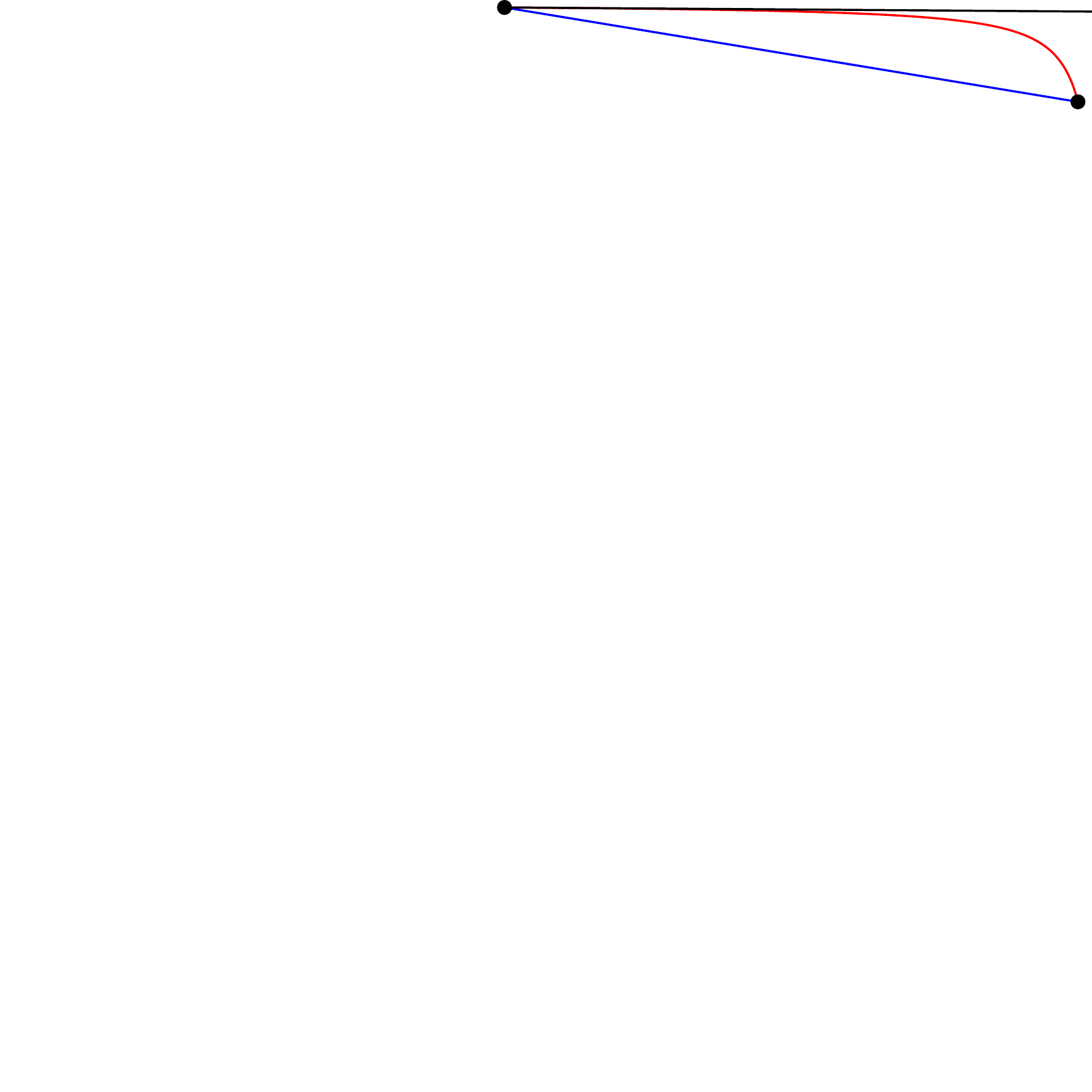}}

\caption{Two points $p$ and $q$ and their primal and dual geodesics visualized using the $\theta$-coordinates (left) and the $\eta$-coordinates (right). The tangent vector $v_{pq}^*\in T_p$ to $\gamma_{pq}^*$ is visualized in the $\theta$-coordinates (left, black line segment), and the tangent vector $v_{pq}\in T_p$ to $\gamma_{pq}$ is visualized in the $\eta$-coordinates (right, black line segment).
Here, we considered the Itakura-Saito manifold with
$\theta(p)=(0.17823054175936948,1.746348492830485)$ and $\theta(q)=(1.6105239241969733,0.6712015045558234)$.
\label{fig:tangentvectors}}
\end{figure}

When we exchange the role played by points $p$ and $r$ in the above Pythagorean theorem, we obtain the following {\em dual} Pythagorean theorem:

\begin{theorem}[Dual Pythagorean theorem]
When a dual geodesic  $\gamma_{pq}^*$ is orthogonal to a primal geodesic $\gamma_{qr}$ at point $q$ ( $\gamma_{pq}^* \perp_q \gamma_{qr}$),
we have $(\eta(p)-\eta(q))^\top (\theta(r)-\theta(q))=0$ and the following divergence identity:
\begin{equation}
D^*_F(p:q)+D^*_F(q:r)=D^*_F(p:r),
\end{equation}
or equivalently
\begin{equation}
D_F(r:q)+D_F(q:p)=D_F(r:p).
\end{equation}
\end{theorem}

Notice that we can write the Bregman generalized law of cosines of Eq.~\ref{eq:bdcos} geometrically (i.e., without relying on any prescribed coordinate system) as follows:
\begin{eqnarray}
D_F(p:q) &=& D_F(p:r)+D_F(r:q) - g_r(\dot\gamma_{rp}(0),\dot\gamma_{rq}^*(0) )  \geq 0,\\
&=& D_F(p:r)+D_F(r:q) - \| \dot\gamma_{rp}(0)\|_r \  \| \dot\gamma_{rq}^*(0)\|_r \cos\left(\alpha_r(\dot\gamma_{rp},\dot\gamma_{rq}^*)\right).
\end{eqnarray}

It follows that 
\begin{itemize}
	\item when  $\alpha_r(\dot\gamma_{rp},\dot\gamma_{rq}^*)\in \left[0,\frac{\pi}{2}\right)$ (acute angle), we have $D_F(p:q) < D_F(p:r)+D_F(r:q)$,
	
	\item when  $\alpha_r(\dot\gamma_{rp},\dot\gamma_{rq}^*) = \frac{\pi}{2}$ (right angle), we have $D_F(p:q) = D_F(p:r)+D_F(r:q)$, and
	
	\item when $\alpha_r(\dot\gamma_{rp},\dot\gamma_{rq}^*)\in \left(\frac{\pi}{2},\pi\right)$ (obtuse angle), we have $D_F(p:q) > D_F(p:r)+D_F(r:q)$.
\end{itemize}

Notice that when $F(\theta)=F_\Euc(\theta)=\frac{1}{2}\theta^\top\theta$ (Euclidean geometry) and $\theta(r)=0$, we get

\begin{equation}\label{eq:lc}
\frac{1}{2} \|p-q\|^2 = \frac{1}{2} \|p\|^2 + \frac{1}{2} \|q\|^2 - \|p\| \|q\| \cos \alpha_r(p,q). 
\end{equation}
Multiplying by two both sides of Eq.~\ref{eq:lc}, we recover the usual law of cosines of Euclidean geometry illustrated in Figure~\ref{fig:lawcosines} with $a=\|p\|$, $b=\|q\|$, $c=\|p-q\|$, and $C=\alpha_r([rp],[rq])$:
\begin{equation}
c^2 = a^2 +b^2 - 2ab \cos C. 
\end{equation}

\begin{figure}
\centering
\includegraphics[width=0.6\textwidth]{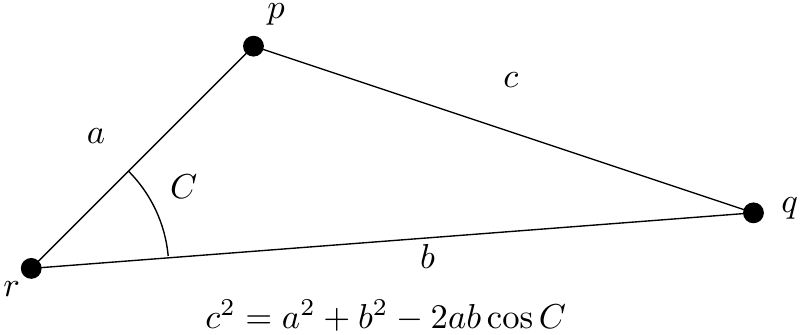}
\caption{The Euclidean law of cosines.\label{fig:lawcosines}}
\end{figure}

We also have the following identity illustrated in Figure~\ref{fig:BregmanParallelogram}:
\begin{eqnarray}
B_F(\theta_1:\theta)+B_F(\theta_2:\theta) &=& B_F\left(\theta_1:\frac{\theta_1+\theta_2}{2}\right) +
 B_F\left(\theta_2:\frac{\theta_1+\theta_2}{2}\right) + 2 B_F\left(\frac{\theta_1+\theta_2}{2}:\theta\right),\label{eq:bdpar}\\
B_F(\theta_1:\theta) + B_F(\theta_2:\theta) &=& 2 \JB_F(\theta_1,\theta_2) + 2 B_F\left(\frac{\theta_1+\theta_2}{2}:\theta\right),\\
B_F(\theta_1:\theta)+B_F(\theta_2:\theta) &=& 2 J_F(\theta_1,\theta_2) + 2 B_F\left(\frac{\theta_1+\theta_2}{2}:\theta\right),
\end{eqnarray}
where $\JB_F$ and $J_F$ denote the {\em Jensen-Bregman divergence}~\cite{JBD-2011} and the {\em Jensen divergence}~\cite{BR-2011}, respectively:
\begin{eqnarray}
\JB_F(\theta_1,\theta_2) &:=& \frac{1}{2}\left( B_F\left(\theta_1:\frac{\theta_1+\theta_2}{2}\right)  +  B_F\left(\theta_2:\frac{\theta_1+\theta_2}{2}\right) \right),\\
&=& \frac{F(\theta_1)+F(\theta_2)}{2} -F\left(\frac{\theta_1+\theta_2}{2}\right)  =: J_F(\theta_1,\theta_2).
\end{eqnarray}

The  family of categorical distributions forms a mixture family~\cite{geowmixtures-2018} in information geometry for the Shannon negentropy generator $F_\Shannon$, 
and we have $B_{F_\Shannon}(\theta_1:\theta_2)=\KL(p_{\theta_1}:p_{\theta_2})$.
It follows that for any three categorical distributions $p$, $q$, and $r$, we have
\begin{equation}
\KL(p:r) + \KL(q:r) =2\JS(p,q)+ 2\KL\left(\frac{p+q}{2}:r\right),
\end{equation}
since $\frac{p+q}{2}$ is  a categorical distribution.
In general, the identity does not hold for members of an exponential family (i.e., Gaussian family) since the mixture density $\frac{p+q}{2}$ does not belong to the exponential family.
The categorical family is a very special case of family that is {\em both} an exponential family and a mixture family~\cite{IG-2016}.

Since there exists a bijection between regular exponential families and regular Bregman divergences~\cite{BD-2005}, we can state the parallelogram-type identity of
Eq.~\ref{eq:bdpar} for parametric densities $\{p_\theta\}$ belonging to the same exponential family~\cite{EF-2009} (with cumulant function $F$) with respect to both the Kullback-Leibler divergence and the {\em Bhattacharyya divergence}~\cite{BR-2011}
\begin{equation}
\Bhat(p,q) \eqdef -\log\int \sqrt{p(x)q(x)} \dmu(x),
\end{equation}
 as:
\begin{equation}\label{eq:parallelogramEF}
\KL(p_\theta:p_{\theta_1})+\KL(p_\theta:p_{\theta_2}) = 2\Bhat(p_{\theta_1},p_{\theta_2}) + 2\KL\left(p_{\frac{\theta_1+\theta_2}{2}}:p_\theta\right),
\end{equation}
since $\KL(p_\theta:p_{\theta'})=B_F(\theta':\theta)$ and $\Bhat(p_\theta:p_{\theta'})=J_F(\theta,\theta')$ for densities $p_\theta$ 
and $p_{\theta'}$ belonging to the same exponential family.
For example, the identity of Eq.~\ref{eq:parallelogramEF} applies to any two densities of the Gaussian family.

\begin{figure}
\centering
\includegraphics[width=0.8\textwidth]{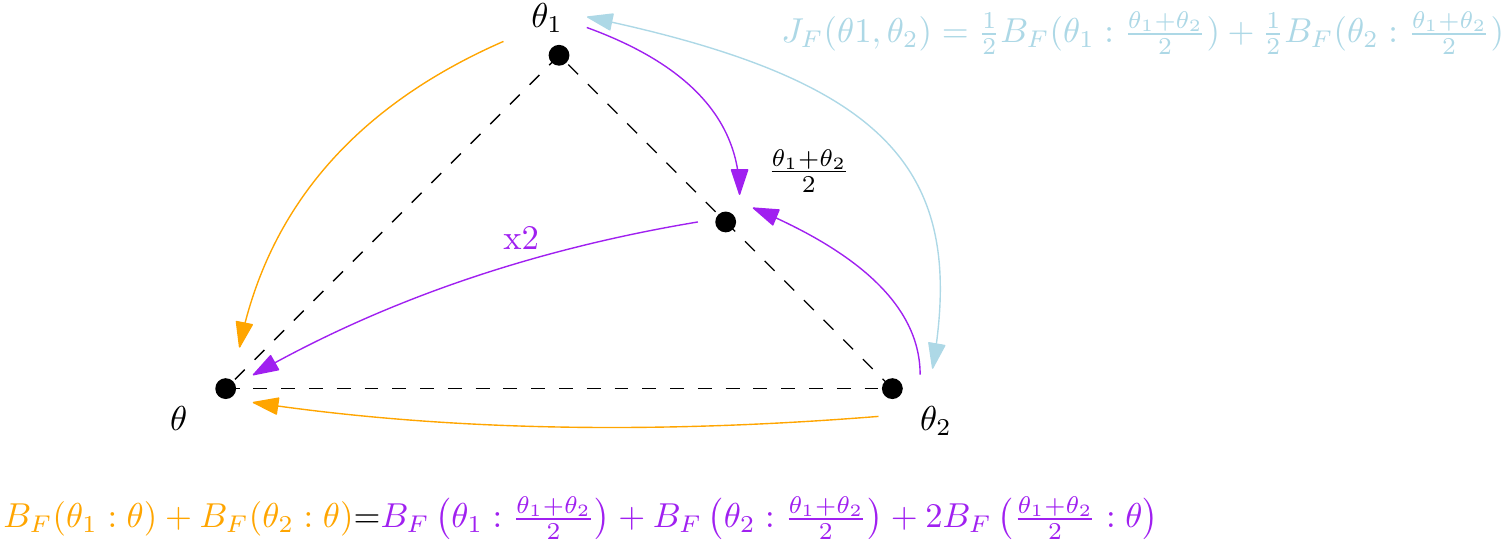}
\caption{A parallelogram-type identity for Bregman divergences.\label{fig:BregmanParallelogram}}
\end{figure}

The $3$-parameter property of Bregman divergences is a particular instance of the following {\em $4$-parameter property}~\cite{Bregman4pt-1997} (a quadrilateral relation):  

\begin{property}[Bregman $4$-parameter property]
For any four points $p_1$, $p_2$, $q_1$, $q_2$, we have the following identity:
\begin{eqnarray}
B_{F}(\theta(p_1) : \theta(q_1)) +  B_{F}(\theta(p_2) : \theta(q_2)) 
 && - B_{F}(\theta(p_1) : \theta(q_2)) - B_{F}(\theta(p_2) : \theta(q_1))  \nonumber\\
&& -(\theta(p_2)-\theta(p_1))^\top (\eta(q_1)-\eta(q_2))=0.
\end{eqnarray}
\end{property}

Indeed, let $q_1=p_1$. Then we recove:r 
\begin{eqnarray}
B_{F}(\theta(p_2) : \theta(q_2)) &=& B_{F}(\theta(p_1) : \theta(q_2)) + B_{F}(\theta(p_2) : \theta(q_1))
 -(\theta(p_2)-\theta(p_1))^\top (\eta(q_2)-\eta(p_1)),\\
D_F(p_2:q_2) &=&  D_F(p_2:p_1) + D_F(p_1:q_2)  - g_{p_1}(\dot\gamma_{p_1p_2}(0), \dot\gamma^*_{p_1q_2}(0)).
\end{eqnarray}
Similarly, let $q_2=p_2$. Then we get:
\begin{eqnarray}
B_{F}(\theta(p_1) : \theta(q_1)) &=& B_{F}(\theta(p_1) : \theta(p_2)) + B_{F}(\theta(p_2) : \theta(q_1))
-(\theta(p_1)-\theta(p_2))^\top (\eta(q_1)-\eta(p_2)),\\
D_F(p_1:q_1)&=&D_F(p_1:p_2)+D_F(p_2:q_1) -g_{p_2}(\dot\gamma_{p_2p_1}(0),\dot\gamma^*_{p_2q_1}(0)).
\end{eqnarray}

To derive the $4$-parameter identity, we use twice the $3$-parameter identity as follows:
\begin{eqnarray}
B_{F}(\theta_1 : \theta_3) &=& B_F(\theta_1:\theta_2) + B_F(\theta_2:\theta_3) -(\theta_3-\theta_2)^\top (\nabla F(\theta_1)-\nabla F(\theta_2)), \label{eq:four1} \\
B_{F}(\theta_4 : \theta_3) &=& B_F(\theta_4:\theta_2) + B_F(\theta_2:\theta_3) -(\theta_3-\theta_2)^\top (\nabla F(\theta_4)-\nabla F(\theta_2)). \label{eq:four2}
\end{eqnarray}
Indeed, subtracting Eq.~ \ref{eq:four2} from Eq.~\ref{eq:four1}, we get

\begin{equation}
B_{F}(\theta_1 : \theta_3) - B_{F}(\theta_4 : \theta_3) = B_F(\theta_1:\theta_2)-B_F(\theta_4:\theta_2) - (\theta_3-\theta_2)^\top
(\nabla F(\theta_1)  - \nabla F(\theta_4)).
\end{equation}

We can interpret geometrically  the $4$-parameter identity as follows:
\begin{equation}
D_F(p_1:p_3) - D_F(p_4:p_3) = D_F(p_1:p_2) - D_F(p_4:p_2) - g_{p_2}(\dot\gamma_{p_2p_3}(0),\dot\gamma_{p_2p_1}^*(0)) 
+ g_{p_2}(\dot\gamma_{p_2p_3}(0),,\dot\gamma_{p_2p_4}^*(0)). 
\end{equation}

Note that the Bregman $4$-parameter identity is said parallelogram-type, because if we choose $F_\Euc(\theta)=\frac{1}{2}\theta^\top\theta=\frac{1}{2}\|\theta\|_2^2$, then we have $B_{F_\Euc}(\theta:\theta')=\frac{1}{2}\|\theta-\theta'\|^2$, and the $4$-parameter identity becomes for $\theta=0$:
\begin{eqnarray}
\frac{1}{2}\|\theta_1\|^2+\frac{1}{2} \|\theta_2\|^2 &=& \left\|\frac{\theta_1-\theta_2}{2}\right\|^2 + \left\|\frac{\theta_1+\theta_2}{2}\right\|^2,\\
2\|\theta_1\|^2 + 2\|\theta_2\|^2 &=& \|\theta_1-\theta_2\|^2 + \|\theta_1+\theta_2\|^2,
\end{eqnarray}
the usual parallelogram identity for the $L_2$-normed space.

Finally, let us mention some special submanifolds in dually flat spaces:
A $k$-dimensional submanifold $S\subset M$ of a dually flat manifold $M$ is said $\nabla$-autoparallel~\cite{IG-2016} ($\nabla^*$-autoparallel) iff. the sets of points of $S$ can be described by a $k$-dimensional affine flat in the $\theta$-coordinate system ($\eta$-coordinate system, respectively).
Primal $\gamma_{pq}$ geodesics are 1D $\nabla$-autoparallel submanifolds and dual $\gamma_{pq}^*$ geodesics are 1D $\nabla^*$-autoparallel submanifolds.
Notice that the intersection of $\theta$-flats is a $\theta$-flat, but the intersection of a $\theta$-flat with a $\eta$-flat is in general neither a $\theta$-flat nor a $\eta$-flat.

In a Bregman manifold, we can also express the divergence between two points $p$ and $q$ using the {\em Fenchel-Young divergence} $A_F$ (also called the canonical divergence of dually flat spaces):
\begin{equation}
D_F(p:q)=A_F(\theta(p):\eta(q)) \eqdef F((\theta(p)) + F^*(\eta(q)) - \theta(p)^\top \eta(q).
\end{equation}

Thus we have
\begin{equation}
D_F(p:q)=B_F(\theta(p):\theta(q))=A_F(\theta(p):\eta(q))=A_{F^*}(\eta(q):\theta(p))=B_{F*}(\eta(q):\eta(p))=D_{F^*}(q:p),
\end{equation}
and $D_{F^*}(q:p)={D_F}^*(q:p)=D_F(p:q)$ (where ${D_F}^*$ is the reverse divergence).

\subsection{Some examples of Bregman manifolds}\label{sec:BMex}

We describe concisely the Mahalanobis manifolds in \S\ref{sec:Mah}, the extended Kullback-Leibler manifold in \S\ref{sec:KL}, 
  the Itakura-Saito manifold in \S\ref{sec:IS}, and the multinoulli manifolds   (\S\ref{sec:Multinoulli}).

\subsubsection{The Mahalanobis manifolds}\label{sec:Mah}
Consider the case where the Bregman generator is defined by
\begin{equation}
F_Q(\theta)=\frac{1}{2}\theta^\top Q\theta,
\end{equation}
for a prescribed symmetric positive-definite $D\times D$ matrix $Q\succ 0$.
The Legendre-Fenchel convex conjugate~\cite{LegendreIG-2010,CR-IG-2013} is 
\begin{equation}
F^*(\eta)=\frac{1}{2}\eta^\top Q^{-1}\eta=F_{Q^{-1}}(\eta),
\end{equation} 
and  it follows that 
\begin{equation}
\eta(\theta)=\nabla F_Q(\theta)= Q\theta, \quad \theta(\eta)=\nabla F_Q^*(\eta)=Q^{-1}\eta.
\end{equation}
The dual Riemannian metrics are 

\begin{eqnarray}
\left[g_{ij}\right] &=& \nabla^2 F_Q(\theta)=Q,\\
\left[{g^*}^{ij}\right] &=& \nabla^2 F_Q^*(\eta)=Q^{-1}.
\end{eqnarray}

The dual geodesics  $\gamma_{pq}$ and $\gamma_{qp}^*$ in a Mahalanobis manifold coincide. 
The dual Bregman divergences are squared Mahalanobis distances~\cite{Mahalanobis-1936,CurvedMahalanobis-2016}:
\begin{equation}
B_{F_Q}(\theta_1:\theta_2)=\frac{1}{2}(\theta_1-\theta_2)^\top Q (\theta_1-\theta_2),
\end{equation}
and 
\begin{equation}
B_{F^*_Q}(\eta_1:\eta_2)=\frac{1}{2}(\eta_1-\eta_2)^\top Q^{-1} (\eta_1-\eta_2).
\end{equation}
We check that 
\begin{equation}
B_{F^*_Q}(\eta_2:\eta_1)=  \frac{1}{2} (Q(\theta_2-\theta_1))^\top Q^{-1} Q(\theta_2-\theta_1)= B_{F_Q}(\theta_1:\theta_2),
\end{equation}
since $Q^\top=Q$.
The squared Mahalanobis Bregman divergences are provably the only symmetric Bregman divergences~\cite{BVD-2010}.
In particular, when $Q=I$, the identity matrix, the dual Bregman divergences $B_F$ and $B_F^*$ coincide with half of the squared Euclidean distance $D_E(\theta_1,\theta_2)=\frac{1}{2} (\theta_2-\theta_1)^\top (\theta_2-\theta_1)$.
By using the Cholesky decomposition of $Q$, i.e., $Q=LL^\top$ where $L$ is a lower triangular matrix with positive diagonal entries,
 we have $B_{F_Q}(\theta_1:\theta_2)=B_{F_I}(L^\top\theta_1:L^\top\theta_2)=\|L^\top(\theta_1-\theta_2)\|^2$, where $I$ denotes the identity matrix. 
Notice that the squared Mahalanobis divergences $B_{F_Q}$ are non-separable Bregman divergences whenever $Q$ has non-zero off-diagonal elements.

\subsubsection{The Kullback-Leibler manifold}\label{sec:KL}
The extended Shannon negative entropy~\cite{BR-2011}
\begin{equation}
F_\KL(\theta)=\sum_{i=1}^D \theta^i\log\theta^i-\theta^i
\end{equation} 
is a strictly convex and $C^3$ function on $\Theta=\mathbb{R}_{++}^D$, i.e., a separable Bregman generator.
Here, we consider the positive orthant domain instead of the probability simplex hence the name extended Shannon negentropy.
We have the following conversion formula between the primal and dual coordinates: 
\begin{equation}
\eta(\theta)=\nabla F_\KL(\theta)=[\log\theta^i],\quad \theta(\eta)=\nabla F_\KL^*(\eta)=[\exp\eta^i].
\end{equation}
The Legendre convex conjugate is 
\begin{equation}
F^*_\KL(\eta)= \theta(\eta)^\top \eta -F(\theta(\eta))=\sum_{i=1}^D \exp(\eta^i).
\end{equation} 
The dual Riemannian metrics are
\begin{equation}
[g_{ij}]=\nabla^2 F_\KL(\theta)=\diag\left(\frac{1}{\theta^1}, \ldots, \frac{1}{\theta^D}\right), \quad 
[{g^*}^{ij}]=\nabla^2 F_\KL^*(\eta)=\diag\left(\exp(\eta^1), \ldots, \exp(\eta^D)\right).
\end{equation}
The dual Bregman divergences are 
\begin{eqnarray}
B_{F_\KL}(\theta_1:\theta_2) &=& \sum_{i=1}^D \theta_1^i\log\frac{\theta_1^i}{\theta_2^i}+ \theta_2^i-\theta_1^i,\\
B_{F_\KL^*}(\eta_1:\eta_2) &=&\sum_{i=1}^D  \exp(\eta_1^i)-\exp(\eta_2^i)-(\eta_1^i-\eta_2^i) \exp(\eta_2^i)=B_{F_\KL}(\theta(\eta_2):\theta(\eta_1)).
\end{eqnarray}

\subsubsection{Itakura-Saito manifold}\label{sec:IS}
The $D$-dimensional Burg information (i.e., Burg negative entropy~\cite{BDcentroid-2009,BR-2011}) is defined by the following separable convex generator: 
\begin{equation}
F_\IS(\theta)=-\sum_{i=1}^D \log(\theta^i).
\end{equation}
We have 
\begin{equation}
\eta(\theta)=\nabla F(\theta)=\left[-\frac{1}{\theta^i}\right]_i,
\end{equation}
 and the  Bregman divergence is called the Itakura-Saito divergence~\cite{BVD-2010}:
\begin{equation}
B_{F_\IS}(\theta_1:\theta_2)= \sum_{i=1}^D \frac{\theta_1^i}{\theta_2^i}-\log \frac{\theta_1^i}{\theta_2^i}-1.
\end{equation}
We convex $\eta$-coordinates to $\theta$-coordinates as follows:
\begin{equation}
\theta(\eta)= \left[-\frac{1}{\eta^i}\right]_i,
\end{equation}
and the dual Bregman generator is 
\begin{equation}
F^*_\IS(\eta)=\eta^\top \theta(\eta)-F_\IS(\theta(\eta))=-D+\sum_{i=1}^D \log\left(-\frac{1}{\eta^i}\right)=-D-\sum_{i=1}^D \log(-\eta_i),
\end{equation}
where $\eta\in\bbR_{--}^D$.
The dual Bregman divergence is 
\begin{equation}
B_{{F_\IS}^*}(\eta_1:\eta_2)= \sum_{i=1}^D \frac{\eta_1^i}{\eta_2^i}-\log \frac{\eta_1^i}{\eta_2^i}-1.
\end{equation}
We check that
\begin{equation}
B_{{F_\IS}^*}(\eta_2:\eta_1) = B_{F_\IS}(\theta_1:\theta_2).
\end{equation}
The dual Riemannian metric tensors are  
\begin{equation}
[g_{ij}]=\nabla^2 F_\IS(\theta) = \diag\left(\frac{1}{\sqr(\theta^1)}, \ldots, \frac{1}{\sqr(\theta^D)}\right),\quad 
[{g^*}^{ij}]=\nabla^2 F_\IS^*(\eta)=\diag\left(\frac{1}{\sqr(\eta^1)}, \ldots, \frac{1}{\sqr(\eta^D)}\right).
\end{equation}

\subsubsection{The multinoulli manifolds}\label{sec:Multinoulli}

We can build a dually flat space from any strictly convex and $C^3$ convex function $F$~\cite{EIG-2018} which also defines a Bregman generator. 
In particular, we can use the {\em log-normalizer} (also called cumulant function or log-partition function) of a regular exponential family~\cite{BD-2005} as such a Bregman generator~\cite{MC-2018}.
Let us consider the multinoulli family (i.e., the multinomial family for a single trial also called the family of categorical distributions).
The probability of a multinoulli distribution with $d$ categories $c_1,\ldots, c_d$ such that $\Pr(x=c_i)=\lambda_i$  is:
\begin{equation}
\Pr(x) = \lambda_1^{x_1}\times\ldots \times \lambda_d^{x_d},
\end{equation}
with $x_i\in\{0,1\}$ and $\sum_{i=1}^d x_i=1$.
Let us write the multinoulli probability mass function in the canonical form of an exponential family as
\begin{eqnarray}
\lambda_1^{x_1}\times\ldots \times \lambda_d^{x_d} &=& \exp\left(\sum_{i=1}^d x_i\log\lambda_i\right),\\
&=& \exp\left(\sum_{i=1}^{d-1} x_i\log\lambda_i + \underbrace{(1-\sum_{i=1}^{d-1}x_i)}_{x_d}\log {\lambda_d} \right),\\
&=&\exp\left(\sum_{i=1}^{d-1} x_i\theta_i-F(\theta)\right),\\
\end{eqnarray}
with the natural parameter $\theta_i=\log\frac{\lambda_i}{\lambda_d}$ for $i\in\{1,\ldots,d-1\}$.
The multinoulli distribution is an exponential family of order $D=d-1$.
The natural parameter space $\Theta$ is $\bbR^{D}=\bbR^{d-1}$.
We can convert back the natural parameter $\theta$ to the original parameters $\lambda$ as follows:
\begin{equation}
\lambda(\theta)=
\left\{
\begin{array}{ll}
 \lambda_i = \frac{\exp(\theta_i)}{1+\sum_{i=1}^D \exp(\theta_i)}, &  \mbox{for $i\in\{1,\ldots,D\}$} \\
 \lambda_d= \frac{1}{1+\sum_{i=1}^D \exp(\theta_i)}. & 
\end{array}
\right.
\end{equation}

The log-normalizer is
\begin{equation}
F_\Mult(\theta)=  -\log\lambda_d = \log\left(1+\sum_{i=1}^D \exp(\theta_i)\right).
\end{equation}

The gradient of the log-normalizer is
\begin{equation}
\eta(\theta)=\nabla F_\Mult(\theta)= \left[ \frac{\exp(\theta_i)}{1+\sum_{i=1}^D \exp(\theta_i)} \right]_i,
\end{equation}
and the reciprocal gradient is
\begin{equation}
\theta(\eta)=\nabla G_\Mult(\eta)= \left[ \log\frac{\eta_i}{1-\sum_{i=1}^D}\eta_i \right]_i.
\end{equation}
The convex conjugate of $F$ is
\begin{equation}
G_\Mult(\eta)= \eta^\top\theta(\eta)-F(\theta(\eta))=\sum_i \eta_i\log\eta_i+\left(1-\sum_i\eta_i\right)\log\left(1-\sum_i\eta_i\right).
\end{equation}

It follows that the Riemannian metric and dual Riemannian metric are Hessians of the potential functions $F_\Mult(\theta)$ and $G_\Mult(\eta)$, respectively:
\begin{eqnarray}
\left[\nabla^2 F_\Mult(\theta)\right]_{ij} &=& 
\left\{
\begin{array}{ll}
 -\frac{\exp(\theta_i+\theta_j)}{\left(1+\sum_i \exp(\theta_i)\right)^2}  & \mbox{if $i\not=j$},\\
\frac{\exp(\theta_i)}{\left(1+\sum_i \exp(\theta_i)\right)} -\frac{\exp(2\theta_i)}{\left(1+\sum_i \exp(\theta_i)\right)^2}	& \mbox{if $i=j$}
\end{array}
\right.
\\
\nabla^2 G_\Mult(\eta) &=& \frac{1}{1-\sum_i \eta_i} 1_M + \diag\left(\frac{1}{\eta_1},\ldots,\frac{1}{\eta_D}\right)
=\left\{
\begin{array}{ll}
 \frac{1}{1-\sum_i \eta_i}  & \mbox{if $i\not=j$},\\
 \frac{1}{1-\sum_i \eta_i}+\frac{1}{\eta_i} & \mbox{if $i=j$}
\end{array}
\right.,
\end{eqnarray}
where $1_M$ denotes the $D\times D$-matrix with all entries equal to $1$, and $\diag(x_1,\ldots,x_D)$ the diagonal matrix with diagonal elements $x_1, \ldots, x_D$.

The dual Bregman divergences are
\begin{eqnarray}
B_{F_\Mult}(\theta_1:\theta_2) &=& \KL^*(p_{\lambda(\theta_1)} : p_{\lambda(\theta_2)} ) = \KL(p_{\lambda(\theta_2)} : p_{\lambda(\theta_1)} ), \\
B_{G_\Mult}(\eta_1:\eta_2)&=& \KL(p_{\lambda(\eta_1)} : p_{\lambda(\eta_2)} ), \\
\end{eqnarray}
where $\KL$ is the discrete Kullback-Leibler divergence and $\KL^*$ the reverse Kullback-Leibler divergence.
We   convert  the expectation parameter $\eta$ to the original parameter $\lambda$ as follows:
\begin{equation}
\lambda(\eta)=
\left\{
\begin{array}{ll}
 \lambda_i = \eta_i, &  \mbox{for $i\in\{1,\ldots,D\}$} \\
 \lambda_d= 1-\sum_{i=1}^D \eta_i. & 
\end{array}
\right.
\end{equation}

In general, the Bregman divergence induced by log-normalizer of an exponential family amounts to a reverse Kullback-Leibler divergence~\cite{MC-2018}. 
Notice that the non-separable log-normalizer of a multinomial family yields a nice example of a non-separable Bregman divergence.
The multinoulli manifold is called the {\em Bernoulli manifold} when $D=1$ (i.e., $d=2$), and the trinoulli manifold when $D=2$ (i.e., $d=3$).

\subsection{Bregman balls and Bregman spheres}\label{sec:BBall}

The space of Bregman spheres and Bregman balls was investigated in~\cite{BVD-2010}.
Here, we report the parametric equations of Bregman spheres which allow one to draw them.
Let us consider a {\em separable} (Bregman) distance $D(\theta_1:\theta_2)$ between two $d$-dimensional vectors 
$\theta_1=(\theta_1^1,\ldots,\theta_1^d)$ and $\theta_2=(\theta_2^1,\ldots,\theta_2^d)$:
$$
D(\theta_1:\theta_2)=\sum_{i=1}^d D_i(\theta_1^i:\theta_2^i),
$$
where the $D_i(\cdot,\cdot)$'s denote the  scalar distances.
To avoid confusion with the distance $D$, the dimension of the spaces is denoted by $d$ in this section.
For example, we may consider the discrete $f$-divergences~\cite{EN-PhD-Csiszar-1967} which includes the Kullback-Leibler divergence, or the Bregman divergences~\cite{Bregman-1967} which includes the extended Kullback-Leibler divergence  and the    Itakura-Saito divergence.
The Kullback-Leibler divergence between two densities of an exponential family amounts to a Bregman divergence~\cite{KLBD-2001}.

Let us define a {\em $D$-sphere} of center $\theta_c$ and radius $r\geq 0$ as follows:  
$$
\sphere_D(\theta_c,r):=\{\theta\ :\ D(\theta_c:\theta)=r\}.
$$

In practice, we can visualize a $D$-sphere by drawing the {\em implicit function} $f(\theta):=D(\theta_c:\theta)-r$.
Plotting packages such as {\tt gnuplot}\footnote{\url{http://www.gnuplot.info/}} use either a marching cube technique or a fast quadtree approximation algorithm to rasterize these implicit functions.

Here, we are concerned with reporting parametric equations of $D$-spheres for separable distances $D$.
Let us characterize the coordinates $\theta^i$ of a point $\theta$ belonging to a  $D$-sphere $\sphere_D(\theta_c,r)$ as follows:

$$
\left\{
\begin{array}{lll}
D_1(\theta_c^1:\theta^1) &=& u_1\geq 0,\\
\ldots\\
D_d(\theta_c^d:\theta^d) &=& u_d\geq 0,\\
u_1+\ldots+u_d &=& r
\end{array}
\right.
$$

For $f$-divergences and Bregman divergences, the equation $D(c,x)=u$ for a scalar distance
 $D$ admits two solutions $x_{D,c,u;1}\geq u$ and 
$x_{D,c,u;-1}\leq u$ for $r>0$, where $\pm 1$ denotes the sign to localize the interval endpoint with respect to the center $c$.
It follows that the parametric equation of a $d$-dimensional $D$-sphere can be written independently on the $2^d$ orthants.
For $o\in \{-1,1\}\times \{-1,1\}\subset\bbR^d$, let $O_o$ denote the orthant
$$
O_o=\left\{\theta \ :\ \theta^i\geq \sign(o^i)\theta_c^i\right\}.
$$

Then the equation of the $D$-sphere on the orthant $O_o$ is

\begin{eqnarray*}
\left(x_{D_1,\theta_c^1,u_1;\sign(o^1)(o)},\ldots, x_{D_d,\theta_c^d,u_d;\sign(o^d)(o)}\right),\\ 
u_1\in[0,r], u_2\in [0,r-u_1],\ldots, u_{d}\in \left[0, r-\sum_{i=1}^{d-2} u_i\right].
\end{eqnarray*}

Thus if we know the parametric equations of a $D$-sphere on two opposite orthants, we can deduce the full parameteric equation of the $D$-sphere.
Notices that all axis-parallel boxes 
$$
B=\prod_{i=1}^d \left[x_{D_i,\theta_c^i,u_i;-1},x_{D_i,\theta_c^1,u_i;1}\right]
$$
are tangent at the $D$-sphere at its $2^d$ corners.
Figure~\ref{fig:tangentbox} illustrates this property for a 2D extended Kullback-Leibler circle (left) and a 2D Itakura-Saito circle (right).

\begin{figure}
\centering
\includegraphics[width=0.45\textwidth]{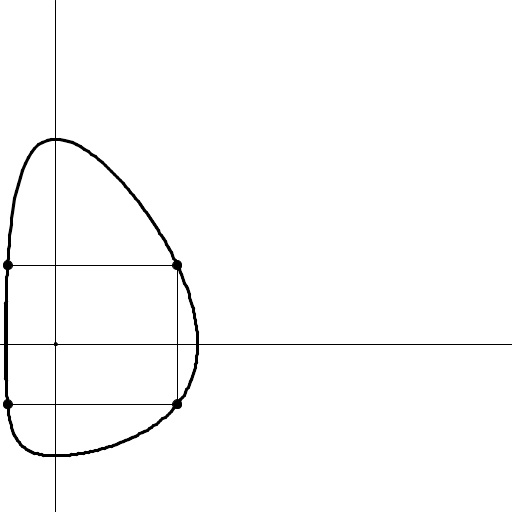}
\includegraphics[width=0.45\textwidth]{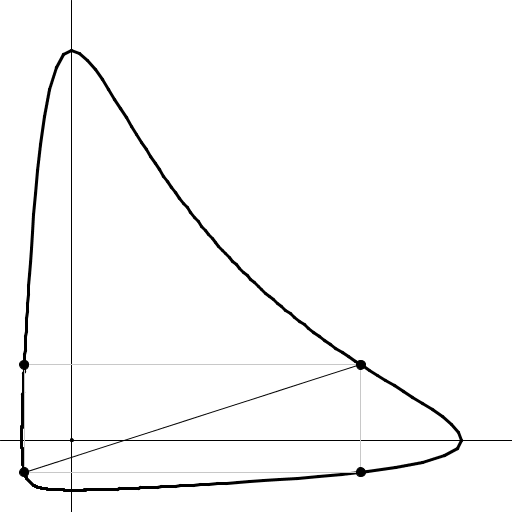}

\caption{$D$-spheres with respect to separable distances shown with a random box tangent at its vertices: To plot a $D$-sphere, we only need to know the parametric equations on two opposite orthants to deduce the other orthants.
\label{fig:tangentbox}}
\end{figure}

We now illustrate how to get the parametric equations of relative entropic spheres defined with respect to the Kullback-Leibler divergence and the Itakura-Saito divergence in terms of the branches of the Lambert W function~\cite{LambertW-1996}.

\subsubsection{Parameterization of the Kullback-Leibler spheres}

The scalar extended Kullback-Leibler divergence defined on $\bbR^+$ is defined by:
$$
D_\KL[\theta_c:\theta] := \theta_c\log\frac{\theta_c}{\theta}+\theta-\theta_c.
$$

The two solutions of this equation will rely on two branches of the Lambert W function~\cite{LambertW-1996}.
Indeed, the equation $xe^x=a$ solves as 
$$
x=\left\{
\begin{array}{ll}
W_0(a) & \mbox{if $a\geq\frac{1}{e}$}\\
W_{-1}(a) & \mbox{if $-\frac{1}{e}\leq a<0$}
\end{array}
\right.
$$

To solve $D_\KL[\theta_c:\theta]=u$ for $\theta_c\in(0,1)$ and $u\geq 0$ amounts to solve the equation
\begin{eqnarray*}
\theta_c\log\frac{\theta_c}{\theta}+\theta &=& u+\theta_c,\\
\theta-\theta_c\log\theta &=& u+\theta_c(1-\log\theta_c).
\end{eqnarray*}

We get the two solutions using the two branches $W_0$ and $W_{-1}$ of the Lambert W function~\cite{LambertW-1996}:
\begin{eqnarray}
x_{D_\KL,\theta_c,u;-1} &=& -\theta_c W_0\left(-\exp\left(-\frac{u}{\theta_c}-1\right)\right),\\
x_{D_\KL,\theta_c,u; 1} &=& -\theta_c W_{-1}\left(-\exp\left(-\frac{u}{\theta_c}-1\right)\right). 
\end{eqnarray}

Figure~\ref{fig:KLsphere} displays a $D_\KL$ sphere in two dimensions by plotting the four quadrant parameterizations.
See the corresponding online video.\footnote{\url{https://www.youtube.com/watch?v=pgDwWJ1DQFY}}
The restriction of the KL sphere to the hyperplane $\sum_{i=1}^d \theta^i=1$ yields a Kullback-Leibler sphere on the $(d-1)$-dimensional probability simplex (i.e., standard simplex) embedded in $\mathbb{R}^d$.

\begin{figure}
\centering
\includegraphics[width=0.3\textwidth]{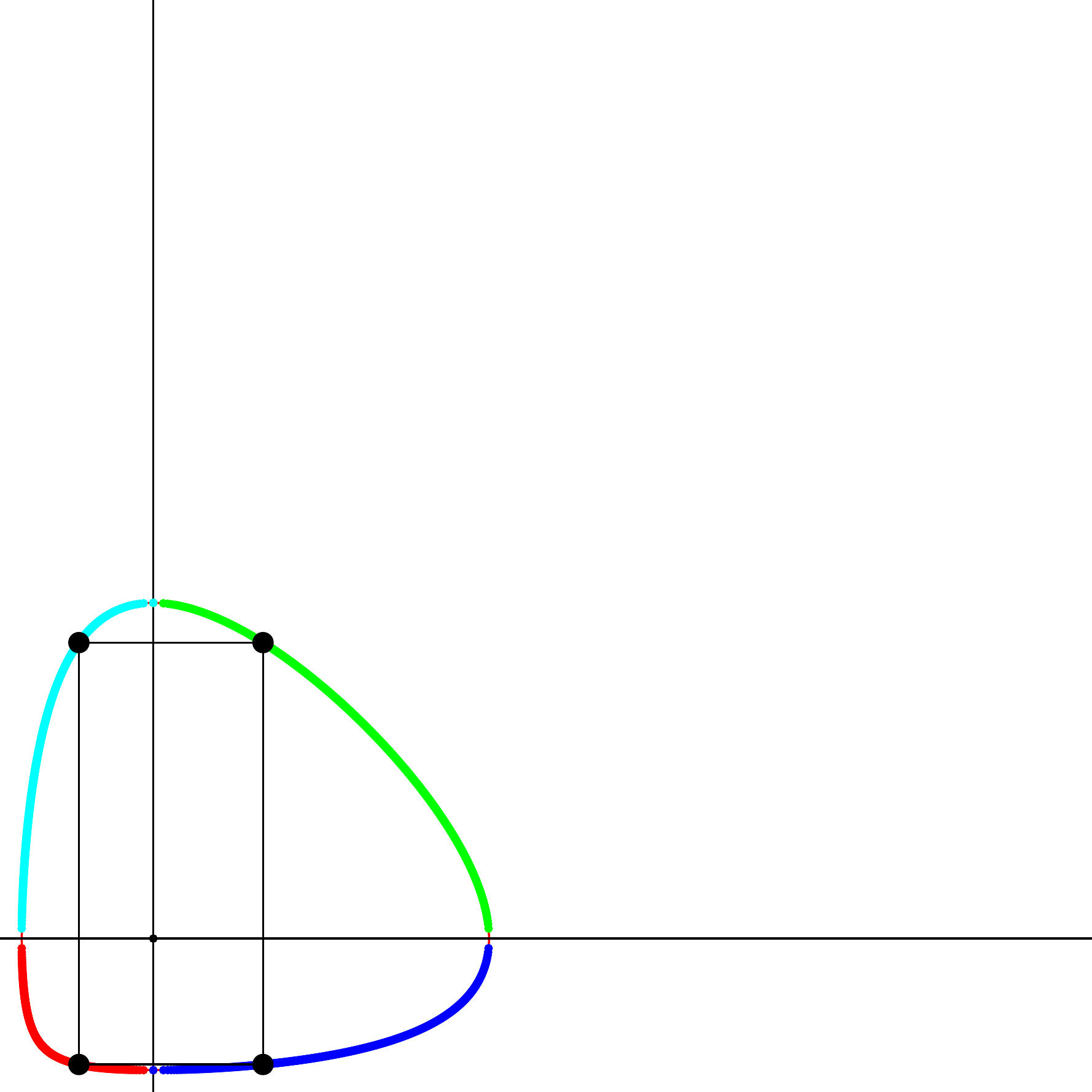}
\includegraphics[width=0.3\textwidth]{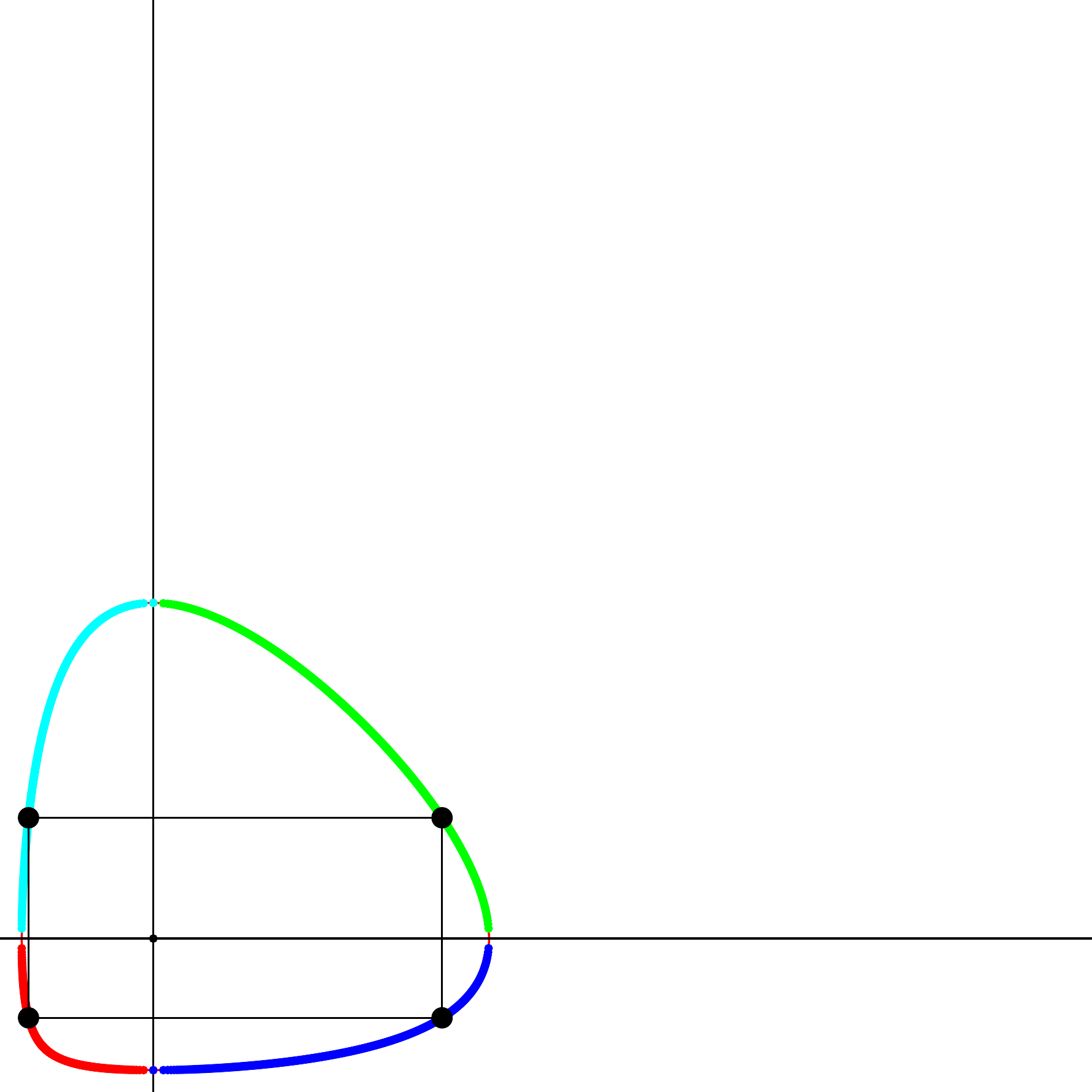}
\includegraphics[width=0.3\textwidth]{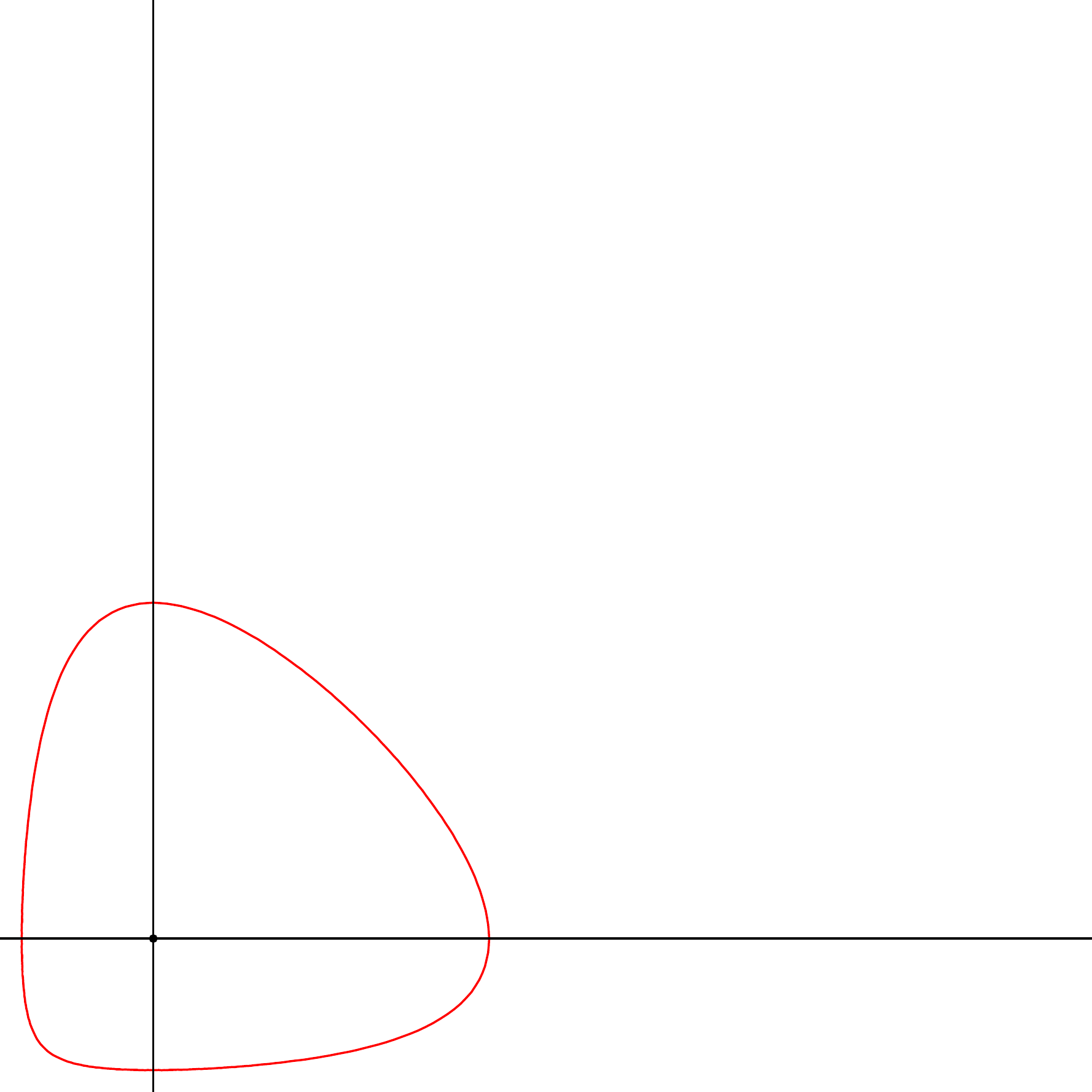}
\caption{Plotting the parametric equations on each quadrant of an extented Kullback-Leibler circle $\sphere_{D_\KL}((\frac{1}{2},\frac{1}{2}),\frac{1}{2})$ on the plane.
The first two plots show some inscribed isorectangles tangent at its four vertices to the  extented Kullback-Leibler circle.
\label{fig:KLsphere}}
\end{figure}

Figure~\ref{fig:KLsphereparam} displays the four parameterized curves of the extended  Kullback-Leibler circle.

\begin{figure}
\centering
\includegraphics[width=\textwidth]{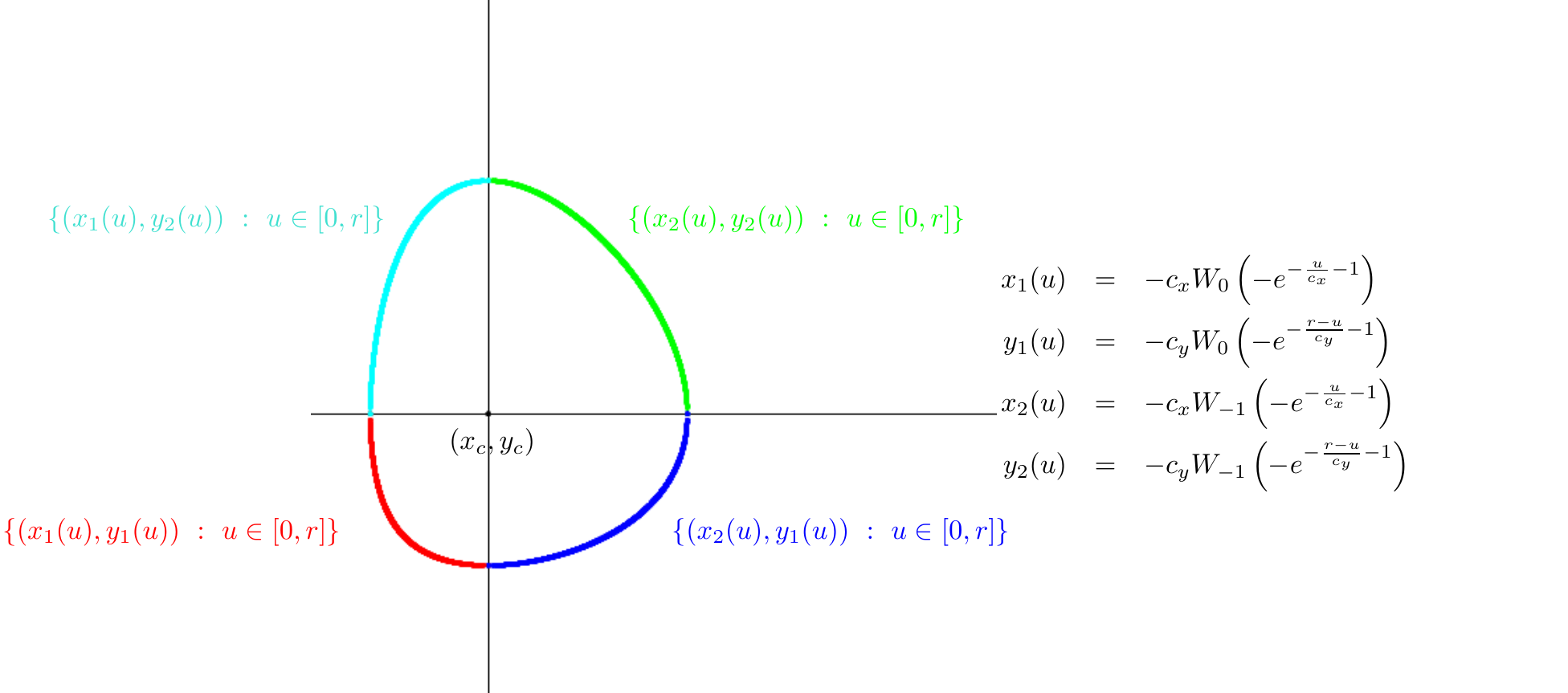}

\caption{Quadrant-based parameterized curves of the extended  Kullback-Leibler circle.
\label{fig:KLsphereparam}}
\end{figure}

\subsection{Parameterization of the Itakura-Saito spheres}

The scalar Itakura-Saito divergence  defined on $\bbR^+$ is
$$
D_\IS[\theta_c:\theta]= \frac{\theta_c}{\theta}-\log\frac{\theta_c}{\theta}-1.
$$
The Itakura-Saito divergence is a Bregman divergence for the generator $F(x)=-\log u$, and can thus be interpreted as a relative entropy.

For the Itakura-Saito divergence, we have to solve the following equation:
$$
\frac{\theta_c}{\theta}-\log\left(\frac{\theta_c}{\theta}\right)-1=u.
$$

Since $\frac{\theta_c}{\theta}=\log \exp(\frac{\theta_c}{\theta})$, we get:

\begin{eqnarray}
x_{D_\IS,\theta_c,u;-1} &=& -\frac{c_x}{W_{-1}(-e^{-u-1})},\\
x_{D_\IS,\theta_c,u;1} &=& -\frac{c_x}{W_{0}(-e^{-u-1})}.
\end{eqnarray}

\begin{figure}
\centering
\includegraphics[width=0.3\textwidth]{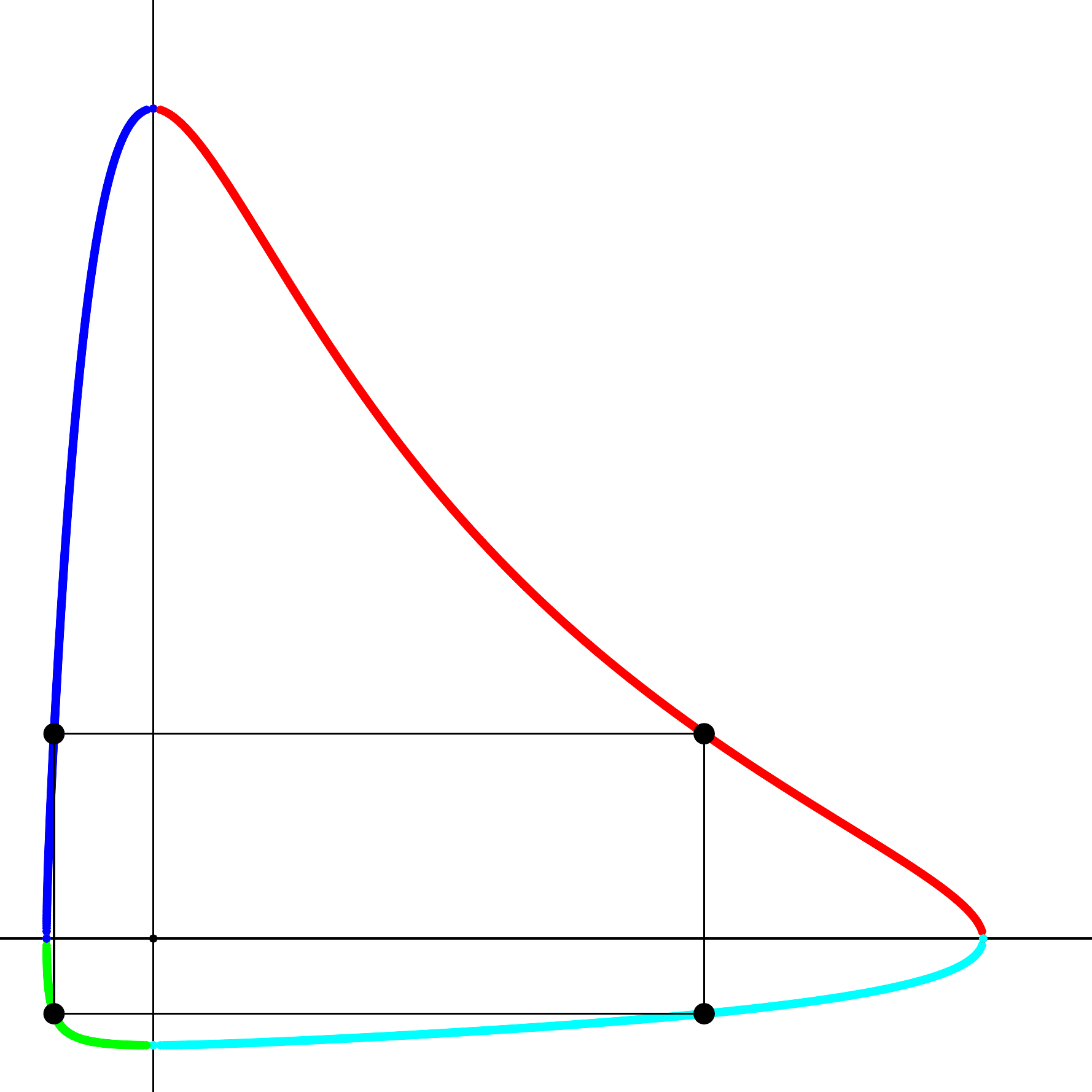}
\includegraphics[width=0.3\textwidth]{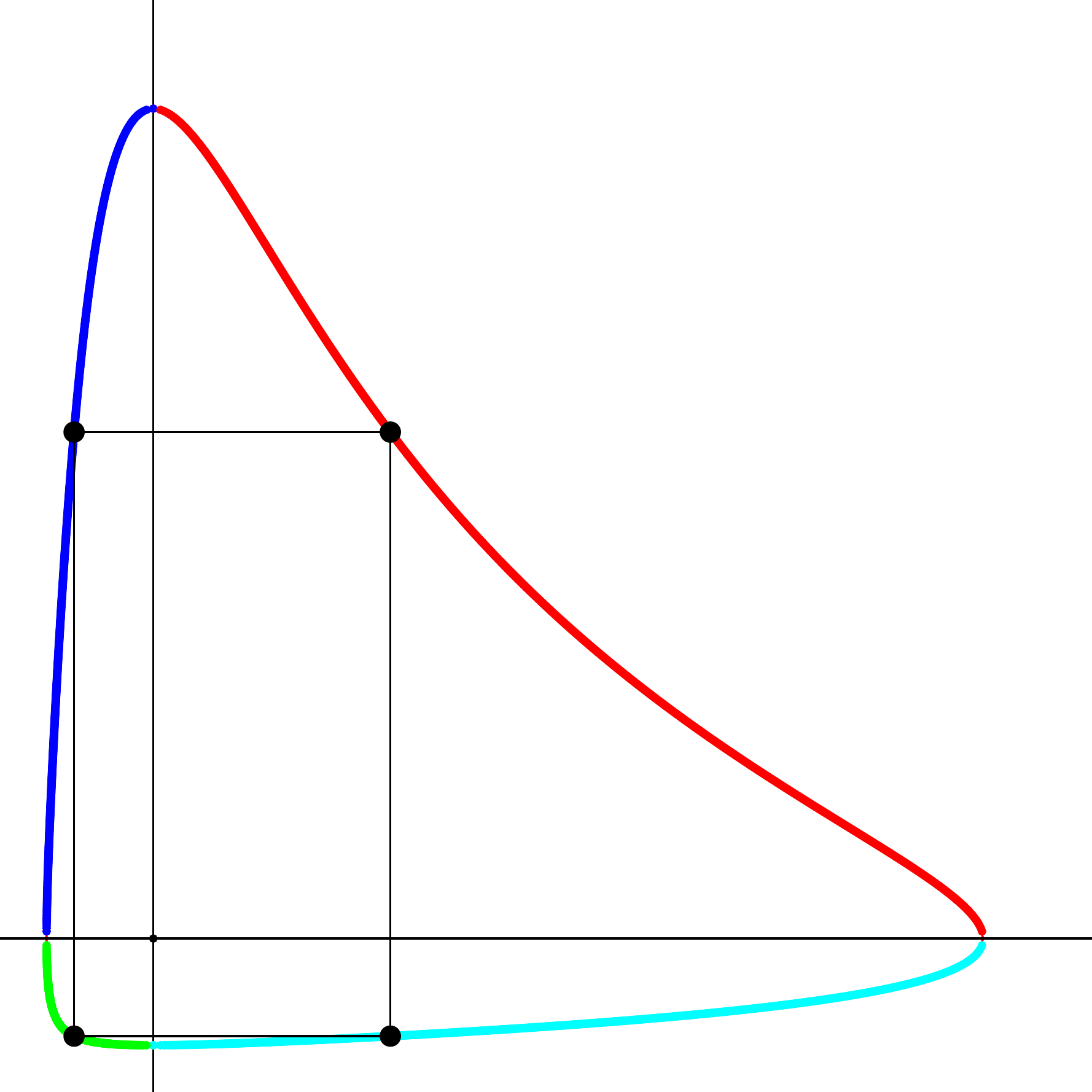}
\includegraphics[width=0.3\textwidth]{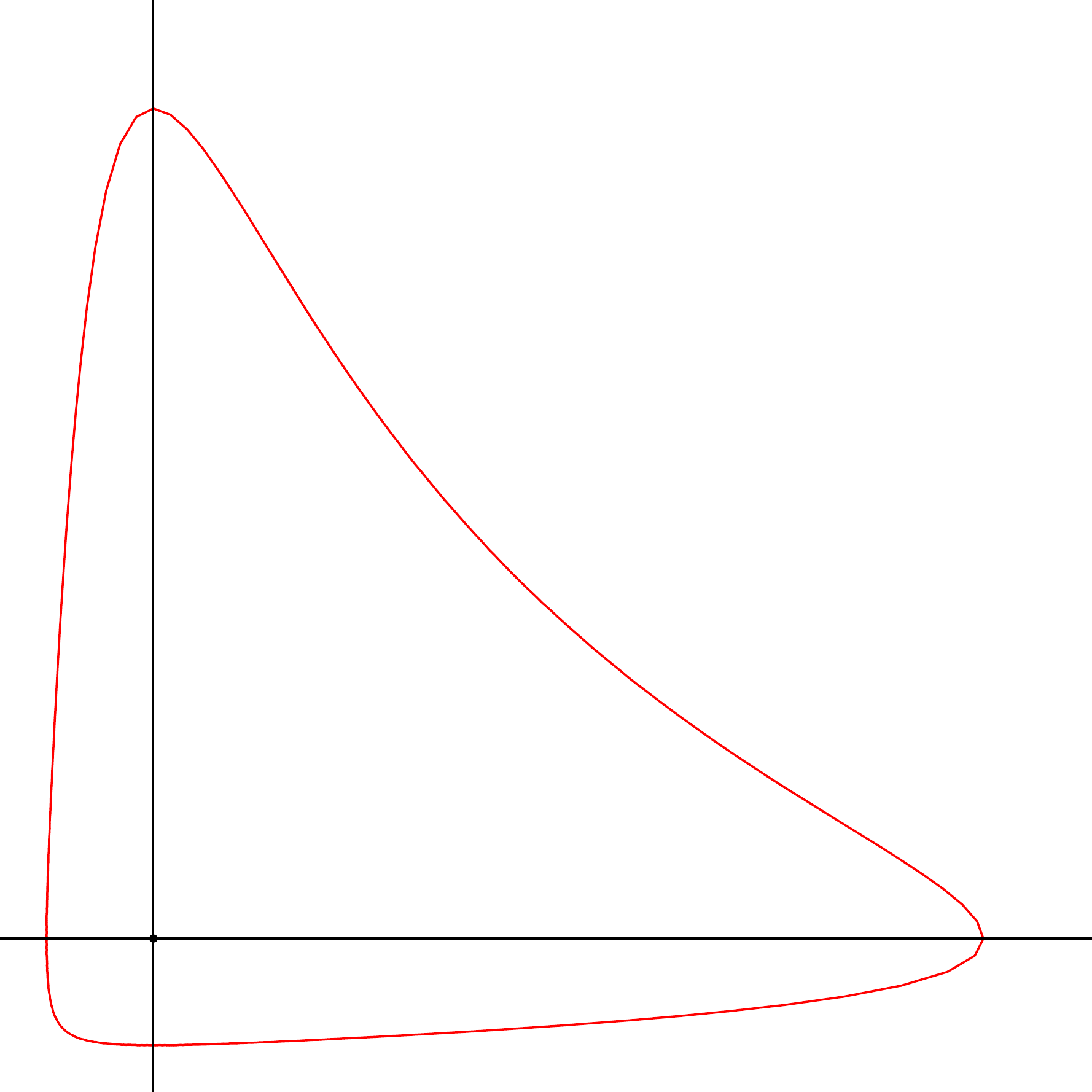}
\caption{Plotting the parametric equations on each quadrant of a Itakura-Saito circle $\sphere_{D_\IS}((\frac{1}{2},\frac{1}{2}),1)$ on the plane.
The first two plots show some inscribed isorectangles tangent at its four corners to the IS sphere.
\label{fig:ISsphere}}
\end{figure}

Figure~\ref{fig:ISsphere} displays a $D_\IS$ sphere in two dimensions by plotting the four quadrant parameterizations.
See the corresponding online video.\footnote{\url{https://www.youtube.com/watch?v=b4eJu1TT-to}}

Figure~\ref{fig:ISsphereparam} displays the four parameterized curves of the Itakura-Saito circle.

\begin{figure}
\centering
\includegraphics[width=\textwidth]{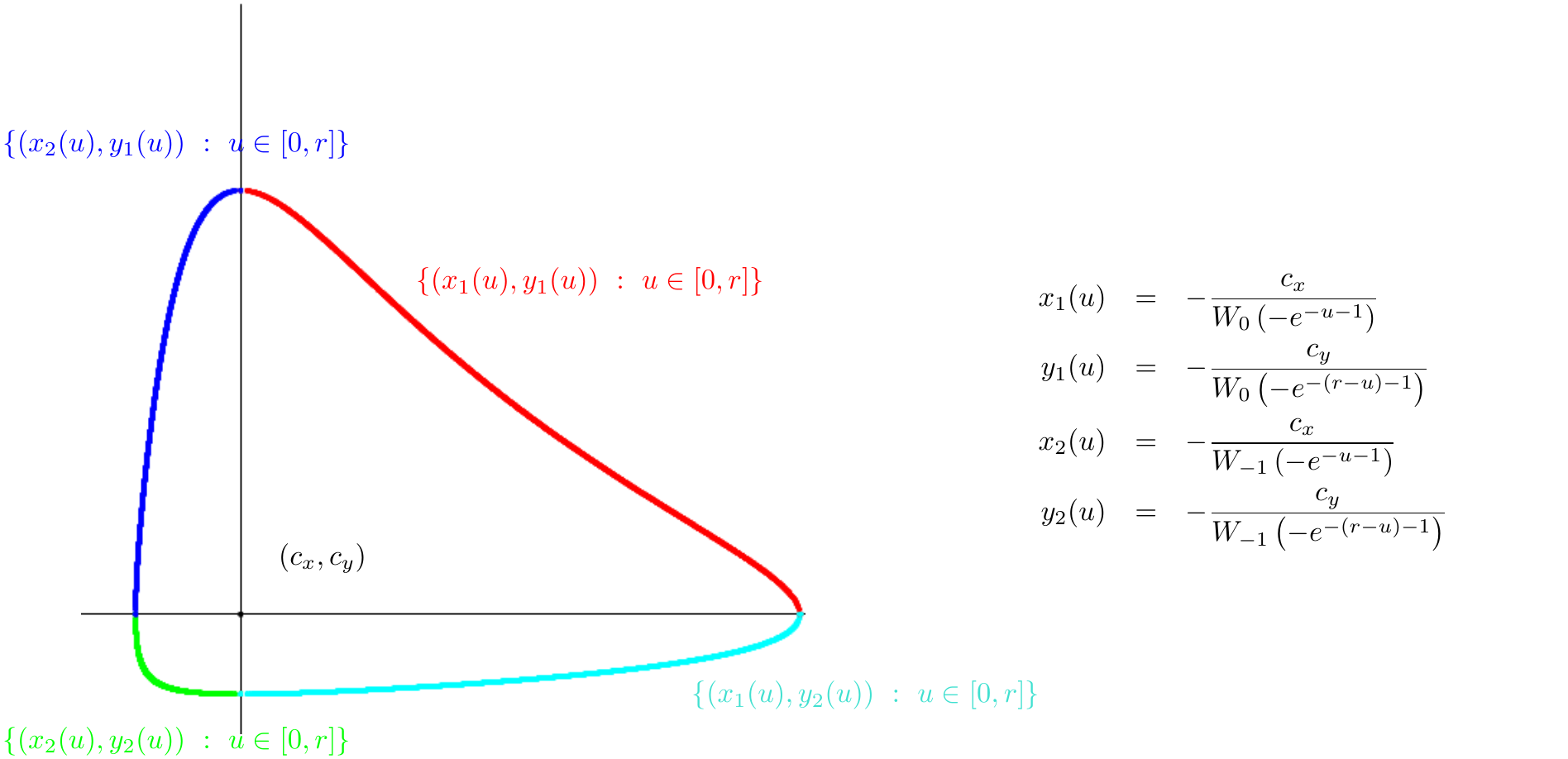}

\caption{Quadrant-based parameterized curves of the extended  Itakura-Saito circle.
\label{fig:ISsphereparam}}
\end{figure}

Figure~\ref{fig:circleparamequations} visualizes and compares the parametric equations  on a quadrant of the extended Kullback-Leibler circle and the Itakura-Saito circle.

\begin{figure}
\centering
\includegraphics[width=\textwidth]{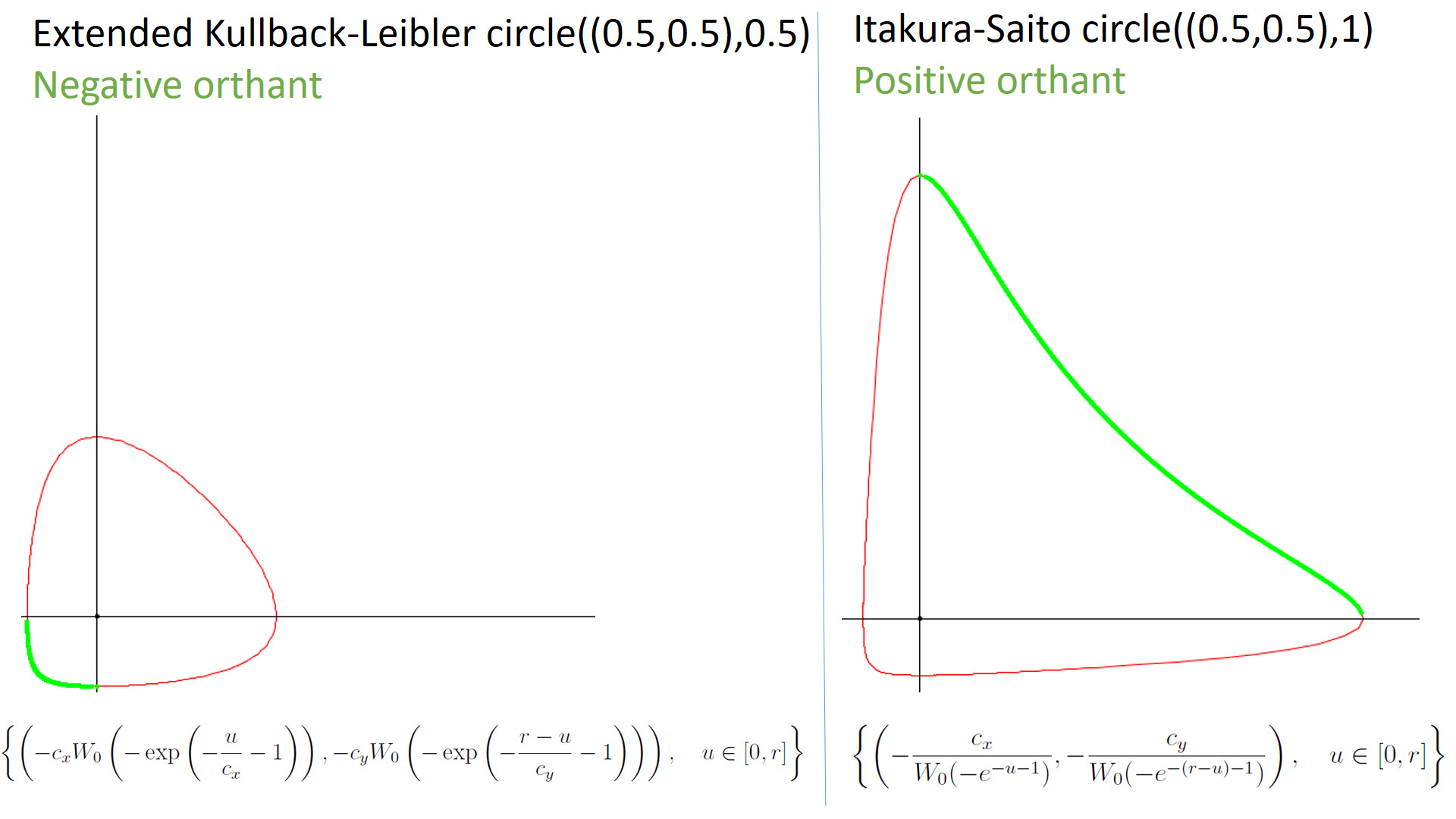}
\caption{Parametric equations of the extended Kullback-Leibler circle and the Itakura-Saito circle on a quadrant (green part).
\label{fig:circleparamequations}}
\end{figure}

\section{Geodesic $\nabla$-triangles with one, two, or three interior right angles}\label{sec:doubleright}

\subsection{Geodesic $\nabla$-triangles with one interior right angle}

To build a geodesic $\nabla$-triangle with a right angle, fix two points $p$ and $q$ (i.e., the first two triangle vertices), and consider the location of the third triangle vertex point $r$
such that  $\gamma_{qr}\perp_q\gamma_{qp}$ (i.e., $\dot\gamma_{qr}(0) \perp_q \dot\gamma_{qp}(0)$).
We end up with the  following linear equation which defines a $\theta$-flat $H_q^\theta$:
\begin{equation}
H_q^\theta: \theta_r^\top \nabla^2 F(\theta_q) \theta_{pq} = \theta_q^\top \nabla^2 F(\theta_q) \theta_{pq}.
\end{equation}

By restricting the solution $r\in H_q^\theta$ to the manifold $M$, we get:

\begin{proposition}
The locii of points  $r$ of a $\nabla$-triangle that form a  right angle at $q$ is $H^\theta_{q}\cap M$.
\end{proposition}

\begin{figure}[hb]
\centering

\begin{tabular}{cc}
 \fbox{\includegraphics[width=0.3\textwidth]{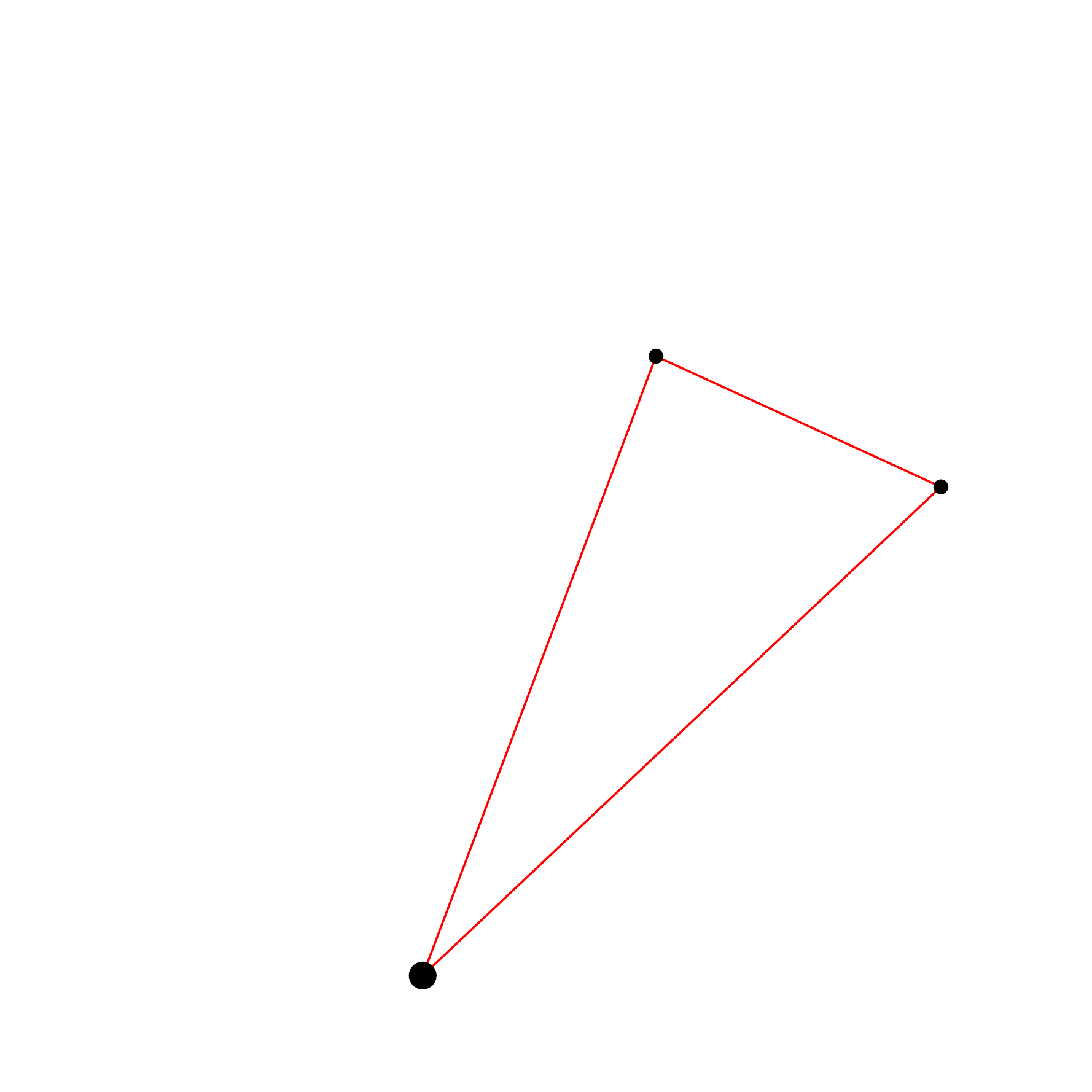}} &
 \fbox{\includegraphics[width=0.3\textwidth]{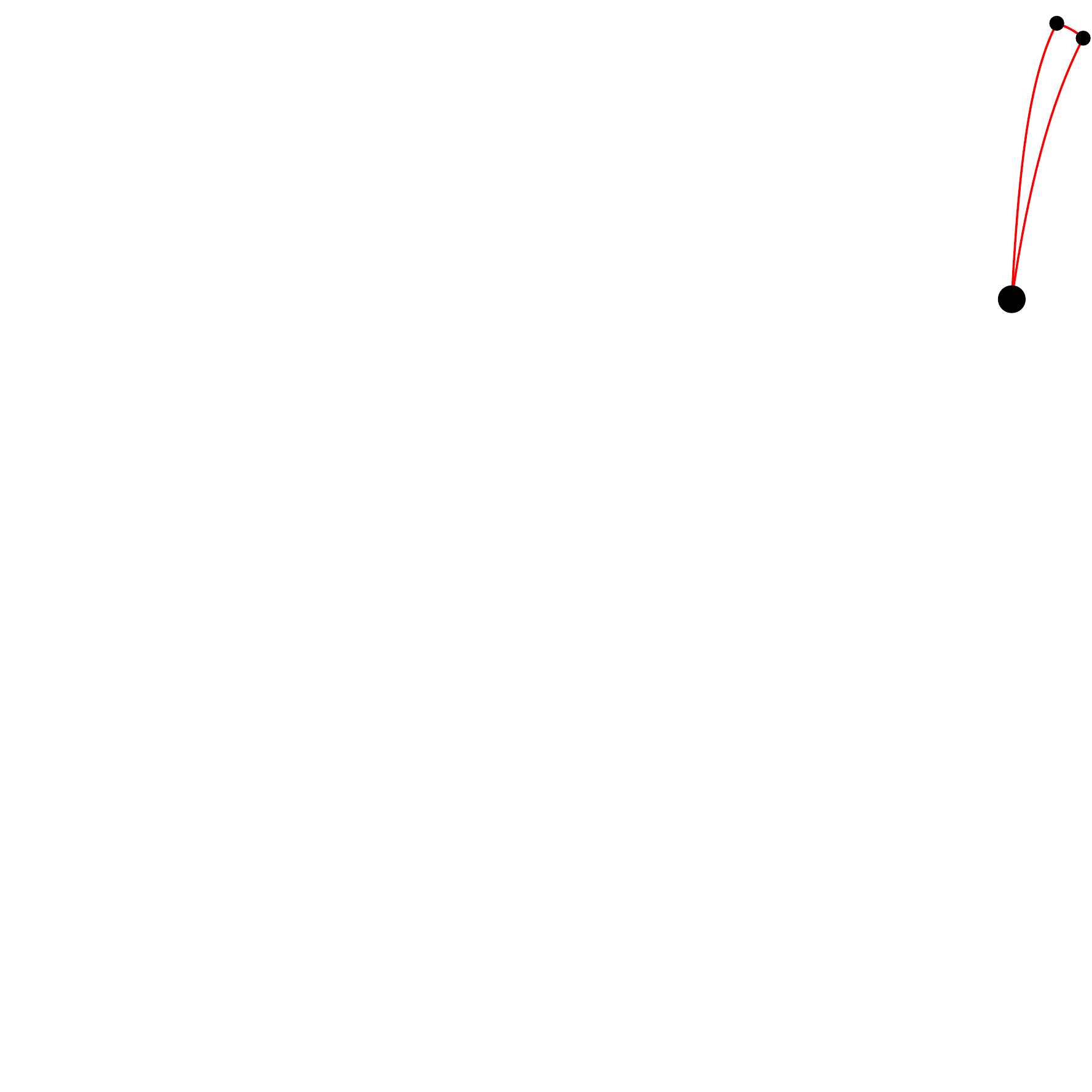}}\\
$\theta$ & $\eta$
\end{tabular}

\caption{An example of a geodesic $\nabla$-triangle with a right angle at $q$: (thicker vertex is $r$).
\label{fig:singleangleIS}}
\end{figure}

Figure~\ref{fig:singleangleIS} displays that a $\nabla$-triangle with one right angle $\alpha_q(r,p)=\frac{\pi}{2}$ in the Itakura-Saito manifold: 
$\theta(p)=(1.2885253880864789, 3.4136709176658546)$,\\
$\theta(q)=(4.9336774965526065, 1.656631440605195)$,\\
$\theta(r)=(3.5399193730133236, 4.6263857851449846)$.\\
The interior angles of the $\nabla$-triangle are 
$\alpha_p(q,r)=1.8276508176456936$, \\
$\alpha_q(p,r)=1.5707963267948966$\\
and
$\alpha_r(p,q)=1.1542328404967954$.\\
That total sum of the interior angles are $4.552679984937385$ radians (equivalent to about $260.8$ degrees).

\subsection{Geodesic $\nabla$-triangles with two interior right angles}

We now report the construction of two right angle $\nabla$-triangles.
That is, geodesic triangles with all primal geodesic edges (i.e., $\nabla$-triangle), with both the right angles $\alpha_p(q,r)=90^o$ and $\alpha_q(p,r)=90^o$.
Figure~\ref{fig:exdoubleright}  displays such a double right angle $\nabla$-triangle for the Burg negentropy generator $F_\IS(\theta)$ yielding the Itakura-Saito divergence.
Observe that because   the metric tensor field $g$ is not a scalar function of the Euclidean metric tensor $g_\Euc$, the Itakura-Saito Bregman geometry is not conformal (see Figure~\ref{fig:metric}).

\begin{figure}
\centering
\includegraphics[width=0.4\textwidth]{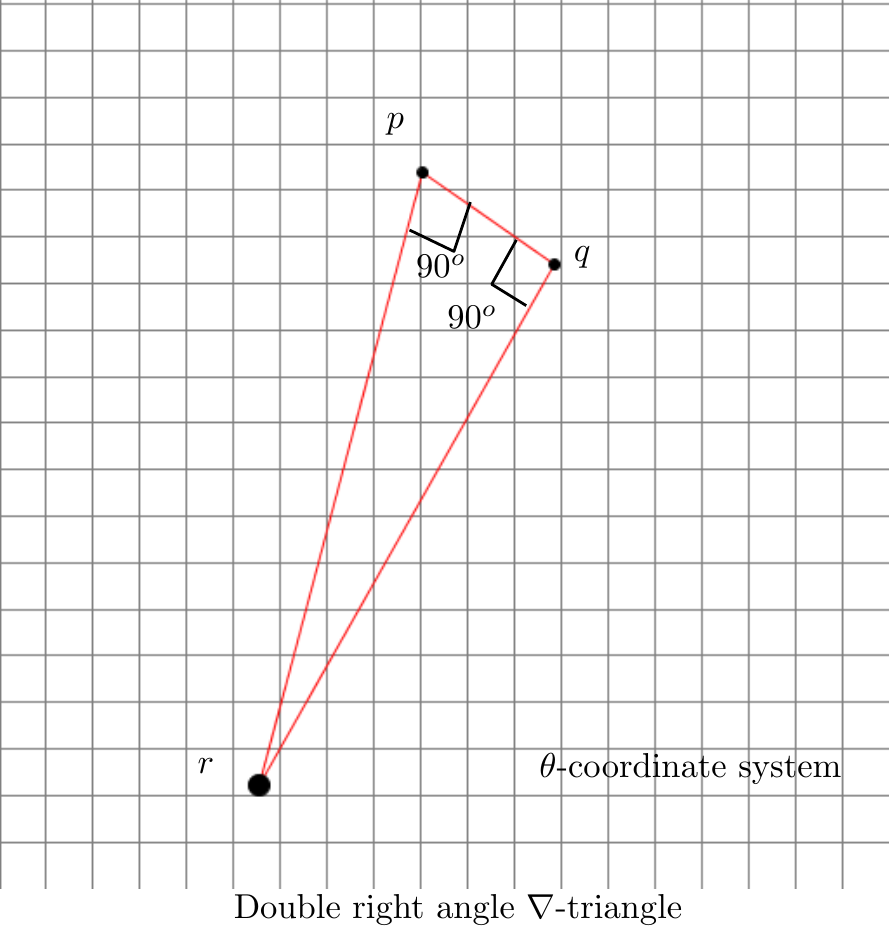}

\caption{An example of double right angle $\nabla$-triangle visualized in the primal $\theta$-coordinate system for the Itakura-Saito Bregman manifold. \label{fig:exdoubleright}}
 \end{figure}

Now, fix two points $p$ and $q$, and let us seek for the third point $r$ of the $\nabla$-triangle $\gamma_{pq}\gamma_{qr}\gamma_{rp}$ 
such that it holds both {\em simultaneously} that (i) $\gamma_{qr}\perp_q\gamma_{qp}$ (i.e., $\dot\gamma_{qr}(0)\perp_q \dot\gamma_{qp}(0)$)
and (ii) $\gamma_{pq}\perp_p\gamma_{pr}$ (i.e., $\dot\gamma_{pq}(0)\perp_p \dot\gamma_{pr}(0)$). 
We end up with the following system of equations:

\begin{equation}
\left\{
\begin{array}{lcl}
\theta_r^\top \nabla^2 F(\theta_q) \theta_{pq} &=& \theta_q^\top \nabla^2 F(\theta_q) \theta_{pq},\\
\theta_r^\top \nabla^2 F(\theta_p) \theta_{pq} &=& \theta_p^\top \nabla^2 F(\theta_p) \theta_{pq}.
\end{array}
\right.
\end{equation}

It is a linear system $A\theta=b$ 
with
\begin{equation}
A = [a_{ij}]  = \vectwo{\nabla^2 F(\theta_q) \theta_{pq}}{\nabla^2 F(\theta_p) \theta_{pq}}, \quad
b = [b_i] = \vectwo{\theta_q^\top \nabla^2 F(\theta_q) \theta_{pq}}{\theta_p^\top \nabla^2 F(\theta_p) \theta_{pq}}.
\end{equation}

When $\nabla^2F(\theta)=Q\succ 0$ for a fixed positive-definite matrix and $p\not =q$, the system does not admit any solution (i.e., case of squared Mahalanobis distances which generalize the squared Euclidean distance and can have at most one right angle).

Otherwise, this linear system solves {\em uniquely}  for asymmetric Bregman divergences~\cite{BVD-2010} using Cramer's rule as

\begin{equation}
\theta_r^1 = \frac{\dettwotwo{b_1}{a_{12}}{b_2}{a_{22}} }{|A|},\quad
\theta_r^2 = \frac{\dettwotwo{a_{11}}{b_{1}}{a_{21}}{b_{2}} }{|A|},
\end{equation} 
where $|\cdot|$ denotes the matrix determinant.

\begin{figure}[ht]
\centering

\begin{tabular}{cc}
 \fbox{\includegraphics[width=0.25\textwidth]{DoubleRight/Ex1-DoubleRightAngle-theta.pdf}} &
 \fbox{\includegraphics[width=0.25\textwidth]{DoubleRight/Ex1-DoubleRightAngle-eta.pdf}}\\
$\theta$ & $\eta$ \\
 \fbox{\includegraphics[width=0.25\textwidth]{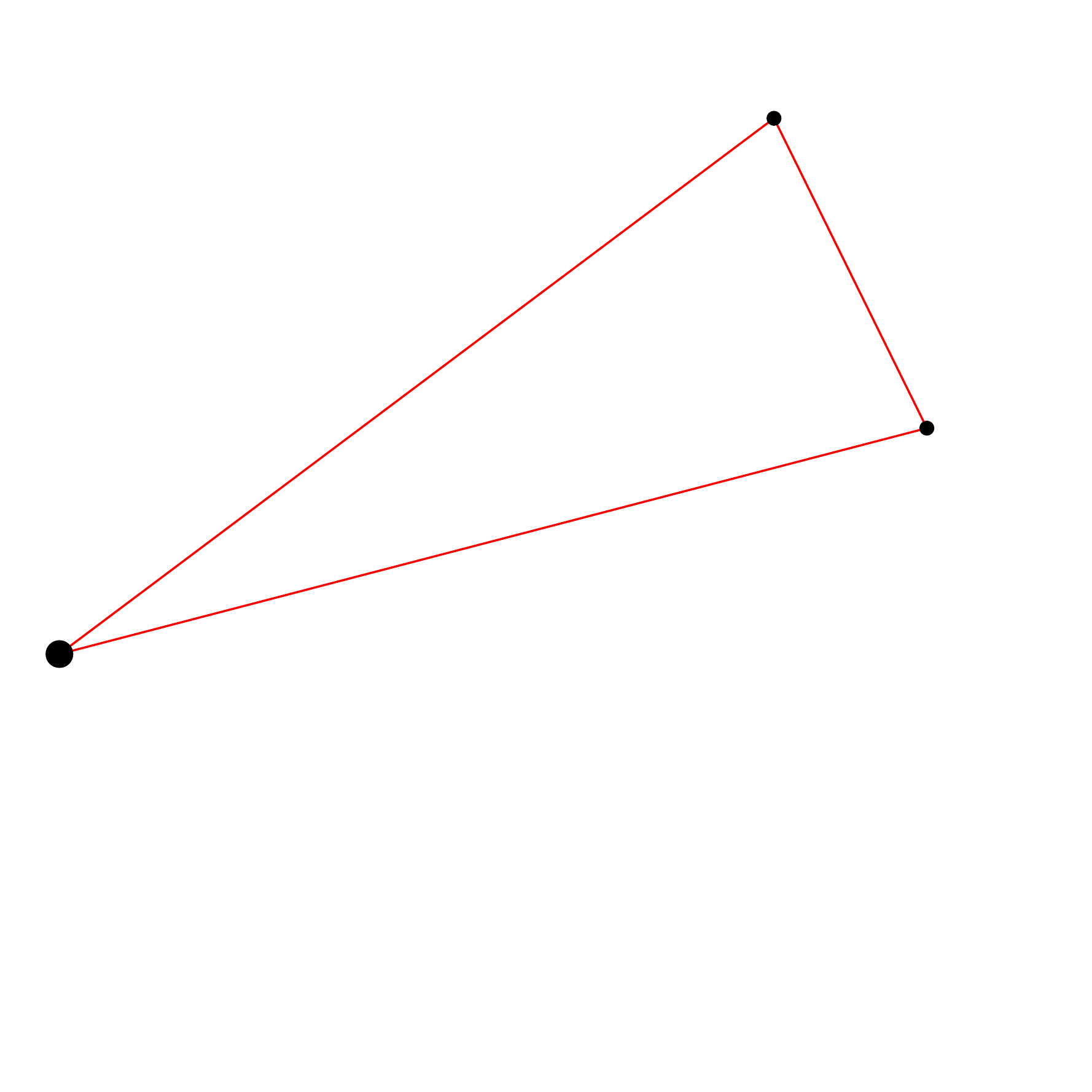}} & 
\fbox{\includegraphics[width=0.25\textwidth]{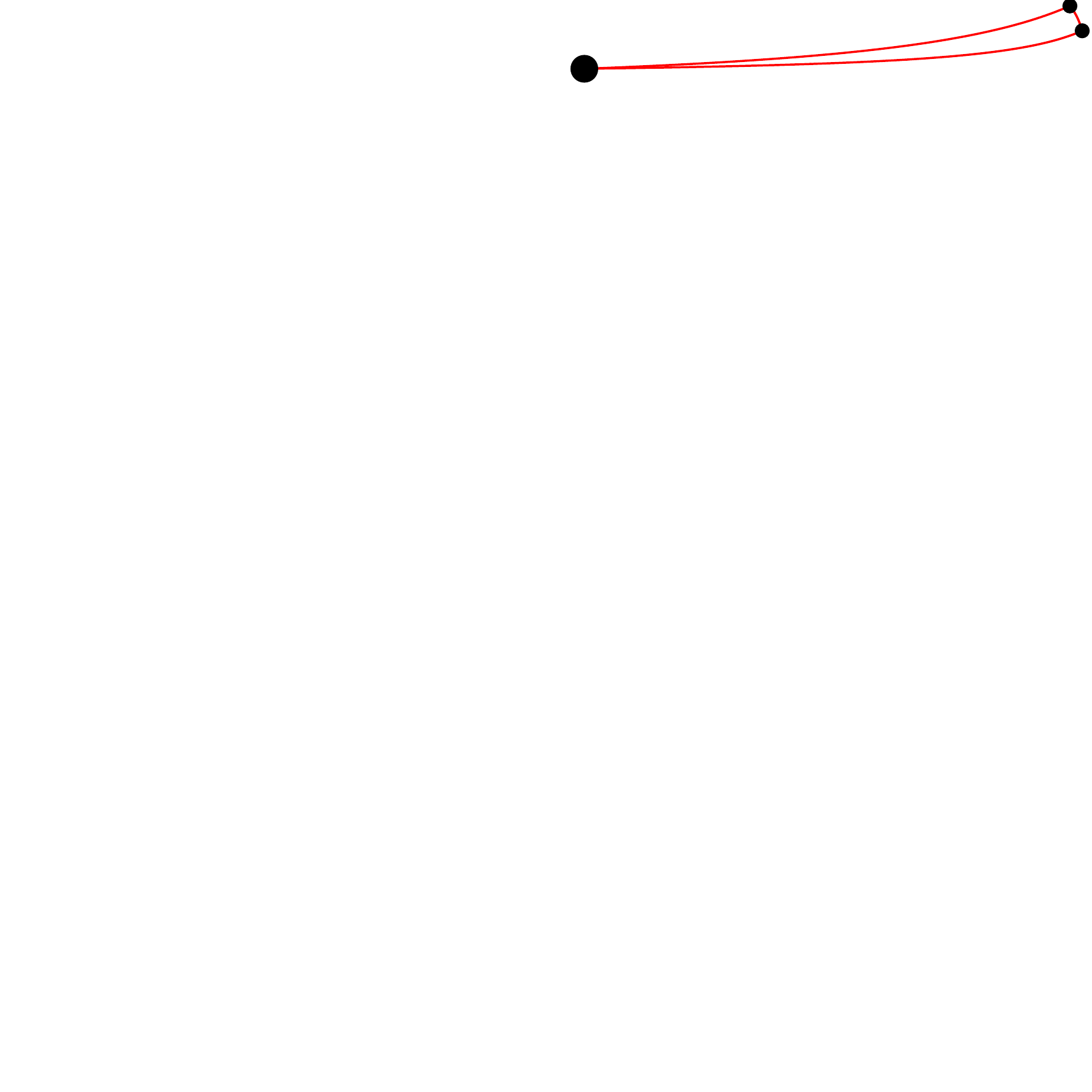}}\\
$\theta$ & $\eta$ 
\end{tabular}

\caption{Two examples of $\nabla$-triangles with two interior right angles (the opposite angles to the thicker vertex $r$).
\label{fig:doubleangleIS}}
\end{figure}

Similarly, we can build $\nabla^*$-triangles with two right angles by exchanging $F\Leftrightarrow F^*$.
Notice that having two right angles in a non-degenerate triangle makes it necessarily having an angle excess.

Figure~\ref{fig:doubleangleIS} displays the following two double right angle $\nabla$-triangles (up to machine numerical precision) obtained for the following settings:
\begin{itemize}
	\item 
$\theta(p)=(1.7372662352145616, 1.148396070619242)$,\\
$\theta(q)=(1.241571556333764, 1.3768479188317202)$,\\
$\theta(r)=(1.614143828700357, 1.8451358255393877)$,\\
$\alpha_p(q,r)=90.00000000000001$,\\
$\alpha_q(p,r)=90.0$,\\
$\alpha_r(p,q)=12.82764159141668$.

	\item
$\theta(p)=(1.7128340504770114, 1.2510418358297621)$,\\
$\theta(q)=(1.446857135939727, 1.7930125176801988)$,\\
$\theta(r)=(1.1177842396781703, 1.5922051785236535)$,\\
$\alpha_p(q,r)=90.00000000000001$,\\
$\alpha_q(p,r)=90.00000000000001$,\\
$\alpha_r(p,q)=6.595093466701274$.
	
\end{itemize}

Define the following two $\theta$-flat submanifolds:
\begin{eqnarray}
H^\theta_{q} &:& \theta_r^\top \nabla^2 F(\theta_q) \theta_{pq} - \theta_q^\top \nabla^2 F(\theta_q) \theta_{pq} =0,\\
H^\theta_{p} &:& \theta_r^\top \nabla^2 F(\theta_p) \theta_{pq}- \theta_p^\top \nabla^2 F(\theta_p) \theta_{pq}=0.
\end{eqnarray}

We end up with the following proposition:

\begin{proposition}
In a non-Mahalanobis Bregman manifold $M$, the locii of points  $r$ that form a double right angle with the geodesic arc $\gamma_{pq}$ is $H^\theta_{q}\cap H^\theta_{q}$.
\end{proposition}

Remember that the intersection of $\theta$-flats is a $\theta$-flat.

Notice that if instead of constraining two interior angles of a $\nabla$-triangle to have right angles, we ask for two {\em prescribed angles} $\alpha_q(p,r)$ and $\alpha_p(q,r)$, we would end up with the following system of non-linear equations to solve:

\begin{equation}
\left\{
\begin{array}{lcl}
\theta_{rq}^\top \nabla^2 F(\theta_q) \theta_{pq} &=&  \|\theta_{rq}\|_{\nabla^2 F(\theta_q)}  \|\theta_{pq}\|_{\nabla^2 F(\theta_q)} \cos\alpha_q(p,r),\\
\theta_{rq}^\top \nabla^2 F(\theta_p) \theta_{pq} &=& \|\theta_{rq}|_{\nabla^2 F(\theta_p)}  \|\theta_{pq}|_{\nabla^2 F(\theta_p)} \cos\alpha_p(q,r).
\end{array}
\right.
\end{equation}

Solving this non-linear system for $r$ gives a solution whenever the system is feasible.

\subsection{Geodesic $\nabla$-triangles with three interior right angles}

Given a triple of points  $(p,q,r)$, we may consider two fundamental types of triangles (up to duality and point permutations):
a $\nabla$-triangle (type ppp) or a triangle of type $pdp$.

Let us consider a geodesic $\nabla$-triangle $\gamma_{pq}\gamma_{qr}\gamma_{rp}$ so that it holds simultaneously:
\begin{equation}
\gamma_{pq} \perp_p\gamma_{pr}, \quad
\gamma_{qr} \perp_q \gamma_{qp},\quad
\gamma_{rp} \perp_r \gamma_{rq}.
\end{equation}

Writing the above constraints in the primal $\theta$-coordinate system, we end up with the following system to solve:
\begin{equation}
\left\{\begin{array}{lcl}
\theta_{qp}^\top \nabla^2 F(\theta_p) \theta_{rp} &=& 0\\
\theta_{rq}^\top \nabla^2 F(\theta_q) \theta_{pq} &=& 0\\
\theta_{pr}^\top \nabla^2 F(\theta_r) \theta_{qr} &=& 0. 
\end{array}
\right.
\end{equation}

Because of the Hessian matrices, this yields in general a non-linear system of equations to solve.
The set of feasible solutions define the $\nabla$-triangles with three right angles.
In dimension $D$, we have 3D unknown (the $D$ $\theta$-coordinates of the points $p$, $q$, and $r$) for $3$ constraints.
That is, the system is underconstrained.

For the 2D Itakura-Saito manifold, the Hessian matrix at $p$ is $\diag\left(\frac{1}{\sqr(\theta_x)}, \frac{1}{\sqr(\theta_y)} \right)$.
Thus we get the following system to solve:

\begin{eqnarray}
(\theta_q^x-\theta_p^x) \frac{1}{\sqr(\theta_p^x)} (\theta_r^x-\theta_p^x) +
 (\theta_q^y-\theta_p^y) \frac{1}{\sqr(\theta_p^y)} (\theta_r^y-\theta_p^y)  &=& 0,\label{eq:seq1}\\
(\theta_r^x-\theta_q^x) \frac{1}{\sqr(\theta_q^x)} (\theta_p^x-\theta_q^x) +
 (\theta_r^y-\theta_q^y) \frac{1}{\sqr(\theta_q^y)} (\theta_p^y-\theta_q^y)  &=& 0,\label{eq:seq2}\\
 (\theta_p^x-\theta_r^x) \frac{1}{\sqr(\theta_r^x)} (\theta_q^x-\theta_r^x) +
 (\theta_p^y-\theta_r^y) \frac{1}{\sqr(\theta_r^y)} (\theta_q^y-\theta_r^y)  &=& 0.\label{eq:seq3}\\ 
\end{eqnarray}

Multiplying Eq.~\ref{eq:seq1} by $\sqr(\theta_p^x)\sqr(\theta_p^y)$, 
Eq.~\ref{eq:seq2} by $\sqr(\theta_q^x)\sqr(\theta_q^y)$, and 
 Eq.~\ref{eq:seq3} by $\sqr(\theta_r^x)\sqr(\theta_r^y)$, we get  a system of polynomial equations to solve~\cite{Sturmfels-2002}:
\begin{eqnarray}
(\theta_q^x-\theta_p^x) \sqr(\theta_p^y) (\theta_r^x-\theta_p^x) +
 (\theta_q^y-\theta_p^y) {\sqr(\theta_p^x)} (\theta_r^y-\theta_p^y)  &=& 0,\\
(\theta_r^x-\theta_q^x) {\sqr(\theta_q^y)} (\theta_p^x-\theta_q^x) +
 (\theta_r^y-\theta_q^y) {\sqr(\theta_q^x)} (\theta_p^y-\theta_q^y)  &=& 0,\\
 (\theta_p^x-\theta_r^x) {\sqr(\theta_r^y)} (\theta_q^x-\theta_r^x) +
 (\theta_p^y-\theta_r^y){\sqr(\theta_r^x)} (\theta_q^y-\theta_r^y)  &=& 0,
\end{eqnarray}
where $\sqr(x)=x^2$.
The system of polynomial equations has $3$ equations with $6$ positive variables defining $3$ pairs of distinct points $(\theta_p^x,\theta_p^y)$,
 $(\theta_q^x,\theta_q^y)$, and $(\theta_r^x,\theta_r^y)$.
 
When considering the extended Kullback-Leibler manifold, since the Hessian matrix at $p$ is $\diag\left(\frac{1}{\theta_p^x}, \frac{1}{\theta_p^y}\right)$, we get the following system of polynomial equations (after multiplying the first equation by $\theta_p^x\theta_p^y$,
 the second equation by $\theta_q^x\theta_q^y$ and the third equation by $\theta_r^x\theta_r^y$):
\begin{equation}
\left\{
\begin{array}{lcl}
(\theta_q^x-\theta_p^x) \theta_p^y (\theta_r^x-\theta_p^x) +
 (\theta_q^y-\theta_p^y) \theta_p^x (\theta_r^y-\theta_p^y)  &=& 0,\\
(\theta_r^x-\theta_q^x) \theta_q^y (\theta_p^x-\theta_q^x) +
 (\theta_r^y-\theta_q^y) \theta_q^x (\theta_p^y-\theta_q^y)  &=& 0,\\
 (\theta_p^x-\theta_r^x) \theta_r^y (\theta_q^x-\theta_r^x) +
 (\theta_p^y-\theta_r^y)\theta_r^x (\theta_q^y-\theta_r^y)  &=& 0.
\end{array}
\right.
\end{equation}

Fixing three variables, we get a cubic system of three equations in three unknowns. 

We used Wolfram alpha\texttrademark{} to check for potential real solution(s) of the system:
{\small
\begin{verbatim}
a>0, b>0, c>a, x>0, y>c,z>0,   
(c-a)*b*(y-a)+(x-b)*a*(z-b)=0,
(y-c)*x*(a-c)+(z-x)*c*(b-x)=0,
(a-y)*z*(c-y)+(b-z)*y*(x-z)=0
\end{verbatim}
}
The system does not admit solution in the positive orthant.
We also checked the feasibility of the system for the Itakura-Saito manifold:
{\small
\begin{verbatim}
a>0, b>0, c>a, x>0, y>c,z>0, 
(c-a)*b*b*(y-a)+(x-b)*a*a*(z-b)=0,
(y-c)*x*x*(a-c)+(z-x)*c*c*(b-x)=0,
(a-y)*z*z*(c-y)+(b-z)*y*y*(x-z)=0.
\end{verbatim}
}
The system does not admit solution in the positive orthant.

Thus it is an ongoing task to report an example of such a $\nabla$-triangle with three right angles for an asymmetric Bregman divergence, or to prove that such a triangle can never exist.

\section{Simultaneous satisfying the dual Pythagorean theorems}\label{sec:simultaneous}

Fix two points $p$ and $q$ of the Bregman manifold $M$. 
We seek for the locii of the third point $r\in M$ such that we have both
 $\gamma_{pq} \perp_q \gamma_{qr}^*$ and $\gamma_{pq}^* \perp_q \gamma_{qr}$.
That is, we need to solve the following system of equations:
\begin{equation}
\left\{
\begin{array}{lll}
(\eta(p)-\eta(q))^\top (\theta(r)-\theta(q))&=& 0,\cr
(\eta(p)-\eta(q))^\top (\theta(r)-\theta(q))&=& 0.
\end{array}
\right.
\end{equation}

Notice that finding such triples of points allow one to construct dual geodesic right-angle triangles.
When $F(\theta)=\frac{1}{2}\theta^\top Q\theta$ for $Q\succ 0$ a positive-definite matrix, we have $B_F$ that is a squared Mahalanobis distances, and  the primal and dual geodesics coincide. Therefore any right-angle triangle is a dually right-angle solution to the problem in Mahalanobis manifolds.
Thus we shall consider {\em asymmetric} Bregman divergences in the remainder since the only symmetric Bregman divergences are  squared Mahalanobis distances~\cite{BVD-2010}.

\subsection{Simultaneous dual orthogonality: A geometric interpretation}
The constraint $\gamma_{pq} \perp_q \gamma_{qr}^*$ is equivalent to
$(\theta(p)-\theta(q))^\top (\eta(r)-\eta(q))=0$.
Let $\theta_{pq}=\theta(p)-\theta(q)$, $\eta_r=\eta(r)$ and $\eta_q=\eta(q)$.
Then we have the following affine equation in $\eta_r$:
\begin{equation}
H^\eta(p,q): \theta_{pq}^\top \eta_r-\theta_{pq}^\top\eta_q=0.
\end{equation}
The locii $r$ of points satisfying the above equation is a $(D-1)$-dimensional $\nabla^*$-autoparallel submanifold $H^\eta(p,q)$
(i.e., a $(D-1)$-dimensional {\em $\eta$-flat}~\cite{IG-2016} or loosely speaking a ``$\nabla^*$-hyperplane'').

Similarly, the constraint $\gamma_{pq}^* \perp_q \gamma_{qr}$ is equivalent to
$(\eta(p)-\eta(q))^\top (\theta(r)-\theta(q))=0$.
Let $\eta_{pq}=\eta(p)-\eta(q)$, $\theta_r=\theta(r)$ and $\theta_q=\theta(q)$.
Then we have the following affine equation in $\theta_r$:
\begin{equation}
H^\theta(p,q): \eta_{pq}^\top \theta_r-\eta_{pq}^\top\theta_q=0.
\end{equation}
The locii $r$ of points satisfying the above equation is a $(D-1)$-dimensional $\nabla$-autoparallel submanifold $H^\theta(p,q)$
(i.e., a $(D-1)$-dimensional {\em $\theta$-flat}~\cite{IG-2016} or $\nabla$-hyperplane).

Notice that point $q$ ought to belong to both $H^\eta(p,q)$ and $H^\theta(p,q)$.
Thus to simultaneously satisfy the two dual geodesic orthogonality constraints, the point $r$ should belong to the intersection of a $\eta$-flat with a $\theta$-flat:
\begin{equation}
S(p,q)=H^\eta(p,q)\cap H^\theta(p,q).
\end{equation}

We should make sure  $S(p,q)$   belongs to the manifold $M$ when solving the equations using either the $\theta$- or $\eta$-coordinate system.

\begin{proposition}
The locii of points $r$ such that $\gamma_{pq} \perp_q \gamma_{qr}^*$ and $\gamma_{pq}^* \perp_q \gamma_{qr}$ is
the intersection of a $\theta$-flat with a $\eta$-flat restricted to the manifold: $S(p,q)=H^\eta(p,q)\cap H^\theta(p,q)$.
\end{proposition}

In general the intersection of a $\theta$-flat with a $\eta$-flat is neither a $\theta$-flat nor a $\eta$-flat.
In 2D, the submanifolds $H^\eta(p,q)$ and $H^\theta(p,q)$ can be interpreted as a dual geodesic and a primal geodesic, respectively.
Thus in 2D, $S_{pq}$ may not be simply connected, and the maximum number of points of the intersection of two geodesics upper bounds the number of solutions for $r$.
The next section illustrates how to build such triples of points for the Itakura-Saito manifold.

\subsection{Explicit construction in the Itakura-Saito manifold}\label{sec:doublyrightangleIS}

Consider the Itakura-Saito manifold described in~\S\ref{sec:IS}.
Let us exhibit some triple of points $(p,q,r)$  such that $\gamma_{pq} \perp_q \gamma_{qr}^*$ and $\gamma_{pq}^* \perp_q \gamma_{qr}$.
To avoid confusion, let us  write the 2D $\theta$- and $\eta$-coordinates of $p$ by $\theta_p^x$ and $\theta_p^y$, and  $\eta_p^x$ and $\eta_p^y$.
 
The second orthogonality constraint equation yields the equation $\eta_{pq}^\top (\theta_r-\theta_q)= \eta_{pq}^x\theta_r^x+\eta_{pq}^y\theta_r^y-\eta_{pq}^\top\theta_q = 0$ which can be rewritten as:
\begin{eqnarray}
\theta_r^y &=& -\frac{\eta_{pq}^x}{\eta_{pq}^y}\theta_r^x  + \frac{\eta_{pq}^\top\theta_q}{\eta_{pq}^y}= a\theta_r^x+b,
\end{eqnarray}
with $a= -\frac{\eta_{pq}^x}{\eta_{pq}^y}$ and $b=\frac{\eta_{pq}^\top\theta_q}{\eta_{pq}^y}$.

The first orthogonality constraint equation yields $(\theta_{pq})^\top (\eta_r-\eta_q)$ with $\eta_r^x=-\frac{1}{\theta_r^x}$  and $\eta_r^y=-\frac{1}{\theta_r^y}$. 
After multiplying both sides of the equation  by $\theta_r^x\theta_r^y$, we find 
\begin{eqnarray}
-\theta_{pq}^x \theta_r^y -\theta_{pq}^y\theta_r^x-\theta_{pq}^\top \eta_q \theta_r^x\theta_r^y=0
\end{eqnarray}

\begin{figure}
\centering
\includegraphics[width=0.3\textwidth]{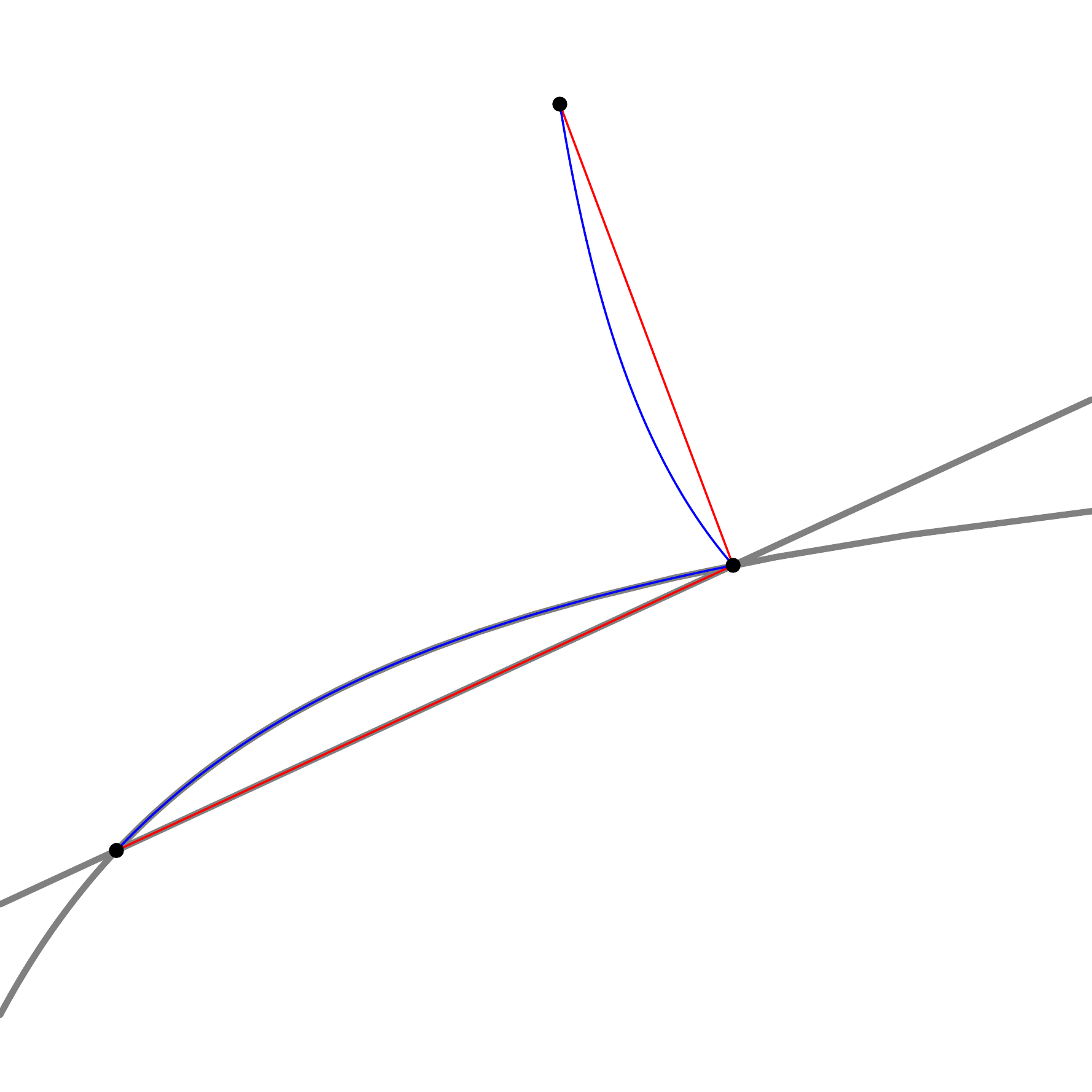}

\caption{The intersection of $H^\eta(p,q)$ and $H^\theta(p,q)$ (displayed in black) yields two points: point $q$ and the solution point $r$.\label{fig:constraints}}
\end{figure}

Letting $\theta_r^y=a\theta_r^x+b$, we get the following quadratic equation to solve:

\begin{eqnarray}
-\theta_{pq}^x (a\theta_r^x+b) -\theta_{pq}^y\theta_r^x-\theta_{pq}^\top \eta_q \theta_r^x (a\theta_r^x+b)=0.
\end{eqnarray}

We write the former equation into the following canonical form of quadratic equations:

\begin{eqnarray}
\underbrace{(-\theta_{pq}^\top \eta_q a)}_{A} (\theta_r^x)^2+
\underbrace{(-\theta_{pq}^xa-\theta_{pq}^y-b\theta_{pq}^\top\eta_q)}_{B} \theta_r^x
+\underbrace{(-\theta_{pq}^xb)}_{C}=0.
\end{eqnarray}

Then we solve the quadratic equation, and get the two solutions: 
Let $\Delta=B^2-4AC>0$ be the discriminant.
We have the two quadratic roots:
 ${\theta_r^x}=\frac{-B-\sqrt{\Delta}}{2A}$ and ${\theta_r^x}''=\frac{-B+\sqrt{\Delta}}{2A}$,
and we recover ${\theta_r^y}'=a{\theta_r^x}'+b$ and ${\theta_r^y}''=a{\theta_r^x}''+b$.
One of the two solutions $r'$ or $r''$ coincide with point $q$ (with coordinate $\theta(q)$), so the solution point $r$ is the remaining  distinct point.

Figure~\ref{fig:constraints} displays an example of a triple $(p,q,r)$ with doubly right-angle at $q$ with the pair of constraints   $H^\eta(p,q)$ and $H^\eta(p,q)$ passing through point $q$.

Let us give a numerical example. Set\\ 
$\theta(p)=(\theta_p^x,\theta_p^y)=(0.7273955397832663,	0.3279475469672596)$,\\
$\theta(q)=(\theta_q^x,\theta_q^y)=(0.46251884248040354,	0.3902872167636309)$.
Then we solve the quadratic equation and find the two solutions\\
$\theta(r)=(0.3065847355580658,	0.13822426240588664)$\\
and\\
$\theta(r'')=(0.4625188424804033,	0.39028721676363043)$.\\
Observe that $r''=q$ so that the solution is $r=r'$.
The triple $(p,q,r')$   holds simultaneously the dual Pythagorean theorems at point $q$.
 
Figure~\ref{fig:simultaneous} displays three examples of triples of points for which the  dual Pythagorean theorems hold
simultaneously. The triples of points displayed in Figure~\ref{fig:simultaneous} are from left to right:

\begin{enumerate}
\item
$\theta(p)=(0.9704854205553236, 1.4760141668100146)$,\\
$\theta(q)=(1.141690604206171, 0.43035569351200803)$ and\\
$\theta(r)=(0.2264761824188501 0.34444830042268043)$.

\item
$\theta(p)=(1.3163859900481611, 1.965380252548788)$,\\
$\theta(q)=(1.5136826962585432, 1.2440688670072433)$ and\\
$\theta(r)=(0.6359397574807304, 0.9494657726625966)$.

\item
$\theta(p)=(0.9511702030611633, 1.291145089053253)$,\\
$\theta(q)=(0.3277859642409383, 1.906447912395776)$ and\\
$\theta(r)=(0.1077217190919158, 0.14622448026891943)$.
 \end{enumerate}

\begin{figure}
\centering

\fbox{\includegraphics[width=0.3\textwidth]{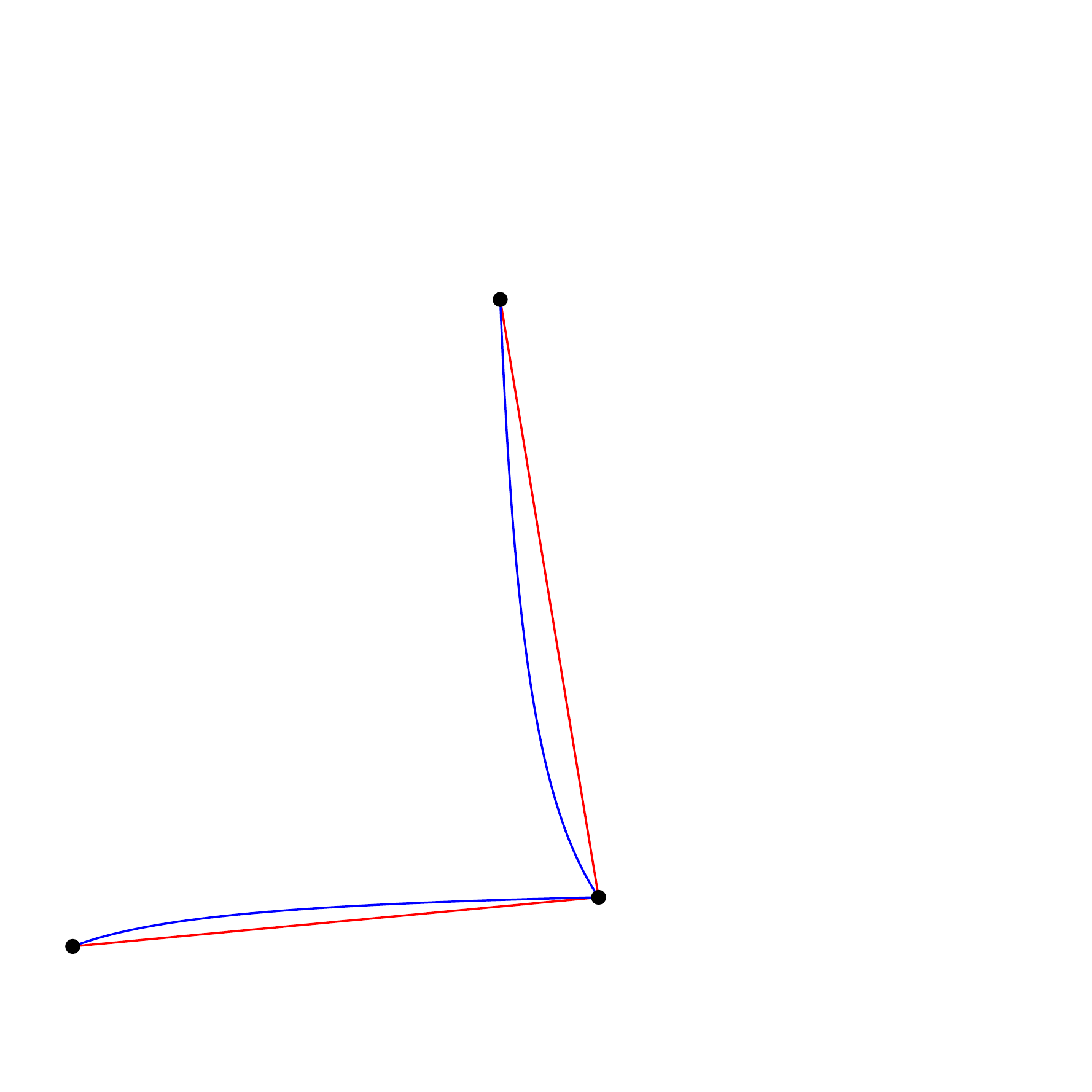}}
\fbox{\includegraphics[width=0.3\textwidth]{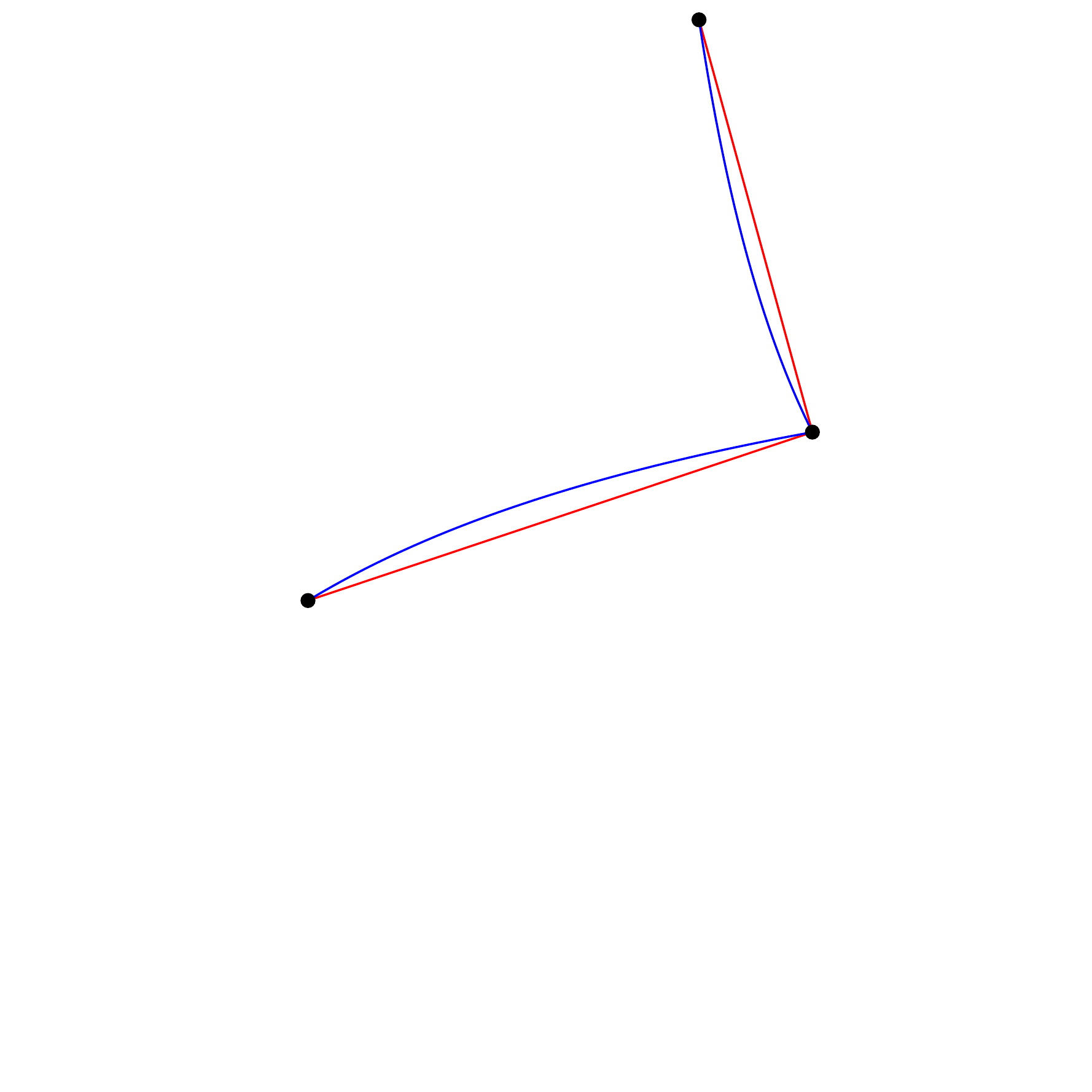}}
\fbox{\includegraphics[width=0.3\textwidth]{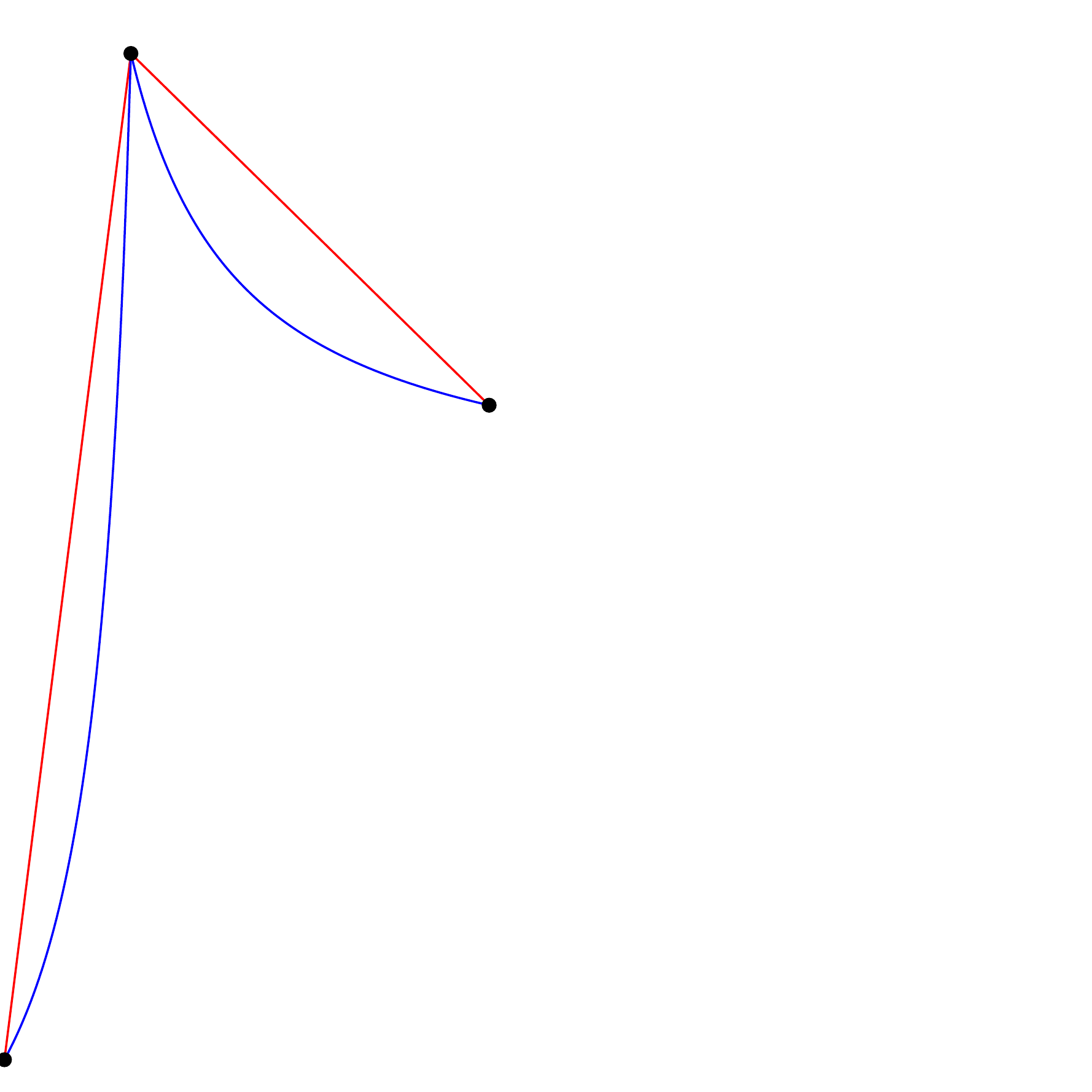}}

\caption{Three examples of triples of points $(p,q,r)$ visualized in the $\theta$-coordinate system for which the dual Pythagorean theorems hold simultaneously at $q$ for the Burg negentropy generator.
That is, we have both $\gamma_{pq} \perp_q \gamma_{qr}^*$ and $\gamma_{pq}^* \perp_q \gamma_{qr}$: 
The two pairs of (red,blue) geodesics form a right-angle at $q$: A ``doubly right-angle''.
\label{fig:simultaneous}}
\end{figure}

From these two pairs of dual right angle geodesic arcs at $q$, we can obtain four geodesic triangles by choosing either the primal or dual geodesic edge for the triangle edge $pq$: Namely, $\gamma_{pq}\gamma_{qr}\gamma_{rp}^*$, $\gamma_{pq}^*\gamma_{qr}\gamma_{rp}^*$ and
$\gamma_{pq}\gamma_{qr}^*\gamma_{rp}$, $\gamma_{pq}^*\gamma_{qr}^*\gamma_{rp}$.
These four triangles can be grouped into  two dual pairs of dual geodesic triangles which exhibit a {\em dually right angle} at vertex $q$:
$(\gamma_{pq}\gamma_{qr}\gamma_{rp}^*,\gamma_{pq}^*\gamma_{qr}^*\gamma_{rp})$
and $(\gamma_{pq}^*\gamma_{qr}\gamma_{rp}^*,\gamma_{pq}\gamma_{qr}^*\gamma_{rp})$.

Similarly, solving the dual orthogonality constraint at $q$ for the cubic generator $F(\theta)=\frac{1}{3}\sum_i\theta_i^3$ yields  a quadratic equation to solve.
However, when considering the extended Shannon negentropy generator $F(\theta)=\sum_{i=1}^D \theta^i\log\theta^i-\theta^i$, we get a nonlinear equation (with sum of logarithmic terms) to solve: In 1D, we can easily run numerical optimization to approximate a solution numerically.

\section{Conclusion}\label{sec:concl}
In a dually flat space~\cite{IG-2016} that we called a Bregman manifold in \S\ref{sec:BM}, geodesic triangles can either have angle excesses or angle defects like in arbitrary Riemannian geometry when the manifold is non self-dual (i.e., not of type Mahalanobis).
First, we explained in \S\ref{sec:doubleright} how to build geodesic $\nabla$-triangles with one, two or three right angles provided that the corresponding system of equations is feasible.
The system of equations is linear up to two right angles but non-linear when dealing with three right angles.
Second, we showed how to build triple of points $(p,q,r)$ such that the dual Pythagorean theorems hold simultaneously  at point $q$ yielding a dually right angle at $q$: two dual pairs of right-angle dual geodesics.
It turned out that the locii of such points $r$ for given points $p$ and $q$ is the intersection of a $\eta$-flat with a $\theta$-flat.
We reported the explicit construction of such triples for the Itakura-Saito manifold in \S\ref{sec:doublyrightangleIS}.  

In future work, we shall consider dually flat spaces for symmetric positive-definite matrices~\cite{BD-Matrix-2013,MIG-2013,IG-Amari2014} where the inner product is the trace of a matrix product.
We would also like to prove the following experimental observation: For the 2D Burg negentropy generator, the total sum $\alpha(T)$ of the interior angles of a geodesic $\nabla$-triangle (geodesic triangle with all primal edges)  plus the total sum $\beta=\alpha(T^*)$ of  the interior angles of a dual geodesic $\nabla^*$-triangle  (geodesic triangle with all dual geodesic edges)  sum up to $2\pi$.
For example, for $\theta(p)=(0.5,0.5)$, $\theta(q)=(0.75,0.75)$ and $\theta(r)=(0.95,0.25)$, we find that
$\alpha(T)=160.19318300825412^{o}$ (angle defect), $\beta=\alpha(T^*)=199.80681699174588^{o}$ (angle excess), and $\alpha+\beta=360.0^{o}$.
This property seems only to hold for the 2D Itakura-Saito divergence and not in higher dimensions.

We shall also consider an extension of this work  to study properties of geodesic convex $k$-gons instead of geodesic triangles (i.e., $3$-gons) in dually flat spaces ($2^k$ such geodesic $k$-gons). 
For example, the Lambert quadrilaterals~\cite{NonEucl-2012} (i.e., $4$-gons) have three  right angles and the remaining angle which is acute in hyperbolic geometry, obtuse in spherical geometry, and a right angle in Euclidean geometry.
In a dually flat space, we have $2^4=16$ types of quadrilaterals defining overall $4 \times 4=16$ interior angles.\footnote{A each quadrilateral vertex, we have $4$ geodesics defining $6$ interior angles between them.}
Figure~\ref{fig:quadrangle} displays all pairs of dual geodesics of some convex quadrilaterals in the $\theta$- and $\eta$-coordinate systems.

\begin{figure}
\centering

\begin{tabular}{|c|c|}\hline
$\theta$-coordinate system & $\eta$-coordinate system\\ \hline
\includegraphics[width=0.3\textwidth]{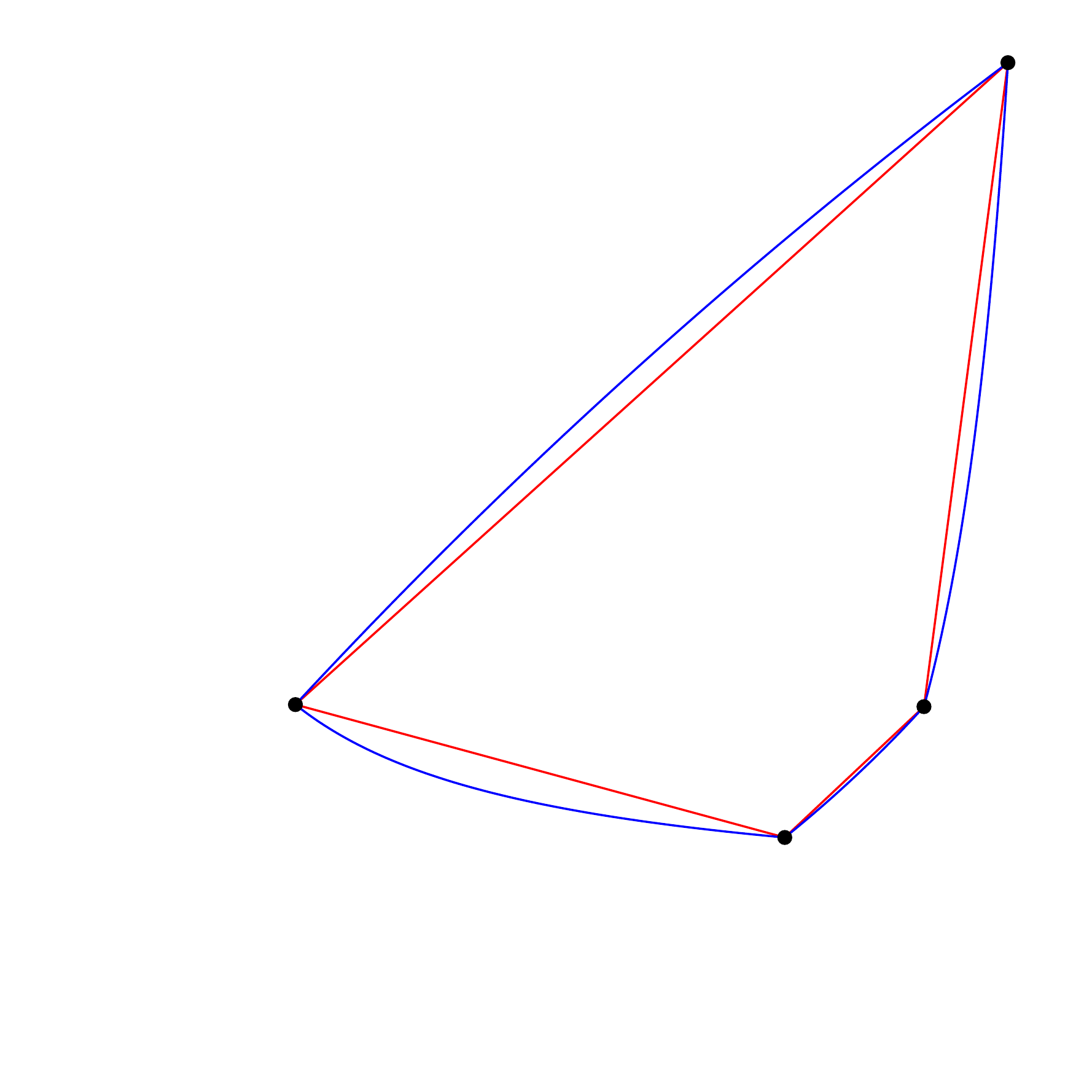} &
\includegraphics[width=0.3\textwidth]{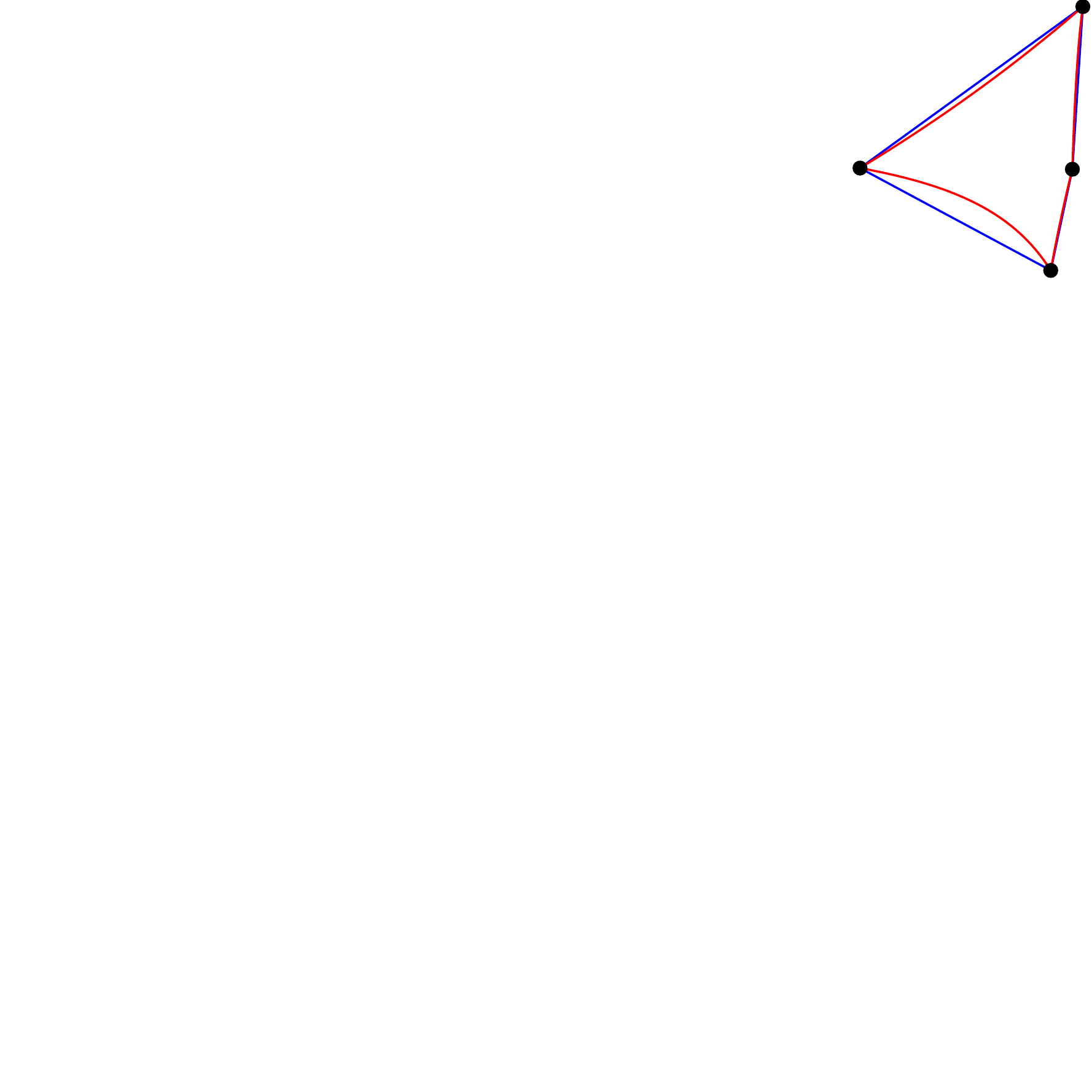} \\ \hline
\includegraphics[width=0.3\textwidth]{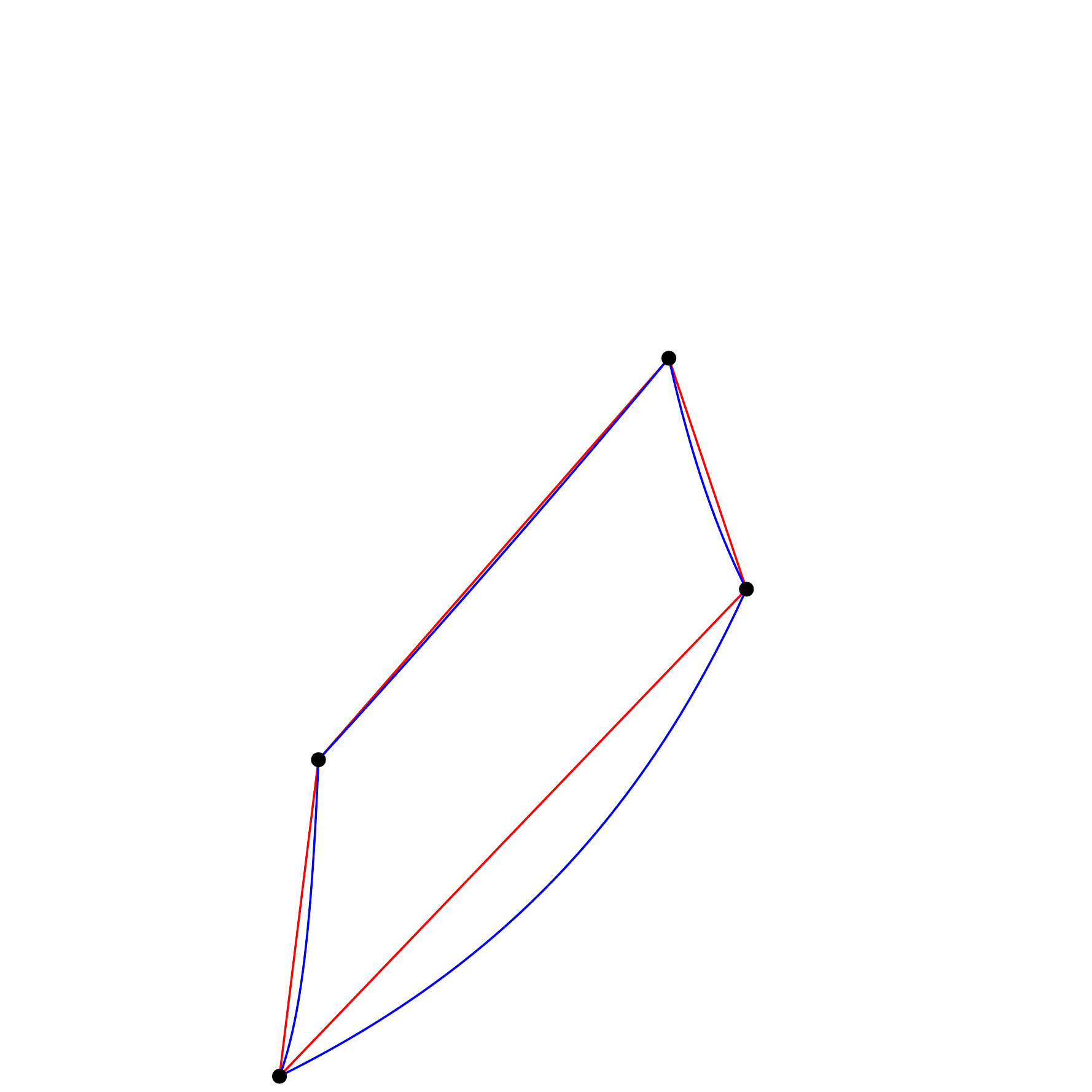} &
\includegraphics[width=0.3\textwidth]{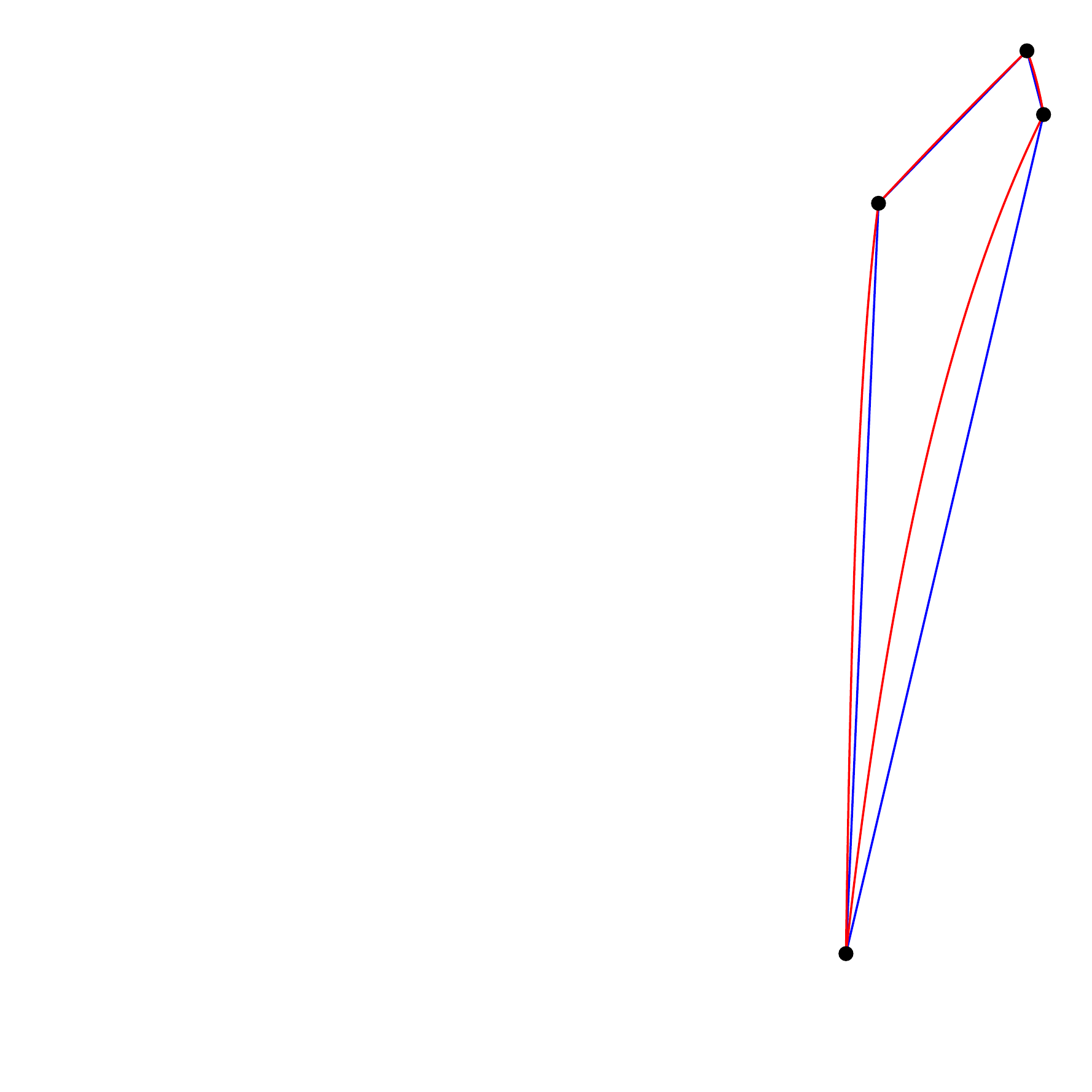} \\ \hline
\end{tabular}

\caption{Examples of convex quadrangles in the Itakura-Saito manifold.
\label{fig:quadrangle}}
\end{figure}

\noindent {Acknowledgments:}  Figures were programmed using \url{processing.org}

\bibliographystyle{plain}
\bibliography{GeodesicTriangleBIB}

\appendix

\section{Notations}\label{sec:notations}

\begin{tabular}{ll}
$F$ & Strictly convex and $C^3$ real-valued function\\
$F^*$ & Dual Legendre-Fenchel convex conjugate\\
$\theta(p)=(\theta^1(p),\ldots,\theta^D(p))$ & primal coordinates of point $p$\\
$\eta(p)=(\eta_1(p),\ldots,\eta_D(p))$ & dual coordinates of point $p$\\
$\theta_{ab}=\theta_a-\theta_b$ & notational shortcut \\
$\eta_{ab}=\eta_a-\eta_b$ & notational shortcut \\
$D_F(p:q)$ & Divergence between points\\
$B_F(\theta(p):\theta(q))$ & Bregman divergence\\
$A_F(\theta(p):\eta(q))$ & Fenchel-Young divergence\\
$(M,g,\nabla,\nabla^*)$ & Dually flat space (Bregman manifold)\\
$T_p$ & Tangent plane at $p$\\
$g_p(u,v)$ & inner product between two vectors $u$ and $v$ of  $T_p$\\
$[g_{ij}]=[g(e_i,e_j)]_{ij}=\nabla^2F(\theta)$ & Riemannian metric\\
$[{g^*}^{ij}]=[g^*({e^*}^i,{e^*}^j)]_{ij}=\nabla^2F^*(\eta)$ & dual Riemannian metric \\
$\prod_{p,q}(v)$ & primal parallel transport of $v\in T_p$ to $T_q$\\
$\prod_{p,q}^*(v)$ & dual parallel transport of $v\in T_p$ to $T_q$\\
%
$\gamma_{ab}(t)$ & Primal geodesic: $\theta(\gamma_{ab}(t))=(1-t)\theta(a)+t\theta(b)$\\
$\gamma_{ab}(t)^*$ & Dual geodesic: $\eta(\gamma_{ab}^*(t))=(1-t)\eta(a)+t\eta(b)$\\
$(v)_B$ & vector components in basis $B$, arranged in a $D$-tuple\\
$[v]_B$ & vector components in basis $B$, arranged in a $D$-dimensional column vector\\
$B_p=\{e_i=\partial_i=\frac{\partial}{\partial\theta^i}\}$ & natural basis at $T_p$   \\
$B_p^*=\{{e^*}^i=\partial^i=\frac{\partial}{\partial\eta_i}\}_i$ &    reciprocal basis  at $T_p$ so that $g(e_i,{e^*}^j)=\delta_{i}^j$\\
$v_{ab}=\frac{d}{\dt}\gamma_{ab}(0)=\dot\gamma_{ab}(0)$ & tangent vector of $\gamma_{ab}(t)$ at $a$ with contravariant components $\theta(b)-\theta(a)$\\
$v_{ab}^*=\frac{d}{\dt}\gamma_{ab}^*(0)=\dot\gamma_{ab}^*(0)$ & tangent vector of $\gamma_{ab}^*(t)$ at $a$ with covariant components $\eta(b)-\eta(a)$\\
$[v^i]_B$ & contravariant components of vector $v$,  $v^i=g(v,{e^*}^i)$\\
$[v_i]_B$ & covariant components of vector $v$ (meaning $[v]_{B^*}$),  $v_i=g(v,e_i)$\\
$g_p(u,v)$ &  inner product at $T_p$ of two vectors: $g_p(u,v)=u_i v^i=u^i v_i$\\
\end{tabular}

\end{document}